\documentclass[11pt, a4paper]{article}
\usepackage[round]{natbib}
\usepackage{latexsym}
\usepackage{enumerate}
\usepackage{epsfig}
\usepackage{amsmath}
\usepackage{amsthm}
\usepackage{amssymb}
\usepackage{amsfonts}
\usepackage[english]{babel}
\usepackage[title]{appendix}
\usepackage{graphicx}
\usepackage{graphics}
\usepackage{psfrag}
\usepackage{amsbsy}
\usepackage{enumerate}
\usepackage{setspace}
\usepackage{float}
\usepackage[linesnumbered,ruled,vlined, norelsize]{algorithm2e}
\usepackage[top = 2.25cm, bottom = 2.25cm, left = 2.15cm, right = 2.15cm, footnotesep=1.5cm]{geometry}

\usepackage[T1]{fontenc}          
\usepackage[latin1]{inputenc}     

\newtheorem{te}{Theorem}[section]

\makeatletter \@addtoreset{equation}{section}

\newcommand{\garchpq}{GARCH($p,q$) }
\newcommand{\garchone}{GARCH($1,1$)}
\providecommand{\norm}[1]{\lVert#1\rVert} \providecommand{\abs}[1]{\lvert#1\rvert}

\title{Evaluation of extremal properties of GARCH(p,q) processes}
\author{Fabrizio Laurini\\ Department of Economics and Management\\
University of Parma \\
Via J.F. Kennedy 6, 43125 Parma, Italy \\
fabrizio.laurini@unipr.it \\
\and Paul Fearnhead \& Jonathan A. Tawn\\
Department of Mathematics and Statistics\\
Lancaster University\\ Lancaster LA1 4YF, UK, \\
p.fearnhead@lancaster.ac.uk, j.tawn@lancaster.ac.uk
}
\date{}

\begin{document}

\maketitle

\today

\abstract{Generalized autoregressive conditionally heteroskedastic  (GARCH) processes are widely used for modelling
features commonly found in observed financial returns. The extremal properties of these processes are of considerable interest for market risk management. For the simplest \garchpq process, with $\max(p,q) = 1$, all extremal features have been fully characterised.  Although the marginal features of extreme values of the process have been theoretically characterised when $\max(p,q)\ge 2$, much remains to be found about both marginal and dependence structure during extreme excursions. Specifically, a reliable method is required for evaluating the tail index, which regulates the marginal tail behaviour and there is a need for methods and algorithms for determining clustering. 
In particular, for the latter, the mean number of extreme values in a short-term cluster, i.e., the reciprocal of the extremal index, has only been characterised in special cases which exclude all \garchpq processes that are used in practice. Although recent research has identified the multivariate regular variation property of stationary \garchpq processes, currently there are no reliable methods for numerically evaluating key components of these characterisations.  We overcome these issues and are able to generate the forward tail chain of the process to derive the extremal index and
a range of other cluster functionals for all \garchpq processes including integrated GARCH processes and processes with unbounded and asymmetric innovations. The new theory and methods we present extend to assessing the strict stationarity and extremal properties for a much broader class of stochastic recurrence equations. }

\vspace{5mm} \noindent \textbf{Keywords:} \textit{Cluster of
extremes; extremal index; fixed point distributions; GARCH process; multivariate regular
variation, particle filtering, stochastic recurrence equations} \vspace{5mm}

\section{Introduction}

Risk management in the stock markets, commonly called {\it market risk management}, requires the use of statistical tools and models which aim at reducing the potential size of losses, occurring by sudden drops or growth in the value of stock. Losses can be amplified during periods of large volatility. Risk managers routinely use strategies to handle, model and predict the volatility of daily log-returns, defined as $X_t = \log{P_t} - \log{P_{t-1}}$, where $P_t$, $t = 1,2,\ldots,$ is the price of a generic asset. The behaviour of the extreme values of daily log-returns is critically important for market risk management. Isolated extreme values of daily log-return can often be managed, but there is major risk when there is a clustering of these extreme values, and so the study of this dependence structure during extreme events is essential.

It is standard to assume that the series $\{X_t\}$ is a stationary series. The most widely adopted models for $\{X_t\}$
are the generalised autoregressive conditionally heteroskedastic (GARCH) models \citep{bol:86} and stochastic volatility (SV) models \citep{ta:86}. These models are capable of capturing many of the empirical features of daily log-returns. Both processes have heavy tailed marginal distributions with the leading decay rates the same for both models. Where they differ is in terms of their extremal dependence structure. One of the most common ways to measure this is through the lag $\tau$ tail dependence
\begin{equation}
\chi_X(\tau)=\lim_{x\rightarrow \infty}\Pr(X_{t+\tau}>x~\mid~X_t>x),
\label{eqn:extremogram}
\end{equation}
proposed by \cite{led+ta:03}, with \cite{da+mi:09b} terming $\{\chi_{\tau}\}_{\tau\ge 0}$ the extremogram.  For SV models \cite{br+da:98} show that there is no clustering of extreme values, so that $\chi_{\tau}=0$
for all $\tau>0$. Thus extreme values from SV processes occur in temporal isolation. In contrast, for any \garchpq process $\chi_{\tau}>0$ for at least one value of $\tau>0$. But, the values of the extremogram, and other extremal dependence features, are only known for a very restricted subclass of \garchpq processes. The aim of this paper is to derive these extremal features for all
\garchpq models used in typical financial applications and to present algorithms for their evaluation. 

 We consider \garchpq models, for $p \in \mathbb{N}$ and  $q\in \mathbb{N}_+$, of the form
\begin{equation} \label{mult.model.intro}
X_t = \sigma_t Z_t
\end{equation}

\noindent where, for every fixed $t\in \mathbb{Z}$, the random variables $Z_t$ and $\sigma_t$ are independent. Furthermore, we assume that $\{Z_t\}$ are independent and identically distributed continuous random variables with $E(Z_t) = 0$ and $\mbox{Var}(Z_t) = 1$. The process $\{\sigma_t\}$, commonly referred to as the conditional volatility of $\{X_t\}$, is given by

\begin{equation} \label{eq:general.garchpq}
\sigma_t^2  = \alpha_0 + \sum_{i=1}^{q} \alpha_i X_{t-i}^2 + \sum_{j=1}^{p} \beta_j \sigma_{t-j}^2, \hspace{5mm} t \in \mathbb{Z},
\end{equation}
\noindent where $\alpha_0 > 0$ and the parameters $\alpha_i \geq 0, i = 1,\ldots, q-1$, $\alpha_q>0$ and $\beta_j \geq 0,
j=1,\ldots,p-1$, $\beta_p>0$ for $p\ge 1$, have to satisfy additional constraints for the process to be strictly stationary; see Section~\ref{sec:assumptions}. When $p=0$,  equations~\eqref{mult.model.intro} and \eqref{eq:general.garchpq} correspond to an ARCH($q$) process and when 
\begin{equation} \label{garchpq.second.order.stationary}
\sum_{i = 1}^q \alpha_i + \sum_{j = 1}^p \beta_j =1
\end{equation}
they correspond to an integrated GARCH($p,q$) process, denoted IGARCH($p,q$), which is strictly stationary but not second-order stationary. We develop new numerically robust and efficient methods for assessing whether any \garchpq process is strictly stationary. Our results cover all of the GARCH($p,q$) processes including these special cases. 

Existing theoretical and computational methods for deriving extremal properties are well established for special cases of the GARCH($p,q$) process, namely: for symmetric $Z_t$ with $p=0,q=1$, corresponding to the ARCH(1) process \citep{deh+r+r+dev:89} and for $p=q=1$, corresponding to a GARCH(1,1) \citep{la+ta:2012}; and for  asymmetric $Z_t$
with $p=q=1$ \citep{eh+fi+ja+sc:2015}. Additional results of other
tails probabilities are derived in the two-dimensional case by
\cite{co+di+vi:2014} but are effective only for ARCH(1) as further
complications arise for the GARCH($1,1$).

For general \garchpq models, with arbitrary $(p,q)$ many theoretical extremal properties have been derived by \cite{ba+da+mi:02b}, \cite{da+mi:09} and \cite{ba+se:09}, including the tails of marginal distributions and some results for the extremal clustering properties.  At first sight it seems that these results give everything that is needed for numerical evaluation of the required extremal properties. But this is far from the case, as we will show.

Firstly consider the marginal tail behaviour of $X_t$ and $X_t^2$. \cite{ba+da+mi:02b} showed that for GARCH($p,q$) processes there is an explicit theoretical expression for $\kappa \ge 0$ such that for fixed $x>1$ as $u\rightarrow \infty$ then
\[
\frac{\Pr(X_t>ux)}{\Pr(X_t>u)}\rightarrow x^{-2\kappa} \mbox{ and } \frac{\Pr(X_t^2>ux)}{\Pr(X_t^2>u)}\rightarrow x^{-\kappa}.
\]
These papers only give an asymptotic limiting expression for the evaluation of $\kappa$, but they do not illustrate its application. We find that direct computation using their  expression gives very poor numerical  performance. \cite{ja:2010} presents an alternative approach to evaluate
$\kappa$, however that approach applies only under the assumption that the innovation $Z_t$  has bounded support, ruling out many important distributions used  by practitioners, e.g., $Z_t$ being Gaussian or $t$-distributed. Furthermore, the associated numerical methods are very slow. We propose the first reliable and computational efficient numerical algorithms for the valuation of $\kappa$, which are valid irrespective of whether $Z_t$ are unbounded or bounded. We also find a new formulation for $\kappa$ which gives new insights and we show that for all IGARCH($p,q$) processes $\kappa=1$.

Now consider the extremal dependence/clustering features of the process.  \cite{ba+se:09} and \cite{ba+seERR:2011} propose computational algorithms for their evaluation. However, these methods have major limitations, which they identify, and only apply to some stochastic recurrence equations with bounded innovations, but do not hold for any GARCH($p,q$) processes, see Section~\ref{sec:first}. So currently no extremal clustering features for GARCH($p,q$) processes, when $\max(p,q)\ge 2$, can be evaluated. We propose an entirely new numerical algorithm to evaluate a range of cluster features for any GARCH($p,q$) process, regardless of the values of $p$ and $q$ and without imposing any restrictive assumptions.
 
There are a range of extremal dependence features of GARCH process that are of interest to practitioners. The most standard features are summarised by the time-normalized point process $N_n$ of exceedances of a level $u_n$,  defined by

\[
N_n(B) = \# \{k/n \in B: X_k >u_n \},
\]

\noindent where $u_n$ tends to the upper endpoint of the distribution of $X_t$ as $n\rightarrow \infty$, such that
$n \bar F_{X} (u_n) \rightarrow \psi$, for any finite $\psi>0$, where $F_{X}$ and $\bar F_{X}$
are the marginal distribution and survivor functions of $\{X_t\}$ respectively. As $n\to\infty$,
$N_n$ converges to a compound Poisson process $N$, where events occur as in a homogeneous Poisson
process with intensity $\psi\theta_X$, where $0<\theta_X\le 1$ is termed the extremal index and with multiplicities  distribution denoted by $\pi_X(k)$ for $k\ge 1$ \citep{hs+hu+lea:88}. We use the term clusters to describe the independent extreme events, with the associated multiplicities corresponding to the number/size of extreme values in each cluster. 
It follows that

\begin{equation} \label{eq.clust.size}
\sum_{i=1}^{\infty} i \pi_{X}(i) = 1/\theta_{X},
\end{equation}

\noindent i.e., the extremal index $\theta_X$ is the reciprocal of the limiting mean cluster size of extreme values. The smaller the extremal index then the larger the average number of extreme values per cluster. The special case $\theta_X = 1$ corresponds to there being no clustering of extremes, so extreme values occur in isolation in time. 
Thus, minimally there is a need to derive $\{\chi_{\tau}\}_{\tau>0}$ and $\theta_X$, and to get the latter we need $\pi_X(\cdot)$. Other functionals are also of interest to financial institutions for managing the duration of a stress period or predicting the total amount of losses that can be faced in such stress period, such as the aggregate of excesses over a cluster (total loss). These can also be derived using our methods. 

All of these cluster functionals can be obtained from the tail chain of the process,  which has been widely used for studying extremal clustering (\cite{roo:88}, \cite{sm+ta+co:97}, \cite{se:03}, \cite{pl+so:2018}). The tail chain is defined for a heavy tailed process $\{X_t\}$ in the following way. When $u\rightarrow  \infty$, if for any $t\in \mathbb{N}_{+}$
\[
(X_0/u, X_1/u, \ldots ,X_t/u )\mid X_0>u,
\]
converges weakly to $(\hat{X}_0, \hat{X}_1, \ldots, \hat{X}_t)$, with $\hat{X}_0$ non-degenerate, 
then the limit process $\{\hat{X}_t\}_{t\ge 0}$ is termed the forward tail chain. For SV models $\hat{X}_t=0$ almost surely for all $t\ge1$, so large values are not followed by large values for SV processes. In contrast, for any GARCH($p,q$) process at least one $\hat{X}_t$, for $t\ge1$, is non-degenerate, and every element of the tail chain is non-degenerate if 
$\min(\alpha_1, \ldots ,\alpha_q)>0$ and $\min(\beta_1, \ldots ,\beta_p)>0$, see Section~\ref{sec:extremogram} for further details. 

Similarly, there is a  backward tail chain $\{\hat{X}_t\}_{t\le 0}$ with identical definition for negative $t$. Here $\hat{X}_0$ is the same value for both forward and backward chains. The connections between forward and backward chains were determined by \cite{ja+se:2014}. However, in many cases only the forward chain is required to  derive the functionals of interest. For example, the extremogram $\{\chi_{X}(\tau); \tau \ge 0\}$ and the extremal index of $\{X_t\}$ can be expressed, respectively,  as
\[
\chi_{X}(\tau)=\Pr(\hat{X}_{\tau}>1 \mid \hat{X}_0>1)\mbox{ and }
\theta_X=\lim_{t\rightarrow \infty} \Pr(\hat{X}_1<1, \ldots ,\hat{X}_t<1\mid \hat{X}_0>1).
\]

In this paper we derive the theory for obtaining the forward tail chain, the extremogram, the extremal index and cluster size distribution for any \garchpq process with bounded or unbounded support for the innovations. We provide a new fast, yet accurate, Monte Carlo algorithm for the numerical evaluation of these  extremal features. We first obtain the
forward tail chain for the squared \garchpq process and then use this in a filtering argument, similar to  \cite{deh+r+r+dev:89}, to derive the features of interest. 

The paper is structured in the following way. In Section~\ref{sec:assumptions} we give the required background details for
the properties of stationary \garchpq processes and the theory of multivariate regular variation that is required for our methodology. In Section~\ref{sec:gamma.kappa} we give new results for testing stationarity and a new formulation for the tail index $\kappa$. In Sections~\ref{sec:tail_chain} and \ref{sec:tailchain} we derive the tail chain for the series
squared and original \garchpq processes respectively, with Section~\ref{sec:initial} containing the key particle filtering algorithm which enables us to sample from a $(p+q)$-dimensional  extreme state of the tail chain. Critical to the development of this algorithm is the theory of fixed point distributions. Section~\ref{sec:numerical} discusses the novel numerical evaluation of all components of the method, including checking for stationarity and evaluating $\kappa$, and it also illustrates the rapid convergence of the particle filter algorithm.
Section~\ref{sec:results} has a study of a range of extremal dependence
features of the \garchpq process over a variety of parameter values. In Section~\ref{sec:discuss} we identify that the methods and algorithms that are developed here in the context of GARCH($p,q$) processes immediately extend to a much broader class of stochastic recurrence equations, and so they are likely to have a much wider impact. 
The proofs of the theorems are given in the Appendix~\ref{sec:app.proofs}.

\section{Known properties of \garchpq processes}
\label{sec:assumptions}

\subsection{Strict Stationarity}
\label{sec:stationarity}

Let us start by defining strict stationarity for \garchpq processes. We focus on the squared GARCH process, $X_t^2$, and rewrite the process as a stochastic recurrence equation (SRE) as this enables the exploitation of a range of established results  \citep{ke:73} for such processes, e.g., the existence of results for the marginal distribution.

Let the $(p+q)$ vector $\mathbf{Y}_t$, the  $(p+q) \times (p+q)$ matrix $\mathbf{A}_t$ and the $(p+q)$ vector $\mathbf{B}_t$ be
\begin{equation} \label{transition.matrix}
\mathbf{Y}_{t} = \begin{pmatrix} X_t^2 \\ \vdots \\ X_{t-q+1}^2\\ \sigma_t^2 \\ \vdots \\ \sigma_{t-p+1}^2 \end{pmatrix}\quad
\mathbf{A}_t = \begin{pmatrix} \alpha^{(q-1)}Z^2_t &\alpha_q Z_t^2&\beta^{(p-1)}Z^2_t &
\beta_p Z_t^2 \\ \boldsymbol{I}_{q-1} &0_{q-1}& \boldsymbol{0}_{(q-1)\times (p-1)} & 0_{q-1} \\ \alpha^{(q-1)} &\alpha_q &\beta^{(p-1)} &
\beta_p \\
\boldsymbol{0}_{(p-1)\times(q-1)} &0_{p-1} & \boldsymbol{I}_{p-1} & 0_{p-1}\end{pmatrix}
\quad \mathbf{B}_t = \begin{pmatrix} \alpha_0 Z^2_t \\0_{q-1} \\\alpha_0 \\0_{q-1} \end{pmatrix}
\end{equation}
where $\alpha^{(s)} = (\alpha_1 , \ldots, \alpha_{s}) \in \mathbb{R}^{s}$,
 $\beta^{(s)}= (\beta_1 , \ldots, \beta_{s})  \in \mathbb{R}^{s}$,
  $\boldsymbol{I}_s$ is the identity matrix of size $s$,   $\boldsymbol{0}_{(r\times s)}$ is a matrix of zeros with $r$ rows and $s$ columns and
  $0_s$ is a column vector of zeros having length $s$. In each case here if $s<0$ then these terms are to be interpreted as being dimensionless. Then it follows that the squared \garchpq processes satisfies the SRE
\begin{equation} \label{eq.sre}
\mathbf{Y}_t = \mathbf{A}_t \mathbf{Y}_{t-1} + \mathbf{B}_t, \hspace{5mm} t \in \mathbb{Z},
\end{equation}

\noindent where $\{\mathbf{A}_t\}$ and $\{\mathbf{B}_t\}$ are each sequences of independent and identically distributed stochastic matrices and vectors.

The formulation of the SRE  via~\eqref{transition.matrix} is due to \citet[Section 2.2.2, page 29]{fr+za:2010}. This SRE formulation  is less parsimonious than that of \cite{bo+pi:92}, but has the benefit of covering all \garchpq processes, unlike
 that of \cite{bo+pi:92} which does not include the case $p=q=1$. In contrast here when $p=q=1$ we have that
 the terms in expression~\eqref{transition.matrix} simplify to
 \begin{equation} \label{transition.matrix.garchone}
\mathbf{Y}_{t} = \begin{pmatrix} X_t^2 \\ \sigma_t^2  \end{pmatrix}\quad \quad \mathbf{A}_t = \begin{pmatrix} \alpha_1 Z_t^2& \beta_1 Z_t^2 \\ \alpha_1 &\beta_1 \end{pmatrix} \quad \quad \mathbf{B}_t = \begin{pmatrix} \alpha_0 Z^2_t \\ \alpha_0 \end{pmatrix},
\end{equation}
\noindent where  $\alpha^{(s)}$ and  $\beta^{(s)}$ are scalar and none of the $\boldsymbol{I}_s, \boldsymbol{0}_{r\times s}$ or $0_s$ are included.

For general SRE of the form~\eqref{eq.sre}, but without the specific specification of $\mathbf{A}_t$  and  $\mathbf{B}_t$
corresponding to a GARCH($p,q$) process, \citet[Theorem 2.4 - page 30/32]{fr+za:2010} show that it is necessary and sufficient that there is a negative top Lyapunov exponent of $\mathbf{A}_t$ for the existence of a unique, strictly stationary solution. Before defining the Lyapunov exponent of a stochastic matrix
first consider the spectral radius of a deterministic square matrix $\mathbf{A}$, denoted $\rho(\mathbf{A})$.
Here $\rho(\mathbf{A})$ is the greatest modulus of its eigenvalues, and an important algebraic result is $\lim_{t \to \infty} t^{-1} \log \norm{\mathbf{A}^t} = \log \rho(\mathbf{A})$, where $\norm{\cdot}$ is any norm on the space of $\mathbf{A}$. The extension to a sequence of strictly stationary and ergodic random matrices $\{\mathbf{A}_t, t \in \mathbb{Z} \}$, for which $E \ln^{+} \norm{\mathbf{A}_t} < \infty$  (here $\ln^{+} x = \ln x,$ if $x\geq 1$ and 0 otherwise), is such that the top Lyapunov exponent is

\begin{equation} \label{eq:lyapunov.def}
\gamma = \lim_{t \to \infty} \frac{1}{t} E \left(\ln \norm{\mathbf{A}_t \mathbf{A}_{t-1} \cdots \mathbf{A}_1}\right),
\end{equation}

\noindent and $\exp(\gamma)$ is the spectral radius of the sequence $\{\mathbf{A}_t, t \in \mathbb{Z} \}$. 
Hence if $E \ln^{+} \norm{\mathbf{A}_t} < \infty$ for all $t$, it is necessary and sufficient that $\gamma<0$ for a strictly stationary process $\mathbf{Y}_t$.

So for strict stationarity of the squared and original GARCH($p,q$) processes $X_t^2$ and $X_t$ we need $E \ln^{+} \norm{\mathbf{A}_t} < \infty$ for $\mathbf{A}_t$ given by expression~\eqref{transition.matrix} and $\gamma<0$.
The finite moment condition holds for GARCH($p,q$) processes as $\norm{\mathbf{A}_t}<CZ^2_t+D$, for constants $C>0$ and $D>0$ and so $E \ln^{+} \norm{\mathbf{A}_t}<E \ln^{+}(CZ^2_t+D) =E(\ln(CZ^2_t+D)\mid CZ^2_t+D>1) \Pr(CZ^2_t+D>1)<KE(\ln(Z_t)\mid CZ^2_t+D>1)<\infty$, for a suitable constant $K$ and where
the last inequality holds as $E(Z_t)$ is finite by the model definition~\eqref{mult.model.intro}.
Consequently we only need to check if $\gamma<0$.

Unfortunately, expression~\eqref{eq:lyapunov.def} is not an ideal starting point for evaluating $\gamma$. Instead we also have
\begin{equation} \label{eq:lyapunov.monte.carlo}
\gamma_t = \frac{1}{t}\ln \norm{\mathbf{A}_t \mathbf{A}_{t-1} \cdots  \mathbf{A}_1}\quad \text{and}\quad \gamma = \lim_{t \to \infty}\gamma_t,
\end{equation}
\noindent see \citet[Theorem 2.3 - page 30]{fr+za:2010}. So, via expression~\eqref{eq:lyapunov.monte.carlo}, it would appear a relatively simple simulation can be performed to obtain  a reliable Monte Carlo estimate of $\gamma$. However, as we will show in Section~\ref{sec:gamma}, this is far from the case and a mix of careful asymptotic approximation analysis and numerical evaluation is required to evaluate $\gamma$.

In some special cases we do not need to evaluate $\gamma$ to find if the process is strictly stationary, e.g., GARCH($p,q$) processes are always strictly stationary when $\sum_{i = 1}^q \alpha_i +\sum_{j = 1}^p \beta_j \le 1$; this includes all IGARCH($p,q$) processes. It is also known that $\sum_{j = 1}^p \beta_j < 1$ is necessary but not sufficient for strict stationarity. Therefore, numerical evaluation of $\gamma$ is required whenever $\sum_{i = 1}^q \alpha_i+ \sum_{j = 1}^p \beta_j>1$ when $\sum_{j = 1}^p \beta_j < 1$.

\subsection{Existing Results}
\label{sec:first}

\cite{ba+da+mi:02b} show that there exists a unique stationary solution to the
SRE~\eqref{eq.sre} and that this solution exhibits a multivariate regular variation property, i.e., for any $t$, any norm
$\norm{\cdot}$ and all $r>0$,
\begin{equation} \label{eq.mult.var.garchpq}
\frac{\Pr(\norm{\mathbf{Y}_t} > rx, \mathbf{Y}_t/\norm{{\mathbf{Y}_t}} \in \cdot)}{\Pr(\norm{\mathbf{Y}_t} > x)} \stackrel{v}{\to}
r^{-\kappa} \Pr(\hat{\mathbf{\Theta}}_t \in \cdot), \hspace{5mm} \mbox{ as } x \to \infty,
\end{equation}

\noindent where $\stackrel{v}{\to}$ denotes vague convergence \citep{ka:83},
$\kappa \geq 0$, and $\hat{\mathbf{\Theta}}_t$ is a $(p+q)$-dimensional random vector in the unit sphere (with respect to a norm $\norm{\cdot}$) defined by $\mathbb{S}^{p+q}\subset \mathbb{R}^{p+q}$, and their $(p+q)$ elements will be denoted by $\hat{\mathbf{\Theta}}_t = (\hat{\vartheta}_t^{(1)}, \ldots, \hat{\vartheta}_t^{(p+q)})$. If condition~\eqref{eq.mult.var.garchpq} holds then $\mathbf{Y}_t$ is said to exhibit multivariate regular variation with index $\kappa$ and the distribution of $\hat{\mathbf{\Theta}}_t$ is  termed the spectral measure of the vector $\mathbf{Y}_t$. See \cite{re:87} for further details on multivariate regular variation. A consequence of the multivariate regular variation property~\eqref{eq.mult.var.garchpq} for GARCH($p,q$) processes is that all the marginal variables of  $\mathbf{Y}_t$ have regularly varying tails with index $\kappa$, so in particular for $r\ge 1$ and all $t$
\begin{equation} \label{eq.uni.var.garchpq}
\Pr(X_t^2> rx \mid X_t^2> x) \to
r^{-\kappa}, \mbox{ and } \Pr(\sigma_t^2> rx \mid \sigma_t^2> x) \to
r^{-\kappa},
\hspace{5mm} \mbox{ as } x \to \infty.
\end{equation}
So both the squared \garchpq process and its variance have regularly varying tails of index $\kappa$.

It is insightful to consider a slightly rearranged version of limit~\eqref{eq.mult.var.garchpq} and to be specific about which norm we will use. We take the $L_1$ norm, and define radial, $R_t$, and angular (two variants $\mathbf{\Theta}_t$ and $\mathbf{\Theta}^-_t$) random variables by
\begin{eqnarray} \label{eq.setting.vectors}
\nonumber R_t &=& \norm{\mathbf{Y}_t}= X_t^2+ \ldots +X_{t-q+1}^2+ \sigma^2_t + \ldots +\sigma^2_{t-p+1},  \\
\mathbf{\Theta}_t &=& \mathbf{Y}_t/\norm{{\mathbf{Y}_t}}= (X_t^2, \ldots ,X_{t-q+1}^2, \sigma_t^2, \ldots ,\sigma_{t-p}^2,\sigma_{t-p+1}^2)/R_t \nonumber\\
\mathbf{\Theta}^-_t &=& (X_t^2, \ldots ,X_{t-q+1}^2, \sigma_t^2, \ldots ,\sigma_{t-p}^2)/R_t,
\end{eqnarray}
with $\mathbb{S}^{p+q}$ the $(p+q)$ dimensional unit simplex. We have two angular variables as the $p+q$ dimensional variable $\mathbf{\Theta}_t$ has redundancy in its final dimension as  $\norm{\mathbf{\Theta}_t}=1$, and so for studying  the distribution of angular variables it is simpler to work with the $p+q-1$ dimensional variable $\mathbf{\Theta}^-_t$, which is related to  $\mathbf{\Theta}_t$ by 
$\mathbf{\Theta}_t=(\mathbf{\Theta}^-_t, 1-\norm{\mathbf{\Theta}^-_t})$ and $\mathbf{\Theta}^-_t$ being $\mathbf{\Theta}_t$ without its last component. We use this $\mathbf{W}^-$ notation to create a $(p+q-1)$ dimensional vector from any $(p+q)$ dimensional vector $\mathbf{W}$ on the simplex $\mathbb{S}^{p+q}$ throughout. Furthermore,  for $\mathbf{w}\in \mathbb{R}^{p+q-1}$, we use the
notation 
\begin{equation}
H_{\mathbf{\Theta}_t }(\mathbf{w})=\Pr(\mathbf{\Theta}^-_t\le \mathbf{w}),
\label{eq:dist}
\end{equation}
with vector algebra, here and elsewhere, interpreted as being componentwise. 

We will denote the limit random variables, that arise in limit~\eqref{eq.mult.var.garchpq}, for $(R_t, \mathbf{\Theta}_t, \mathbf{\Theta}^-_t)$ by $(\hat{R}_t ,\hat{\mathbf{\Theta}}_t,\hat{\mathbf{\Theta}}^-_t)$, with the distribution function of $\hat{\mathbf{\Theta}}_t$,  denoted by $H_{\hat{\mathbf{\Theta}}_t}$, defined similarly to distribution~\eqref{eq:dist}. Subsequently $H_{\hat{\mathbf{\Theta}}_t}$ is referred to as the spectral measure, the term coming from multivariate regular variation terminology \citep{re:87}. Then, for $r \ge1$, as $x \to \infty$, the limit~\eqref{eq.mult.var.garchpq}  becomes, 
\begin{equation}\label{eqn:spectral_limit}
\Pr(R_t > rx, \mathbf{\Theta}^-_t \le \mathbf{w} \mid R_t> x) \stackrel{v}{\to}
\Pr(\hat{R}_t>r)\Pr(\hat{\mathbf{\Theta}}^-_t \le \mathbf{w}) = r^{-\kappa} H_{\hat{\mathbf{\Theta}}_t}(\mathbf{w}). 
\end{equation}

\noindent 
 From the first expression for the asymptotic form in limit~\eqref{eqn:spectral_limit} we see that the radial variable $R_t$ and the angular variables $\mathbf{\Theta}_t$ become asymptotically independent, as  the radial variable $R_t$ grows due to $x \to \infty$, i.e., the variables 
$ \hat{R}_t$ and  $\hat{\mathbf{\Theta}}_t$ are independent. The second term in this limit shows that $\hat{R}_t$ is a Pareto random variable with tail index $\kappa$, i.e.,
\begin{equation} \label{eqn:Rmarginal}
\Pr(\hat{R}_t>r)=r^{-\kappa} \hspace{5mm} \mbox{ for } r\ge 1.
\end{equation}

\noindent There is additional structure imposed on both $H_{\hat{\mathbf{\Theta}}_t}$ and $\kappa$ by the \garchpq process, which  
is identified by \cite{ba+se:09} and \cite{ja:2010} respectively. We discuss this structure in each case below.
 
For the GARCH($1,1$) process, \cite{la+ta:2012} provided an expression for the spectral measure, for a different description of the angular variable to that used here. For the choice of  the angular variable~\eqref{eq.setting.vectors}, for all $t$, their result translates to 
\begin{equation}
H_{\hat{\mathbf{\Theta}}_t}(w)= \frac{2}{E(\abs{Z}^{2\kappa})} \int_{0}^{\left(\frac{w}{1-w}\right)^{1/2}}(1+s^2)^{\kappa} F_Z(ds),\hspace{1cm} \mbox{ for } 0\le w\le 1,
\label{eqn:HforGARCH11}
\end{equation}
where $F_Z$ is the distribution function of the innovations $Z_t$. When $\max(p,q)\geq2$, through highly skilled use of the multivariate regular variation structure,  \citet[Propositions 3.3, 5.1]{ba+se:09} show, that when  $\mathbf{A}$ is independent and identically distributed to $\mathbf{A}_t$, that
\begin{equation}
E(\norm{\mathbf{A}\hat{\mathbf{\Theta}}_t}^{\kappa}) = 1
\label{eqn:moment}
\end{equation}
\noindent and uniquely
\[
\Pr(\hat{\mathbf{\Theta}}_t \in \cdot)=E(\norm{\mathbf{A}\hat{\mathbf{\Theta}}_t}^{\kappa}; \mathbf{A}\hat{\mathbf{\Theta}}_t/\norm{\mathbf{A}\hat{\mathbf{\Theta}}_t} \in \cdot)
\]
\noindent where the notation $E(X;Y):=E(X \mathbf{1}_Y)$ where $\mathbf{1}_Y$ is the indicator of the event $Y$. Thus
\begin{equation} \label{eqn:Hform}
H_{\hat{\mathbf{\Theta}}_t}(\mathbf{w})=E(\norm{\mathbf{A}\hat{\mathbf{\Theta}}_t}^{\kappa}; (\mathbf{A}\hat{\mathbf{\Theta}}_t/\norm{\mathbf{A}\hat{\mathbf{\Theta}}_t})^- \le \mathbf{w}).
\end{equation}

\noindent \citet[p. 1075]{ba+se:09}  propose an approach to simulate from $H_{\hat{\mathbf{\Theta}}_t}(\mathbf{w})$
for an SRE of the form~\eqref{eq.sre}, with the required distribution $H_{\hat{\mathbf{\Theta}}_t}$ being the invariant distribution of a Markov chain and 
MCMC methods used for its evaluation. However, this method cannot be used for general GARCH($p,q$) processes for the following reasons. Firstly, they make an assumption that $\mathbf{A}_t$ is bounded, which excludes the possibility of $Z_t$ being, for example, Gaussian or $t_{\nu}$ distributed. Much more critically though, \cite{ba+seERR:2011} note that the proof that $H_{\hat{\mathbf{\Theta}}_t}(\mathbf{w})$ is the stationary distribution of the Markov Chain that they proposed was flawed, and their claimed results only hold under one of two very specific conditions on the matrix $\mathbf{A}_t$ in the SRE framework~\eqref{eq.sre}.
Neither of these conditions are satisfied by the form of $\mathbf{A}_t$ for GARCH$(p,q)$ processes, when $\max(p,q)\geq2$, even with a bounded $Z_t$. Thus  the algorithm proposed by \cite{ba+se:09} cannot be used for obtaining $H_{\hat{\mathbf{\Theta}}_t}(\mathbf{w})$ for a GARCH$(p,q)$ process. Our approach in Section~\ref{sec:initial} overcomes both of these restrictions. 

Next we focus on how $\kappa$ is determined. In particular, \cite{ba+da+mi:02b}
showed that there exist a $\kappa>0$ which is the unique positive solution of the equation
\begin{equation} \label{eq:root.equation}
\lim_{t \to \infty} \frac{1}{t} \ln E \left(\norm{\mathbf{A}_t \mathbf{A}_{t-1} \cdots \mathbf{A}_1}^{\kappa}\right) = 0
\end{equation}
where here and throughout the matrix norm we use will use is $\norm{\mathbf{A}} = \sum \abs{a_{ij}}$, where $a_{ij}$
is $(i,j)$th element of matrix $\mathbf{A}$.  For the GARCH($1,1$)  process  \cite{mi+st:00} show that $\kappa$ is simple to evaluate using expression~\eqref{eq:root.equation}.  Specifically, taking $\mathbf{A}_t$ as in expression~\eqref{transition.matrix.garchone} we have that 

\[
\mathbf{A}_t \mathbf{A}_{t-1} \cdots  \mathbf{A}_1 = \mathbf{A}_t \prod_{i=1}^{t-1}(\alpha_1 Z^2_i + \beta_1),
\]

\noindent from which it simply follows that expression~\eqref{eq:root.equation} holds and $\kappa$ satisfies

\begin{equation}
E\left[\left(\alpha_1 Z^2_t + \beta_1\right)^\kappa\right] = 1.
\label{eqn:GARCH11}
\end{equation}

\noindent Setting $\beta_1 = 0$ for the \garchone\, process gives the same result for $\kappa$ derived by \cite{deh+r+r+dev:89} for the ARCH(1) process. For general \garchpq processes no such existing simplification of  equation~\eqref{eq:root.equation} gives an easier expression for $\kappa$. So it is natural to try to find $\kappa$ by a numerical solution of the limit equation~\eqref{eq:root.equation}. However, direct numerical solution is non-trivial due to numerical instabilities.

The only existing feasible method to evaluate $\kappa$ was proposed by \citet[Proposition 4.3.1]{ja:2010}, which  exploits \citet[Proof of Theorem 3]{ke:73}. With $\mathbf{A}$ specified in~\eqref{eq.sre} the conditions required for the results of \cite{ke:73} apply and the equality 
\begin{equation} \label{eq:rhohforkappa}
\int_{\mathbb{S}^{p+q}} E \left[\norm{\mathbf{A} \mathbf{w}}^{k} g\left( \frac{\mathbf{A} \mathbf{w}}{\norm{\mathbf{A} \mathbf{w}}}\right)\right] 
H_k(d\mathbf{w})
= \rho_k \int_{\mathbb{S}^{p+q}} g\left(\mathbf{w}\right) 
H_k(d\mathbf{w}),
\end{equation}

\noindent holds for all continuous functions $g$, all unit measures $H_k$ on the space $\mathbb{S}^{p+q}$, and where $\rho_k$ is a constant. The special case of $g\equiv 1$ in \eqref{eq:rhohforkappa} gives the simplification 

\begin{equation} \label{eq:sampling.togetkappa}
\int_{\mathbb{S}^{p+q}} E \left[\norm{\mathbf{A} \mathbf{w}}^{k}\right] H_k(d\mathbf{w}) = \rho_k \int_{\mathbb{S}^{p+q}}  H_k(d\mathbf{w}) = \rho_k.
\end{equation}

\noindent  For any given $k\in (0, \infty)$, whatever the chosen unit measure $H_k$,  if the pair $(\rho_k, H_k)$ satisfies equality~\eqref{eq:rhohforkappa}, then $\rho_k$ is determined solely by $k$, i.e., not by the choice of $H_k$. 

As we know from condition~\eqref{eqn:moment} that the $\kappa$ moment of $\norm{\mathbf{A}\hat{\mathbf{\Theta}}_t}$ is equal to $1$, where $\hat{\mathbf{\Theta}}_t \sim H_{\hat{\mathbf{\Theta}}_t}$  on the space $\mathbb{S}^{p+q}$ then it follows from property~\eqref{eq:sampling.togetkappa} that $\rho_k = 1$ when $k = \kappa$. \cite{ke:73} and \cite{ja:2010} shows that there is only one solution to equation~\eqref{eqn:moment}, so $\kappa$ is the unique solution of   
\begin{equation} \label{eq:sampling.togetkappa=1}
\int_{\mathbb{S}^{p+q}} E \left[\norm{\mathbf{A} \mathbf{w}}^{\kappa}\right] H_{\kappa}(d\mathbf{w}) = 1,
\end{equation}
and that the unit measure $H_{\kappa}$ must correspond to the distribution function $H_{\hat{\mathbf{\Theta}}_t}$. Thus if we can find, or simulate from, $H_{\hat{\mathbf{\Theta}}_t}$ we can find $\kappa$.
 To evaluate $\kappa$ all that is required is to define a class of unit measures $H_k$, over $k\in (0, \infty)$, which contains within it as an interior point $H_{\hat{\mathbf{\Theta}}_t}$, and then vary $k$ until property~\eqref{eq:sampling.togetkappa=1} is found.

\cite{ja:2010} proposes an algorithm to simulate from a class of functions $H_k$ which adapts the invalid algorithm of \cite{ba+se:09}. This gives a valid method for finding  $\kappa$ and $H_{\hat{\mathbf{\Theta}}_t}$ but critically it only applies when the innovations $Z_t$ have bounded support. Therefore there is a need for an algorithm to cover cases where $Z_t$ have unbounded support. In Section~\ref{sec:initial} we describe such an algorithm that applies whatever the support of $Z_t$. Furthermore, our methods for calculating $\kappa$, have substantial computational performance and efficiency gains compared to the algorithm of \cite{ja:2010};  see Section~\ref{sec:kappaEval}.

\section{Extremal properties of squared GARCH processes}
\label{sec:tailchainsquared}

\subsection{New formulations for $\gamma$ and $\kappa$}\label{sec:gamma.kappa}

Expression~\eqref{eq:lyapunov.monte.carlo} suggests using Monte Carlo for the evaluation of $\gamma$, by taking $t$ to be very large.  This approach suffers from serious numerical instabilities. As the norm of the product is tending to zero it seems sensible to first normalise the
size of the individual terms in the product. We do this be scaling each matrix $\mathbf{A}_i$ by its largest (in magnitude) eigenvalue, which we denote by  $\lambda_i$, with  $\lambda_i>0$ for all $i$. From \cite{ke+sp:84} for an arbitrary $\mathbf{A}_i$, satisfying conditions of Section~\ref{sec:stationarity}, it is guaranteed that $\lambda_i$ is simple and exceeds all other eigenvalues in absolute value.  Let
\[
\Delta_t = \prod_{i=1}^{t} \left(\frac{\mathbf{A}_{t+1-i}}{\lambda_{t+1-i}}\right)
\]
so that the product~\eqref{eq:lyapunov.monte.carlo} can be re-written as
\begin{equation}
\gamma_t=\frac{1}{t}\ln \norm{\Delta_t}+\frac{1}{t}\sum_{i=1}^t \ln \lambda_i,
\label{eq:stable}
\end{equation}

\noindent and let 
\begin{equation}\label{eq:defin.eta}
\eta_t= \frac{1}{t}\ln \norm{\Delta_t} \quad \text{and} \quad \eta= \lim_{t\rightarrow \infty} \eta_t.
\end{equation}
\noindent We study the limit behaviour of both components in equation~\eqref{eq:stable}  in Theorem~\ref{th.gamma}, whose proof is postponed to Appendix~\ref{sec:app.proofs}.

\begin{te} \label{th.gamma}
	If $\mathbf{A}_t, t\in \mathbb{Z}$,  is  a sequence of independent and identically distributed random matrices, with non-negative entries, and
	
	\[
	\mathbf{C}_t := \prod_{i=1}^t \frac{\mathbf{A}_{t+1-i}}{(\lambda_{t+1-i}\exp(\eta))},
	\]
	where $\lambda_i$ is the magnitude of the largest eigenvalue of $\mathbf{A}_i$ and $\eta \in \mathbb{R}$ is such  that
	\begin{equation}
	\ln \norm{\Delta_t}/t  \rightarrow \eta, \mbox{ or equivalently }   \ln \norm{\mathbf{C}_t}/t\rightarrow 0, \mbox{ almost surely as } t\rightarrow \infty,
	\label{eqn:Ccondition}
	\end{equation}
	then $ \gamma = \lim_{t \to \infty}\gamma_t$ if and only if
	\[
	\gamma=E(\ln\lambda)+\eta.
	\]
\end{te}

To assess this result in terms of what is already known we first compare with the GARCH($1,1$) process. In that case
\cite{mi+st:00} show that $\gamma=E[\ln(\alpha_1Z_t^2+\beta_1)]$, but the only, and hence largest in magnitude, eigenvalue of $\mathbf{A}_t$ is $\alpha_1Z_t^2+\beta_1$,
thus Theorem~\ref{th.gamma} is identical to their result when $\eta=0$. To show that $\eta=0$ for all GARCH($1,1$) processes note that
\begin{eqnarray*}
	\prod_{i=1}^{t}\mathbf{A}_{t+1-i}  & = & \mathbf{A}_t \prod_{i=1}^{t-1} (\alpha_1Z^2_i+\beta_1)\\
	& = & \mathbf{A}_t \prod_{i=1}^{t-1} \lambda_i
\end{eqnarray*}
\noindent so
\begin{eqnarray*}
	\frac{1}{t}\ln \norm{\prod_{i=1}^{t}\mathbf{A}_{t+1-i}}  & = & \frac{1}{t} \ln \norm{\mathbf{A}_t} +\frac{1}{t}\sum_{i=1}^{t-1} \ln \lambda_i\\
	& \rightarrow & 0+ E(\ln \lambda).
\end{eqnarray*}
\noindent Hence $\eta=0$ from Theorem~\ref{th.gamma}. 

The practical evaluation of the tail index $\kappa$ is not discussed by \cite{ba+se:09} or subsequent authors.
Using representation~\eqref{eq:root.equation} as the basis for numerical evaluation of $\kappa$ for a GARCH($p,q$) process turns out to be trivial only for very ARCH(1) and GARCH($1,1$) processes. Monte Carlo is essential, but solving the limiting equation~\eqref{eq:root.equation} is non-trivial when $\max(p,q)\ge 2$ due to major numerical instabilities, In Theorem~\ref{th.kappa} we present a new representation which provides both insight into which factors determine $\kappa$ as well as a basis for a method of evaluation with greater numerical stability.

\begin{te} \label{th.kappa}
	Under the same notation and conditions of Theorem~\ref{th.gamma} and that $\gamma<0$, the unique solution 	
	$\kappa>0$, of the limiting equation~\eqref{eq:root.equation}, satisfies 
	\[
	E \left([\lambda\exp(\eta)]^{\kappa}\right) = 1.
	\]
	Furthermore, it follows that 
	\begin{equation}
	\label{eqn:niceeta}
	\eta=-\frac{1}{\kappa}\ln (E \left(\lambda^{\kappa}\right)).
	\end{equation}

\end{te}
For all strictly stationary GARCH(1,1) processes, \cite{mi+st:00} show that $\kappa$ must satisfy $E[(\alpha_1Z_t^2+\beta_1)^\kappa]=1$, but this is simply $E[\lambda^\kappa]=1$, so Theorem~\ref{th.kappa} gives that $\eta=0$ for all GARCH($1,1$) processes, as shown directly above. 

The results of Theorem~\ref{th.kappa} are particularly powerful as they allow the simple evaluation of $\kappa$, if $\eta$ is known, or vice-versa. Theorem~\ref{th.gamma} gives a limiting expression from which to approximate $\eta$ and hence $\kappa$ can be found approximately, but better still, Section~\ref{sec:initial} gives a reliable numerical method to calculate $\kappa$ and then Theorem~\ref{th.kappa} provides the ideal way to find $\eta$. These approaches are illustrated in Section~\ref{sec:numerical}. Finally as an immediate consequence of Theorems~\ref{th.gamma} and \ref{th.kappa} if we know the process is stationary, and we know $\kappa$, we can directly calculate $\gamma$ using the following result.
\begin{te} \label{th.combined}
	Under the same notation and conditions of Theorem~\ref{th.gamma},  if $\gamma<0$ and limiting equation~\eqref{eqn:moment} gives 	$\kappa>0$, then	
	\[
	\gamma=E(\ln\lambda)-\frac{1}{\kappa}\ln (E \left(\lambda^{\kappa}\right)).	
	\]
\end{te}
Unlike Theorem~\ref{th.gamma}, we cannot use Theorem~\ref{th.combined} to test for stationarity of the process as this result only provides an expression for $\gamma$ given that $\gamma<0$, i.e., stationarity needs to be confirmed prior to its use. The only comparable existing result to Theorem~\ref{th.combined} is by \cite{ke+sp:84} where it is shown that 
$\gamma\leq \ln(E(\lambda))$ for general random matrix Markov processes. Finally, Theorem~\ref{th.IGARCHkappa} shows  that for all IGARCH($p,q$) processes $\kappa=1$, in which case  expression~\eqref{eqn:niceeta} gives that $\eta=-\ln (E (\lambda))$ and  Theorem~\ref{th.combined} gives that
$\gamma=E(\ln\lambda)-\ln (E \left(\lambda\right))$.
\begin{te} \label{th.IGARCHkappa}
	For all  IGARCH($p,q$) with $\gamma<0$, we have that $\kappa=1$. Furthermore, if we have a stationary GARCH($p,q$) process with $\kappa=1$ then it must be an IGARCH($p,q$) process.
	\end{te}

This seems to be the first time that it has been claimed that any IGARCH($p,q$) process with $\max(p,q)\ge 2$ has $\kappa=1$, although it was proved for $p=q=1$  by \cite{mi+st:00}. The finite mean and infinite variance of the IGARCH($p,q$) process implies that $0.5< \kappa \le 1$. So our result gives much more, e.g., 
all IGARCH($p,q$) processes  have $E(|X_t| ^{2-\epsilon}) < \infty$ for any $\epsilon \in (0, 2]$.  This finding for $\kappa$ is not too surprising though as the 
variance of $X_t$ is infinite when 
$\sum_{i = 1}^q \alpha_i + \sum_{j = 1}^p \beta_j=1$ but was finite when this sum is less than $1$, suggesting $\kappa$ for the IGARCH($p,q$) was a critical boundary point for having a finite variance.

\subsection{Evaluating the Spectral Measure and the Tail Index}
\label{sec:initial}

This section gives the details of our algorithm for sampling from the limit distribution $H_{\hat{\mathbf{\Theta}}_t}$ and then uses this algorithm repeatedly to find $\kappa$. The algorithm requires no assumptions on the support for $Z_t$. 
Throughout we take $t=0$, both to help simplify notation here and as it will be from time $t=0$ that we start the tail chains in Section~\ref{sec:tail_chain}. We will first assume that $\kappa$ is known and present Algorithm~\ref{alg:Paul} for generating from $H_{\hat{\mathbf{\Theta}}_0}$ and then discuss the case when $\kappa$ is unknown.

To simulate from the spectral measure $H_{\hat{\mathbf{\Theta}}_0}$, defined via (\ref{eqn:Hform}), our approach is to introduce a stochastic process whose invariant distribution is $H_{\hat{\mathbf{\Theta}}_0}$. We will then use sequential importance sampling \cite[]{doucet:2000} to generate approximate samples from the state of this stochastic process at consecutive time-steps. We denote the stochastic process by $\widetilde{\mathbf{\Theta}}_s$, for $s=0,1,\ldots,$ and denote its joint distribution function at iteration $s$ by 
\begin{equation}
\widetilde{\mathbf{\Theta}}_s\sim H^{(s)}_{\widetilde{\mathbf{\Theta}}}. 
\label{eqn:PaulchainDist}
\end{equation}
This stochastic process is constructed such that as $s\to \infty$, $\widetilde{\mathbf{\Theta}}_s\stackrel{d}{\to} \hat{\mathbf{\Theta}}_0\sim H_{\hat{\mathbf{\Theta}}_0}$. We perform sequential updates until it appears that the distribution of the state of the stochastic process has converged to the invariant distribution of the process, $H_{\hat{\mathbf{\Theta}}_0}(\mathbf{w})$. The samples at this final iteration are then taken as samples from $H_{\hat{\mathbf{\Theta}}_0}$. This algorithm is similar to the use of sequential Monte Carlo for sampling from Feynman-Kac distribution \cite[e.g.][]{del2000branching} and fixed point distributions \cite[e.g.][]{delMoral:2003}.

Let $\widetilde{\mathbf{\Theta}}_s$, for $s=0,1,\ldots,$ be a Markov process, with initial state an arbitrary distribution on $\mathbb{S}^{p+q}$ and whose transitions for $s\geq 1$ are given by
\begin{equation} \label{eqn:MP}
 \Pr(\widetilde{\mathbf{\Theta}}_s \in \cdot)= \frac{ E(\norm{\mathbf{A}\widetilde{\mathbf{\Theta}}_{s-1}}^{\kappa}; \mathbf{A}\widetilde{\mathbf{\Theta}}_{s-1}/\norm{\mathbf{A}\widetilde{\mathbf{\Theta}}_{s-1}} \in \cdot)}{ E(\norm{\mathbf{A}\widetilde{\mathbf{\Theta}}_{s-1}}^{\kappa})},
\end{equation}
where, as above, $E(X;Y):=E(X \mathbf{1}_Y)$. By construction, the  invariant distribution of this process is $H_{\hat{\mathbf{\Theta}}_0}$. To see this notice that if $\widetilde{\mathbf{\Theta}}_{s-1}$
is drawn from $H_{\hat{\mathbf{\Theta}}_0}$, then the right-hand side of (\ref{eqn:MP}) is equal to
\[
 \frac{ E(\norm{\mathbf{A}\hat{\mathbf{\Theta}}_0}^{\kappa}; \mathbf{A}\hat{\mathbf{\Theta}}_0/ \norm{\mathbf{A}\hat{\mathbf{\Theta}}_0} \in \cdot)}{ E(\norm{\mathbf{A}\hat{\mathbf{\Theta}}_0}^{\kappa})}.
\]
As $E(\norm{\mathbf{A}\hat{\mathbf{\Theta}}_0}^{\kappa}) = 1$, this is equal to the definition of $H_{\hat{\mathbf{\Theta}}_0}$ given by expression~\eqref{eqn:Hform}.

Furthermore, if we have a sample from $\Pr(\widetilde{\mathbf{\Theta}}_{s-1} \in \cdot)$, we can use importance sampling to generate a sample from $\Pr(\widetilde{\mathbf{\Theta}}_{s} \in \cdot)$. This would involve first simulating a value for $\widetilde{\mathbf{\Theta}}_{s}$ via
\[
 \widetilde{\mathbf{\Theta}}_{s}= \mathbf{A}\widetilde{\mathbf{\Theta}}_{s-1}/\norm{\mathbf{A}\widetilde{\mathbf{\Theta}}_{s-1}},
\]
and assigning this value an importance sampling weight proportional to $\norm{ \mathbf{A} \widetilde{ \mathbf{\Theta}}_{s-1}}^\kappa$. Thus we can use sequential importance sampling to generate samples of $\widetilde{\mathbf{\Theta}}_{s}$ values for $s\ge 1$. Specifically, we implement Algorithm~\ref{alg:Paul}, with our choice of initial distribution, in step 1, being chosen to be close to $H_{\hat{\mathbf{\Theta}}_0}(\mathbf{w})$, so as to speed up convergence, see Section~\ref{sec:initial_guess} for details. For details of how we determine convergence in step 6 of Algorithm~\ref{alg:Paul} see Section \ref{S:convergence}.

Our approach is closely related to approaches for sampling from quasi-stationary distributions \cite[see][and references therein]{griffin:2017}. This can be most clearly seen in situations where $\norm{\mathbf{A}}$ is bounded. In this case we can define a stochastic process with transitions given by (\ref{eqn:MP}) and with killing at each iteration, with the probability of survival being proportional to $\norm{\mathbf{A}\widetilde{\mathbf{\Theta}}_s}^\kappa$. In this case the spectral measure, $H_{\hat{\mathbf{\Theta}}_0}$, is the quasi-stationary distribution of the process, i.e., the stationary distribution of the process conditional on survival.

Now consider the situation when $\kappa$ is not known. 
For a trial value of $k$ (for $\kappa$), apply Algorithm~\ref{alg:Paul} until convergence, giving a sample of weighted particles  $\{ \widetilde{\mathbf{\Theta}}^{(j)}(k) , m^{(j)}(k) \}_{j=1}^J$ after the chain is deemed to have converged. Using these particles approximate the expectation 
$E(\norm{\mathbf{A}\widetilde{\mathbf{\Theta}}_0}^{k})$
using the Monte Carlo sample by
\begin{equation}
\tilde{\rho}_k=\int_{\mathbb{R}}\sum_{j=1}^J \norm{\mathbf{A}\widetilde{\mathbf{\Theta}}^{(j)}(k)}^k\frac{m^{(j)}(k)}{\sum_{n=1}^J m^{(n)}(k)} F_Z(dz),
\label{eqn:rho_k}
\end{equation}
where $\tilde{\rho}_k$ is the Monte Carlo approximation of $\rho_k$ and $F_Z$ is the distribution function of the GARCH($p,q$) process. We repeat this evaluation over $k>0$ until we find the unique value of $k$ which gives this weighted mean to be equal to 1. This value is $k=\kappa$.

\begin{minipage}{0.875\textwidth}
\begin{algorithm}[H] \caption{Sampling from $H_{\hat{\mathbf{\Theta}}_0}(\mathbf{w})$} \label{alg:Paul} 
\DontPrintSemicolon
Generate a sample of $\widetilde{\mathbf{\Theta}}_0$ from any distribution  of $\mathbb{S}^{p+q}$. See Section~\ref{sec:initial_guess} for discussion on optimisation this choice. Set $s=1$.\;

Generate $J$ independent copies of $\mathbf{A}$, denote these as $\mathbf{A}_s^{(j)}$ for $j=1,\ldots,J$.\;

Generate $J$ equally weighted particles at time $s-1$ by sampling independently from our approximation to the distribution of $\widetilde{\mathbf{\Theta}}_{s-1}$. Denote these particles as $\mathbf{\Theta}_{s-1}^{\star (j)}$ for $j=1,\ldots,J$.\;
		
Generate $J$ particles at time $s$,
		\begin{equation}\label{eq:update.spec.dist}
		\widetilde{\mathbf{\Theta}}_s^{(j)} = \frac{\mathbf{A}_s^{(j)}\mathbf{\Theta}_{s-1}^{\star (j)}} 
		{ \norm{\mathbf{A}_s^{(j)}\mathbf{\Theta}_{s-1}^{\star (j)}}}, ~~ j=1,\ldots,J.
		\end{equation}\;
Assign each particle a weight,
		\[
		m_s^{\star (j)} = \norm{\mathbf{A}_s^{(j)}\mathbf{\Theta}_{s-1}^{\star (j)}} ^\kappa,\;
		\]
		for $j=1,\ldots,J$, and then normalise these via
	\begin{equation}\label{eq:update.part.weights}
	m^{(j)}_s = \frac{m_s^{\star (j)} }{\sum_{j=1}^{J}m_s^{\star (j)}}.\;
	\end{equation}
	The resulting set of weighted particles, $\{ \widetilde{\mathbf{\Theta}}_s^{(j)} , m^{(j)}_s \}_{j=1}^J$ is our empirical approximation to the distribution of  $\tilde{\mathbf{\Theta}}_{s}$.\;
		
		If we have converged to stationarity, output the set of weighted particles. Otherwise set $s=s+1$ and go to step 2.\;
\end{algorithm}
\end{minipage}

\subsection{Generation of the tail chain of the squared process}
\label{sec:tail_chain}
The tail chain $\{\hat{X}^2_t\}_{t\ge 0}$ can be evaluated using Algorithm~\ref{alg:tail_chain}. 
There are two stages to the algorithm, initialisation and propagation of the chain.  Key to getting the tail chain is finding the joint behaviour of $\hat{\mathbf{\Theta}}_t=(\hat{\vartheta}_t^{(1)}, \ldots ,\hat{\vartheta}_t^{(p+q)})$ over time. 

For initialisation we first need to consider the behaviour of the process conditional on it being in an extreme state, and we take the time of this, for convenience, to be $t=0$. Focusing on limit~\eqref{eqn:spectral_limit} when $t=0$, we have that the limit variables 
$(\hat{R}_0, \hat{\mathbf{\Theta}}_0)$, are independent with distributions~\eqref{eqn:Rmarginal}  and  
$H_{\hat{\mathbf{\Theta}}_0}$, given by expression~\eqref{eqn:Hform}, respectively. We initialise the chain when 
 $\hat{X}^2_0>1$, which is equivalent to $\hat{R}_0\hat{\vartheta}_0^{(1)}>1$.
For propagation of the chain we can use the established results of \cite{ba+da+mi:02b}, which give that for $t\ge 1$ that 
$\hat{\mathbf{\Theta}}_t = \mathbf{A}_t \hat{\mathbf{\Theta}}_{t-1}$ and hence 
\begin{equation}
\hat{\mathbf{\Theta}}_t = \mathbf{A}_t \mathbf{A}_{t-1} \cdots \mathbf{A}_1 \hat{\mathbf{\Theta}}_0.
\label{eqn:TailChain}
\end{equation}
We extract the tail chain $\{\hat{X}^2_t\}_{t\ge 0}$ from the product of $\hat{R}_0$ and $\hat{\vartheta}_t^{(1)}$ for all $t$. 

If interest is in the tail chain of the $\hat{\sigma}_t^2$ process instead, i.e., $\{\hat{\sigma}^2_t\}_{t\ge 0}$ given that 
$\hat{\sigma}^2_0>1$, then exactly the same approach can be taken as in Algorithm~\ref{alg:tail_chain} but with the condition in step 2 changed to $\hat{R}_0\hat{\vartheta}_0^{(q+1)}>1$. 

\noindent \begin{minipage}{0.85\textwidth}
\begin{algorithm}[H] \caption{Obtaining the tail chain of \garchpq} \label{alg:tail_chain}
	\DontPrintSemicolon
Simulate $\hat{\mathbf{\Theta}}_0$ using Algorithm~\ref{alg:Paul} and independently set $\hat{R}_0=U^{-1/\kappa}$ where $U$ is uniform(0,1).\;
Repeat step 1 until $\hat{R}_0\hat{\vartheta}_0^{(1)}>1$.\;
Evaluate $\hat{\mathbf{\Theta}}_t=\mathbf{A}_t\hat{\mathbf{\Theta}}_{t-1}$ for $t=1, \ldots ,T$, for large $T$.\;
The tail chain $\hat{X}_t^2 = \hat{R}_0\hat{\vartheta}_t^{(1)}$ and associated volatilities are $\hat{\sigma}^2_t = 
\hat{R}_0\hat{\vartheta}_t^{(q+1)}$, for $t=0, \ldots ,T$.\;
Repeat steps 1-4 to evaluate properties of the tail chain of $\hat{X}_t^2$ given $\hat{X}^2_0>1$.\;
\end{algorithm}
\end{minipage}

Here $T$ is selected so that $\hat{X}^2_t=\hat{R}_0  \hat{\vartheta}_t^{(1)}<1$ for all $t>T$ with probability as close to $1$ as possible. This is achievable as all components of $\hat{\mathbf{\Theta}}_t$ have negative drift and converge to 0 almost surely. In practice, $T$ is taken as large as possible subject to limits of storage and computational time, we took $T=1000$ to save running chains unnecessarily long, but for processes with weak extremal dependence $T=50$ is more than sufficient.

\subsection{The evaluation of cluster functionals}
\label{sec:CF}

From repeated realisations of the tail chain $\{\hat{X}_t^2\}_{t\ge 0}$ for the squared GARCH process we can derive the properties of key cluster functionals for the $X_t^2$ process, e.g., the extremogram, the extremal index and
the cluster size distribution.

First note that the extremogram for the squared GARCH$(p,q)$ process
is
\[
\chi_{X^2}(\tau)=\Pr(\hat{X}^2_{\tau}>1 | \hat{X}^2_{0}>1) 
\]
Thus we can numerically determine $\chi_{X^2}(\tau)$ as the proportion of tail chains starting above 1 at $t=0$ with an exceedance at $t=\tau$ over different replicate tail chains. Any required precision of this value can be achieved by a suitable selection of the number of Monte Carlo replicate tail chains. To derive both the extremal index and
the cluster size distribution we first define the measure introduced by \cite{roo:88}, namely
\[
\theta_{X^2}^{(i)} = \Pr \bigl(\#\{t=0,1,\ldots :\hat{X}^2_t>1\} = i
\mid \hat{X}_0^2 > 1 \bigr),
\]
i.e., the probability that there are at least $i$ values in a cluster given that we look at a cluster only forwards in time from an arbitrary exceedance. \cite{obr:87} showed that the extremal index $0<\theta_{X^2}\le 1$, is given by $\theta_{X^2}=\theta_{X^2}^{(1)}$ and
\cite{roo:88} showed that the cluster size distribution is given by
\begin{equation} \label{eqn:ratio.theta}
\pi_{X^2}(i) = \frac{\theta_{X^2}^{(i)} -
\theta_{X^2}^{(i+1)}}{\theta_{X_2}^{(1)}} \mbox{ for } i=1,2, \ldots ,
\end{equation}
with the reciprocal of the mean of this distribution being $\theta_{X^2}$.

\section{Tail chain properties for GARCH processes}
\label{sec:tailchain}

First note that if  $X_t^2$ is regularly varying with index $\kappa>0$ and if
\begin{equation} \label{eq.tailbalance}
\Pr(X_t> x \mid \abs{X_t}>x) \to \delta  \hspace{5mm} \mbox{ as } x \to \infty,
\end{equation}
where $0<\delta<1$ then it follows that $X_t$ is a regularly varying random variable, with index $2\kappa$, in both its upper and lower tails. Details of the evaluation of $\delta$ are given in Section~\ref{sec:delta}.

To translate results about the tail chain $\hat{X}^2_t$ of the squared GARCH process,  into properties for the tail chain $\hat{X}_t$ of the GARCH process we adopt a similar strategy to \cite{deh+r+r+dev:89} and \cite{eh+fi+ja+sc:2015}. It is key to recognise that there are two tails chains for $X_t$, an upper and a lower tail chain
$\hat{X}^L_t$ and $\hat{X}^U_t$ respectively, with
\begin{equation}
\hat{X}^U_t=I_t(\hat{X}^2_t)^{1/2} \mbox{ and }\hat{X}^L_t=-(1-I_t)(\hat{X}_t^2)^{1/2}
\label{eqn:thin}
\end{equation}
where $I_t$ is a sequence of independent and identically distributed Bernoulli$(\delta)$ variables, with $I_t=\{0,1\}$ with respective probabilities $\{1-\delta, \delta\}$ and where $\delta$ is given by limit~\eqref{eq.tailbalance},
The sequence $I_t$ is also independent of $\{\hat{Y}_t\}$.

Many properties of the cluster functions for $\hat{X}^U_t$ and $\hat{X}^L_t$ chains can simply be derived using Monte Carlo methods from the  $\hat{X}^2_t$ tail chains by using Bernoulli thinning, implied by expression~\eqref{eqn:thin}, but
some functionals can be explicitly determined, we study a few of these below.

First, we focus on the $\hat{X}^U_t$ tail chain, corresponding to positive events in the GARCH process.
The extremogram for $\hat{X}^U_t$ is given by
\[
\chi_{X^U}(\tau)=\Pr(\hat{X}^U_{\tau}>1 | \hat{X}^U_{0}>1)= \Pr(I_t\hat{X}^2_{\tau}>1)=\delta \chi_{X^2}(\tau)
\]
where $\chi_{X^2}(\tau)$ is the lag $\tau$ extremogram for the $X^2_t$ process.
An event in the tail chain for $\{\hat{X}^2_t\}$ with $i$ exceedances of the level $1$
does not occur in the tail chain of $\{\hat{X}^U_t\}$ with probability $\delta^i$.
Therefore summing over all event lengths, the probability of no exceedances of level $1$ from an event in the tail chain  of
$\{\hat{X}^2_t\}$ is given by
\[
\mathit{\Pi}^U= \sum_{i=1}^\infty \pi_{X^2}(i) (1-\delta)^i
\]
where $\pi_{X^2}(i)$ is the probability that a cluster of length $i$ in the $\{X_t^2\}$ series, see Section~\ref{sec:CF}. The probability of a cluster of length $j$ in the $\hat{X}^U_t$ tail chain is
\[
\pi_{X^U}(j)= \sum_{k\ge j} \pi_{X^2}(k) {k\choose j}\delta^j(1-\delta)^{k-j}/(1-\mathit{\Pi}^U),
\]
where the denominator corresponds to conditioning on the cluster for $\{\hat{X}^2_t\}$ process being retained for the
$\{\hat{X}^U_t\}$ series. Then the mean cluster size, $1/\theta_{X^U}$ for  $\{\hat{X}^U_t\}$ is given by
\begin{eqnarray*}
1/\theta_{X^U} & = & \sum_{j=1}^{\infty} j\sum_{k\ge j} \pi_{X^2}(k) {k\choose j}\delta^j(1-\delta)^{k-j}/(1-\mathit{\Pi}^U)\\
& = & \sum_{k=1}^{\infty} \pi_{X^2}(k)\sum_{j\le k} j \pi_{X^2}(k) {k\choose j}\delta^j(1-\delta)^{k-j}/(1-\mathit{\Pi}^U)\\
& = & \sum_{k=1}^{\infty} \pi_{X^2}(k) k\delta/(1-\mathit{\Pi}^U)\\
& = & \frac{\delta}{\theta_{X^2}(1-\mathit{\Pi}^U)},
\end{eqnarray*}
where $\theta_{X^2}$ is the extremal index of the squared GARCH process, see Section~\ref{sec:CF}.
So the extremal index of  $\{X^U_t\}$ is
\[
\theta_{X^U} =\theta_{X^2}(1-\mathit{\Pi}^U)/\delta.
\]

Similarly, for the lower tail behaviour of $\{X_t\}$ it follows that the extremogram is
\[
\chi_{X^L}(\tau)=\Pr(\hat{X}^L_{\tau}>1 | \hat{X}^L_{0}>1)=(1-\delta)\chi_{X^2}(\tau),
\]
the probability of no values below the level $-1$ by $X_t^L$ from an event in the tail chain  of  $\{\hat{X}^2_t\}$ is given by
\[
\mathit{\Pi}^L= \sum_{i=1}^\infty \pi_{X^2}(i) \delta^i,
\]
the probability of a cluster of length $j$ in the $\hat{X}^L_t$ tail chain is
\[
\pi_{X^L}(j)= \sum_{k\ge j} \pi_{X^2}(k) {k\choose j}(1-\delta)^j \delta^{k-j}/(1-\mathit{\Pi}^L),
\]
and extremal index of  $\{X^L_t\}$ is  $\theta_{X^L} =\theta_{X^2}(1-\mathit{\Pi}^L)/(1-\delta)$.

\section{Numerical Methods} \label{sec:numerical}
\subsection{Introduction}
Throughout this section a range of GARCH($p,q$) models will be illustrated. The details of these models are given here and will be referenced subsequently as GARCH models A-E, where 
\begin{description}
	\item [A]: $p=q=2$ with $(\alpha_1,\alpha_2, \beta_1,\beta_2)=(0.3,0.15,0.2,0.1)$
	\item [B]:  $p=q=2$ with $(\alpha_1,\alpha_2,\beta_1,\beta_2)=(0.07,0.04, 0.8,0.08)$ 
	\item [C]:  $p=q=1$ with $(\alpha_1,\beta_1)=(0.1, 0.9)$
	\item [D]:  $p=q=2$ with $(\alpha_1,\alpha_2,\beta_1,\beta_2)=(0.07, 0.03, 0.8, 0.1)$
	\item [E]:  $p=2$, $q=0$ with $(\alpha_1,\alpha_2)=(1.2, 0.5 )$.
\end{description}
We selected these models to give a variety of stationarity and extremal behaviours. Models A and B are second order stationary. Models C and D, which are IGARCH models, and model E are not second order stationary \citep[p.\ 35]{fr+za:2010}. 

Let $\phi=\sum_{i=1}^q \alpha_i+\sum_{i=1}^p \beta_i$ be the sum of the meaningful GARCH parameters. 
The parameter $\phi$ is increasing from model A to E with $\phi=1$ for models C and D. We will show that all these models are strictly stationary and that the marginal tail index $\kappa$ decreases with increasing $\phi$ for these models. From Theorem~\ref{th.IGARCHkappa} we have that when $\phi=1$ then $\kappa=1$, we will also illustrate this numerically for models C and D. Model B corresponds to the model studied by \citet{mi+st:00}. Case C, though being an IGARCH(1,1) process, is not covered by previous results of \citet{la+ta:2012} given its IGARCH form, but is of interest here as it helps to illustrate the new methods in a case where analytical solutions are possible. 

In Sections~\ref{sec:numerical} and \ref{sec:results} we take the the distribution of the innovation process $Z_t$, 
to be standard Gaussian, a scaled Student-$t_{\nu}$ distribution, and the skew Student-$t_{\nu}$ distribution introduced in \cite{az+ca:03}. In each case the innovation distribution has zero mean and unit variance. First consider the univariate skew-$t$ distribution, denoted by $\textrm{St}(\mu, \omega, \xi,\nu)$, where $(\mu, \omega, \xi,\nu) \in \mathbb{R}\times \mathbb{R}_+\times \mathbb{R}\times (2,\infty)$ are location, scale, skewness and degree of freedom parameters respectively. For the existence of the variance of $Z$ we require that $\nu >2$. The distribution of $Z$ has density
\[
f_Z(z; \mu, \omega, \xi,\nu) = \frac{2}{\omega}f_T(z_S;  \nu) F_T\left(z_S\xi\sqrt{\frac{\nu+1}{\nu + z_S^2}}; \nu+1\right),
\]
\noindent where $z_S=(z-\mu)/\omega$, and $f_T$ and $F_T$ denote,  respectively, the density and distribution function of a Student-$t$ random variable with location and scale parameters 0 and 1 respectively and with $\nu$ degrees of freedom. With this notation we have
\[
E(Z) = \mu + \omega  b_{\nu, \xi} \quad \quad \text{and} \quad \quad \mbox{var}(Z) = \omega^2 \Bigl[\frac{\nu}{\nu -2} - b_{\nu, \xi}^2\Bigr]
\]
where 
\[
b_{\nu, \xi} = \frac{\xi}{(1+\xi^2)^{1/2}}\left(\frac{\nu}{\pi}\right)^{1/2} \frac{\Gamma(\nu/2 + 1/2)}{\Gamma(\nu/2)},
\]
with $\Gamma$ being the gamma function. To impose the moment conditions on $Z$ of model~\eqref{mult.model.intro} we also require that
\[
\omega=\left[\frac{\nu}{\nu-2}- b^2_{\nu, \xi} \right]^{-1/2}
\mbox{ and } \mu=-\omega b_{\nu, \xi}, 
\]
where $\nu$ and $\xi$ are constrained further to ensure that $\omega$ is a positive real number. 
The parameter  $\xi \in \mathbb{R}$ controls the skewness: $\xi>0$ and $\xi<0$ correspond to right and left skew respectively, and $\xi=0$ to a symmetric distribution. 
Important special cases arise when: $\xi = 0$, $Z$ is the (scaled) Student-$t_{\nu}$ distribution; $\nu \to \infty$, $Z$ is a skew-Normal distribution; $\xi = 0$ and $\nu \to \infty$, a standard normal distribution. If not stated below we assume that  $Z_t$ follows a standard normal distribution.

Table~\ref{tab:key-stationarity} presents values for the key stationarity and extremal properties, $\gamma, \eta, \kappa, \theta_{X^2}, \theta_{X^U}, \theta_{X^L}$ and $\delta$ for each of models A-E and for  Student-$t$, asymmetric Student-$t$ and Gaussian innovations. These values are derived using the numerical methods in the rest of Section~\ref{sec:numerical}, 
and their values are discussed in Sections~\ref{sec:numerical} and \ref{sec:results}.

\begin{table}[htbp]
  \centering
\begin{tabular}{|l|r|r|r|r|c|c|l|}
\hline
Model & \multicolumn{1}{c|}{$\gamma$} & \multicolumn{1}{c|}{$\eta$} & \multicolumn{1}{c|}{$\kappa$} & \multicolumn{1}{c|}{$\theta_{X^2}$} & \multicolumn{1}{c|}{$\theta_{X^U}$} & \multicolumn{1}{c|}{$\theta_{X^L}$} & \multicolumn{1}{c|}{$\delta$}  \\
\hline
A - 1 & -0.4186 & 0.071 & 1.27  & 0.64  & \multicolumn{2}{c|}{0.76} &  0.5 \\
\cline{1-8}A - 2 & -0.4611 & 0.041 & 1.23  & 0.69  & \multicolumn{1}{r|}{0.72} & \multicolumn{1}{r|}{0.89} & 0.80  \\
\cline{1-8}A - 3 & -0.3358 & 0.023 & 2.37  & 0.59  & \multicolumn{2}{c|}{0.72} &  0.5 \\
\hline
B - 1 & -0.0400 & 0.003 & 1.26  & 0.38  & \multicolumn{2}{c|}{0.49} &   0.5\\
\cline{1-8}B - 2 & -0.0305 & 0.016 & 1.09  & 0.37  & \multicolumn{1}{r|}{0.41} & \multicolumn{1}{r|}{0.57} & 0.74  \\
\cline{1-8}B - 3 & -0.0155 & 0.005 & 1.92  & 0.16  & \multicolumn{2}{c|}{0.24} & 0.5  \\
\hline
C - 1 & -0.0300 & 0     & 1     & 0.21  & \multicolumn{2}{c|}{0.29} &  0.5 \\
\cline{1-8}C - 2 & -0.0335 & 0     & 1     & 0.29  & \multicolumn{1}{r|}{0.33} & \multicolumn{1}{r|}{0.45} & 0.71  \\
\cline{1-8}C - 3 & -0.0082 & 0     & 1     & 0.03  & \multicolumn{2}{c|}{0.05} & 0.5  \\
\hline
D - 1 & -0.0208 & 0.007 & 1     & 0.21  & \multicolumn{2}{c|}{0.29} &  0.5 \\
\cline{1-8}D - 2 & -0.0234 & 0.007 & 1     & 0.27  & \multicolumn{1}{r|}{0.31} & \multicolumn{1}{r|}{0.44} & 0.72  \\
\cline{1-8}D - 3 & -0.0062 & 0.002 & 1     & 0.03  & \multicolumn{2}{c|}{0.05} & 0.5  \\
\hline
E - 1 & -0.7461 & 0.137 & 0.65  & 0.27  & \multicolumn{2}{c|}{0.40} &  0.5 \\
\cline{1-8}E - 2 & -0.7595 & 0.129 & 0.68  & 0.29  & \multicolumn{1}{r|}{0.40} & \multicolumn{1}{r|}{0.45} & 0.56  \\
\cline{1-8}E - 3 & -0.2411 & 0.152 & 0.25  & 0.13  & \multicolumn{2}{c|}{0.22} & 0.5  \\
\hline
\end{tabular}%
 \caption{Values of key stationarity and extremal properties for models A-E for three innovation distributions: 1 $t_3$; 2 skew $t_3$ with $\xi=1$; and 3 Gaussian, with Model A - 3  denoting GARCH model formulation A with innovation distribution 3.}
 \label{tab:key-stationarity}%
\end{table}%

\subsection{Evaluation of $\gamma$ and $\eta$}
\label{sec:gamma}

Expression~\eqref{eq:lyapunov.monte.carlo} suggests using Monte Carlo for the evaluation of $\gamma$, by taking $t$ to be very large.  We have also introduced, through Theorems~\ref{th.gamma} and \ref{th.combined}, two new ways to evaluate $\gamma$, with the latter only applicable once it is known that the process is stationary, and hence it is known that $\gamma<0$. Here we compare these approaches to illustrate the superior computational stability and reliability of our proposed approaches.
First we illustrate the methods in Figure~\ref{fig.numerical.convergence} for model A which is known, from the discussions in Section~\ref{sec:stationarity}, to be stationary without the requirement of the evaluation of $\gamma$.

Figure~\ref{fig.numerical.convergence} (left) shows that $\gamma_t$, evaluated
using expression~\eqref{eq:lyapunov.monte.carlo}, has 
serious numerical instabilities for large $t$. All ten independent realisations of
$\gamma_t$ appear to be converging to roughly the same negative value as $t$ increases. This finding suggests that
the limit $\gamma$ is negative and so the process is correctly found to be strictly stationary. 
But at different random values of large $t$ each replicate stops at a time when the norm for that replicate is calculated as $0$ to machine precision. In these cases $\gamma_t=-\infty$ for all subsequent $t$ even though $\gamma$ is known to be finite. Hence wrong conclusions about strict stationarity can be reached for model A using this method. This numerical instability for evaluating $\gamma$ does not appear to have been reported. For example, estimates of $\gamma$ using this approach are presented for ARCH(2) processes in \citet[p. 34/35]{fr+za:2010}, however they stop evaluating $\gamma_t$, when $t=1000$, which is before we see  is the critical failure of numerical evaluation in  Figure~\ref{fig.numerical.convergence}. By increasing $t$ we find similar numerical problems to those experienced for model A in Figure~\ref{fig.numerical.convergence}. 
\begin{figure}
	\includegraphics[width=0.225\textwidth, angle= 270]{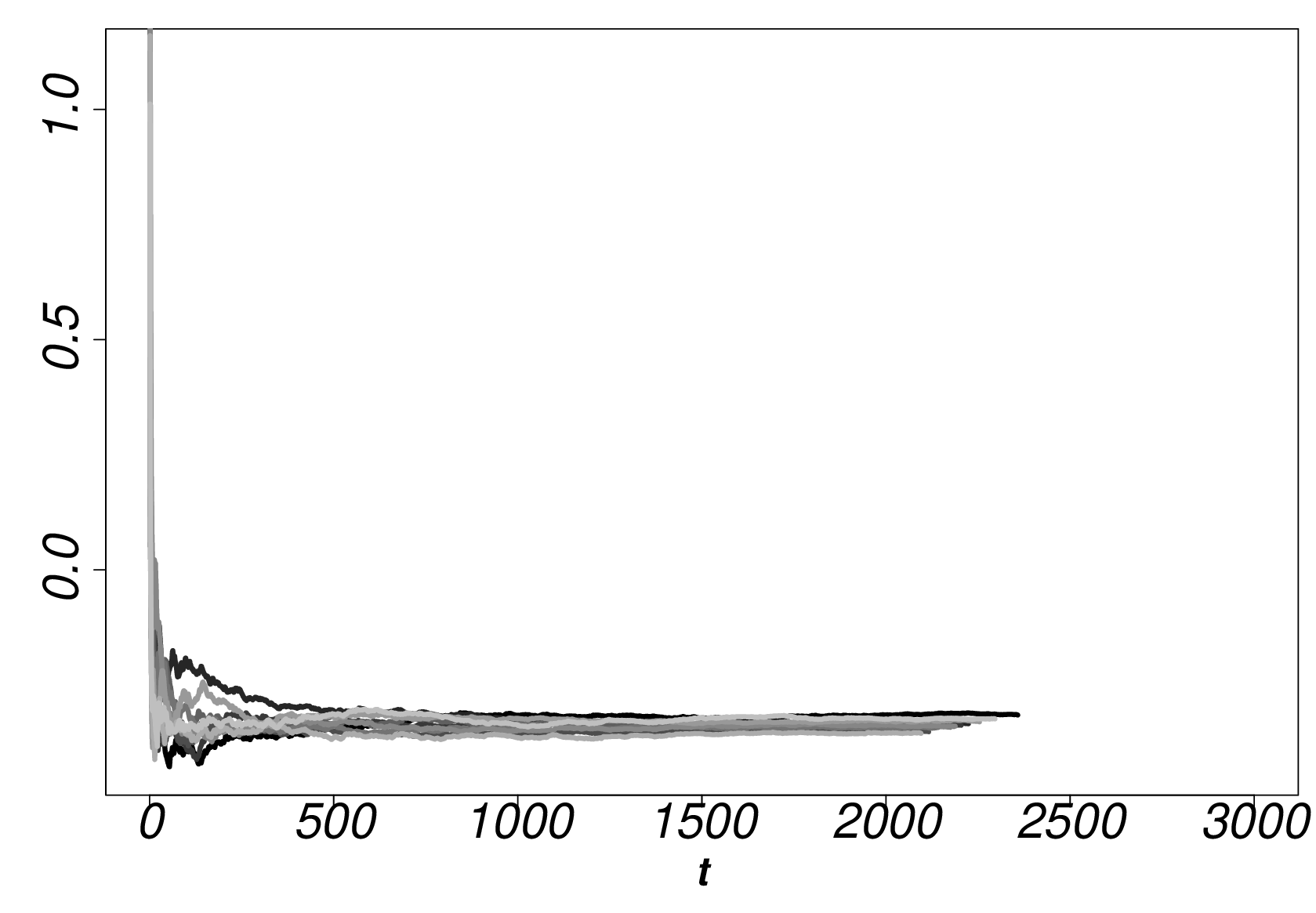}
	\includegraphics[width=0.225\textwidth, angle= 270]{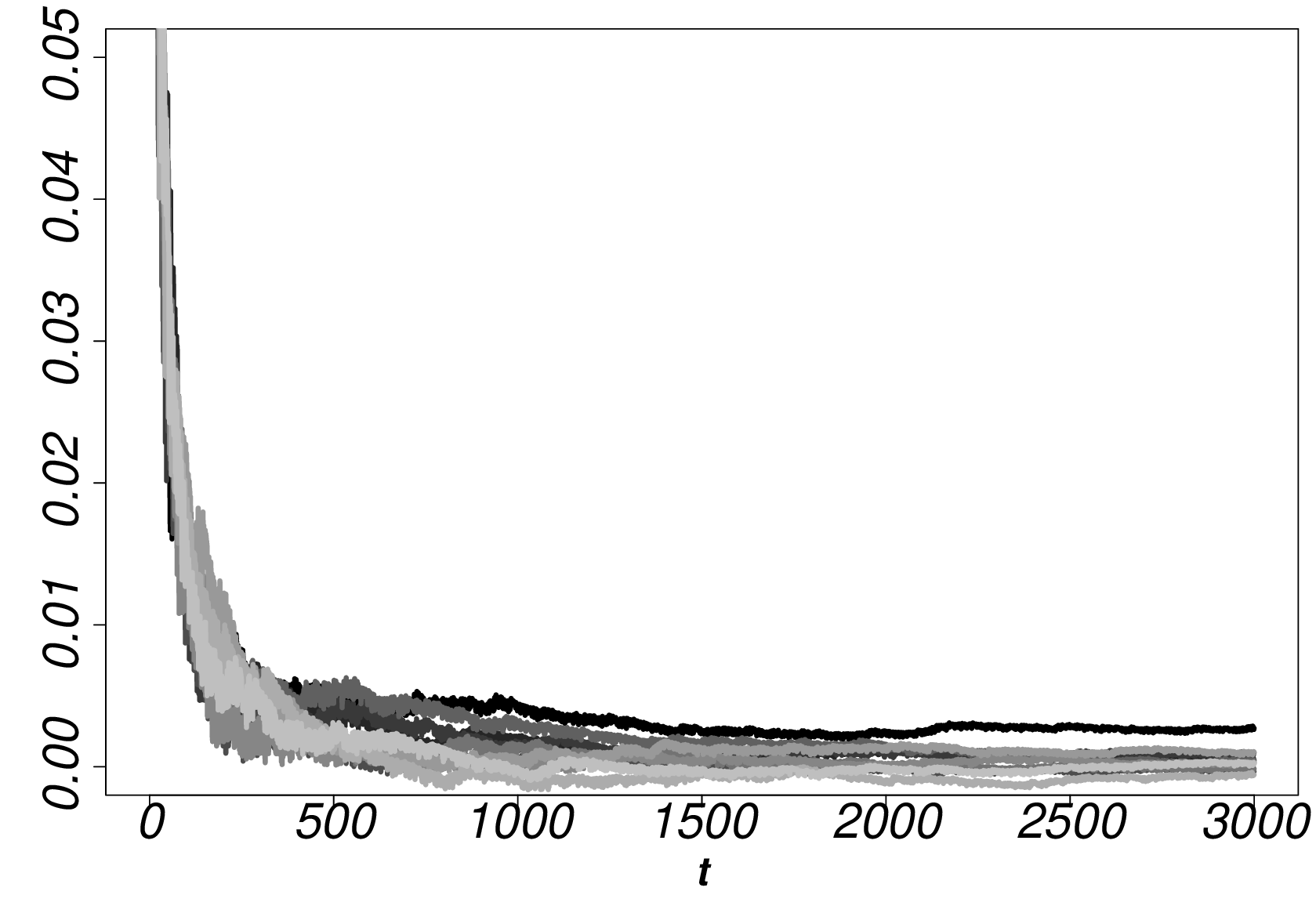}\includegraphics[width=0.225\textwidth, angle= 270]{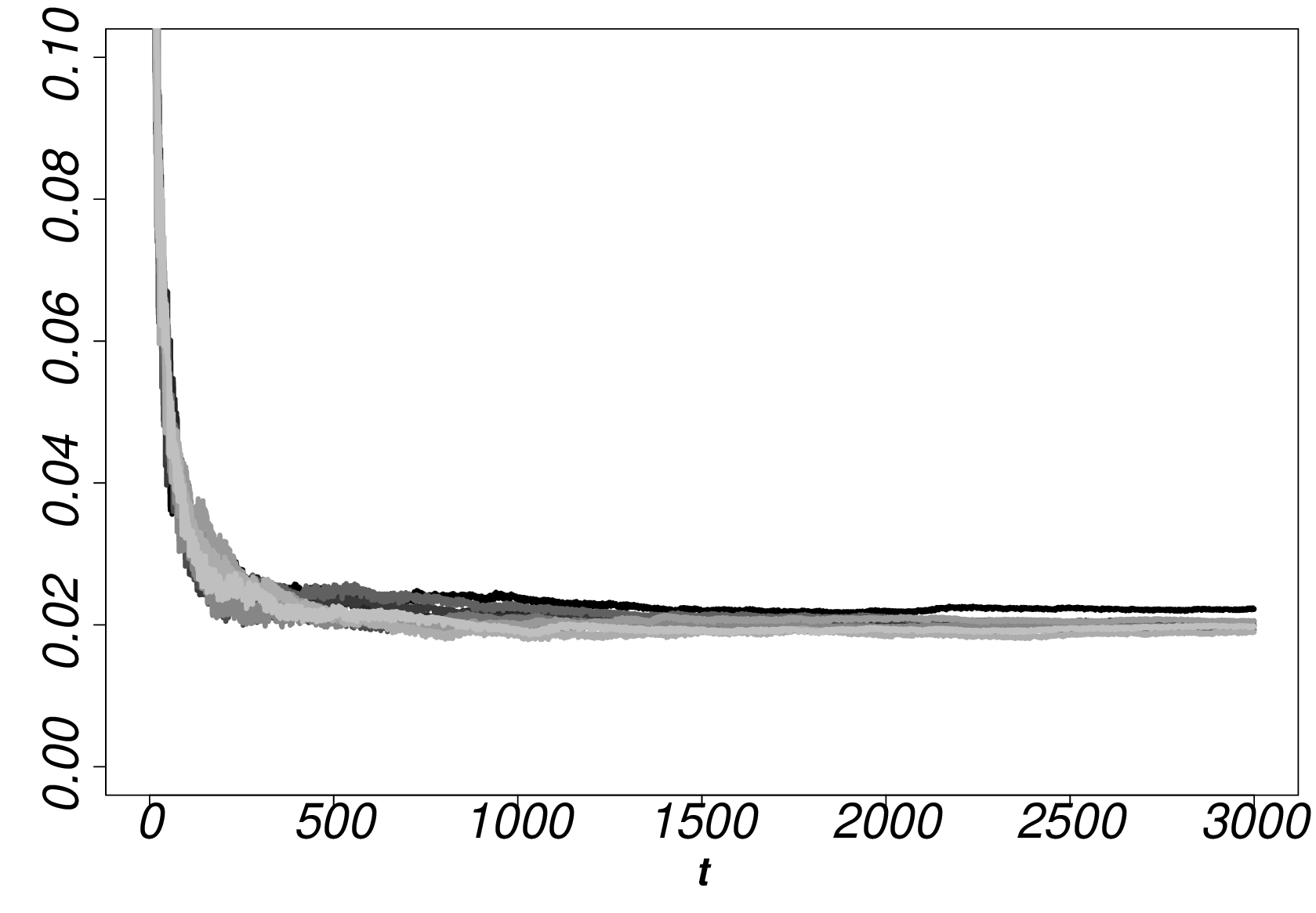}
	\caption{Monte Carlo properties for evaluation of $\gamma$ for GARCH model A against iteration $t$: left shows $\gamma_t$; middle $\frac{1}{t}\ln \norm{\mathbf{C}_t}$; right $\eta_t$.  All panels have the same 10 replicates of $\mathbf{A}_1, \ldots, \mathbf{A}_{t}$ displayed by greyscale lines.} \label{fig.numerical.convergence}
\end{figure}

Theorem~\ref{th.gamma} shows that actually $\gamma=E(\ln\lambda)+\lim_{t\rightarrow \infty} \eta_t$ with $\eta_t$ is defined by expression~\eqref{eq:defin.eta}, if condition~\eqref{eqn:Ccondition} holds. To use this result first we check that 
condition~\eqref{eqn:Ccondition} holds. For model A, Figure~\ref{fig.numerical.convergence}, shows $\frac{1}{t}\ln \norm{C_t}$ for the same replicates as in the left panel. This figure indicates that this quantity appears to be converging to $0$ as $t\rightarrow 0$, i.e., condition~\eqref{eqn:Ccondition} appears to hold.

For model A using the formulation of Theorem~\ref{th.gamma} we find that $\gamma=-0.359+0.019=-0.34<0$, which clearly satisfies the strict stationarity condition. We have evaluated $E(\ln\lambda)$ by both numerical integration and using the Monte Carlo approximation $\sum_{i=1}^{t}\ln(\lambda_i)/t$ for large $t$ where $\lambda_i$ is the largest magnitude of the eigenvalue of $\mathbf{A}_i$, and obtained identical and stable results once Monte Carlo errors are accounted for. Evaluations of the $\eta_t$ term are shown in  Figure~\ref{fig.numerical.convergence}  right panel, for the same 10 replicates as before. The numerical problems are now resolved with convergence similar in all cases, with greater variation between realisations than between iterations over large ranges of  $t$.  The value we present for $\eta=\lim_{t\rightarrow \infty} \eta_t$ is evaluated as the mean over the ten different realisations using when $t=30000$. Although Monte Carlo results are subject to noise,  it is possible to obtain any desired level of accuracy by running sufficient replicates when assessing the variability using central limit results \citep{golds:91}. Using the results of Theorem~\ref{th.kappa}, we can also evaluate $\eta$, but this requires knowledge of the value of $\kappa$. Using the methods of Section~\ref{sec:kappaEval}, based on  Algorithm~\ref{alg:Paul},  for model A we have that $\kappa=2.37$ and so it follows from Theorem~\ref{th.kappa} that $\eta=0.0233$. Thus this shows that trying to evaluate $\eta$ using the Monte Carlo limit, as in Figure~\ref{fig.numerical.convergence}, is liable to a small numerical error.

As stationarity of model A has been derived we can also evaluate $\gamma$ using the even more numerically reliable method given by Theorem~\ref{th.combined}, which comes from the expectations of two functions of the largest eigenvalue of the matrix $\mathbf{A}$. This evaluation requires the knowledge of $\kappa$. As for the evaluation of $\eta$ above, for model A we use $\kappa=2.37$, and we  then obtain that $\gamma=-0.336$. 

All of the values of both $\eta$ and $\gamma$ reported in Table~\ref{tab:key-stationarity} are evaluated using the methods based on Theorems~\ref{th.kappa} and \ref{th.combined} respectively, but of course they can only be used once stationarity has been determined, or at least there is strong evidence that $\gamma$ may be negative based on using the method based on Theorem~\ref{th.gamma}. Table~\ref{tab:key-stationarity} shows that although $\eta=0$ for all GARCH$(1,1)$ processes (confirming Theorem~\ref{th.kappa}) this does not hold for  any of our GARCH$(p,q)$ processes with $\max(p,q)\ge 2$. Furthermore, in all models, we have $\gamma$ closest to 0 with the Gaussian innovation, then the symmetric $t_3$ distribution. There is no obvious pattern in the behaviour of $\eta$ over the factors we explore in Table~\ref{tab:key-stationarity}.

\subsection{Initialising Algorithm~\ref{alg:Paul}}
\label{sec:initial_guess}

To be able generate realisations from the tail chain~\eqref{eqn:TailChain}, through Algorithm~\ref{alg:tail_chain}, we first need to generate samples for $\hat{\mathbf{\Theta}}_0$ using  Algorithm~\ref{alg:Paul}. However, to use Algorithm~\ref{alg:Paul}
we need to be able to sample from a suitable random variable $\widetilde{\mathbf{\Theta}}_0$ on $\mathbb{S}^{p+q}$
with a distribution which is as close as possible to the target limit distribution function $H_{\hat{\mathbf{\Theta}}_0}(\mathbf{w})$, 
so that the rate of convergence of $\widetilde{\mathbf{\Theta}}_s\rightarrow^d \hat{\mathbf{\Theta}}_0$ as $s\rightarrow \infty$ is maximised. 
From limit~\eqref{eqn:spectral_limit} we have that 
\[
\Pr(\mathbf{\Theta}^-_0 \le \mathbf{w} \mid R_0> x) {\to}  H_{\hat{\mathbf{\Theta}}_0}(\mathbf{w}), \hspace{5mm} \mbox{ as } x \to \infty.
\]
So for large enough $x$, i.e., $x\ge u$ for some high threshold $u$, if we treat this limiting
representation as an equality this gives us an initial estimate $H^{(0)}_{\widetilde{\mathbf{\Theta}}}$ of $H_{\hat{\mathbf{\Theta}}_0}$. We select $u$ as a high threshold of $R_t$ such that limit property~\eqref{eqn:spectral_limit}  appears to be well represented, i.e., radial values appear to have a Pareto tail and radial and angular values appear independent.

In practice to obtain $H^{(0)}_{\widetilde{\mathbf{\Theta}}}$ we generate a sample of length $n$ from the required GARCH($p,q$) process and take the empirical distribution of simulated values of $\mathbf{\Theta}_t$ given that $R_t>u$ after a burn in period of $n_b$
i.e., 
\[
H^{(0)}_{\widetilde{\mathbf{\Theta}}}(\mathbf{w})=\frac{\sum_{j=n_b+1}^n \mathbf{1}(R_j>u,
\mathbf{\Theta}^-_j\leq\mathbf{w})}{n_u},
\]
where $\mathbf{1}(F)$ is the indicator function of event $F$ and with $n_u=\sum_{j=n_b+1}^n \mathbf{1}(R_j>u)$.
As initial particles for Algorithm~\ref{alg:Paul} we use all the realisations of $\mathbf{\Theta}_t$ given that $R_t>u$, for $t=1,\ldots ,n$.
We used $n=1.1\times 10^7$, and $u$ to be the $99.99\%$ quantile of $R_t$ giving $J=10^3$ particles, each with equal weight $J^{-1}$.

\subsection{Investigation into convergence of Algorithm~\ref{alg:Paul}} 
\label{S:convergence}

We illustrate the convergence of Algorithm~\ref{alg:Paul} for models A and C. First consider model C where
the true distribution of $\hat{\mathbf{\Theta}}_0$, here a scalar, is given by expression~\eqref{eqn:HforGARCH11}. Hence we can compare the $s$ iteration estimate $\hat{H}^{(s)}_{\widetilde{\mathbf{\Theta}}}(w)$ against the truth
$\hat{H}_{\hat{\mathbf{\Theta}}_0}(w)$.  Figure~\ref{fig.stability.igarch11}  illustrates this distributional convergence as well as that of the distribution of the $J$ particle weights, $\mathbf{m}^{(s)}$ of expression~\eqref{eq:update.part.weights} on iteration $s$. First note that the 95\% pointwise confidence intervals of the initial estimate $\hat{H}^{(0)}_{\widetilde{\Theta}}$ given in Section~\ref{sec:initial_guess} does not the contain the true target distribution $\hat{H}_{\hat{\mathbf{\Theta}}_0}(w)$. Despite our efforts to obtain a good initial guess for Algorithm~\ref{alg:Paul}, the statistically significant difference between them is due to the slow convergence of the distribution of $\Pr(\mathbf{\Theta}_0<w|R_0>u)$ to 
$H_{\hat{\mathbf{\Theta}}_0}(w)$ as $u\rightarrow \infty$. 

After one step of Algorithm~\ref{alg:Paul} we have that 
$H^{(1)}_{\widetilde{\mathbf{\Theta}}}(w)$ is equal to $H_{\hat{\mathbf{\Theta}}_0}(w)$ to within visible detection. Further iterations of the algorithm lead to no visible changes in $\hat{H}^{(s)}_{\widetilde{\Theta}}(w)$ for $s>1$. In fact, in this example, we found essentially a perfect convergence after one iteration whatever the initial distribution estimate indicating a unique solution with the algorithm being robust and highly efficient in converging to it.  Now focus on the particle weights obtained in the algorithm. Initially, i.e., for $s=0$, all weights are equal $J^{-1}$, but, as Figure~\ref{fig.stability.igarch11} shows, within an iteration they have quite a different distribution of weights and that this distribution essentially has converged at $s=2$ to its limit form. Thus Algorithm~\ref{alg:Paul} works exceptionally well in this case where we know the answer. Similar tests over other GARCH(1,1) processes gave identical convergence performances.

\begin{figure}[H]
\begin{center}	
\includegraphics[width=0.25\textwidth, angle= 270]{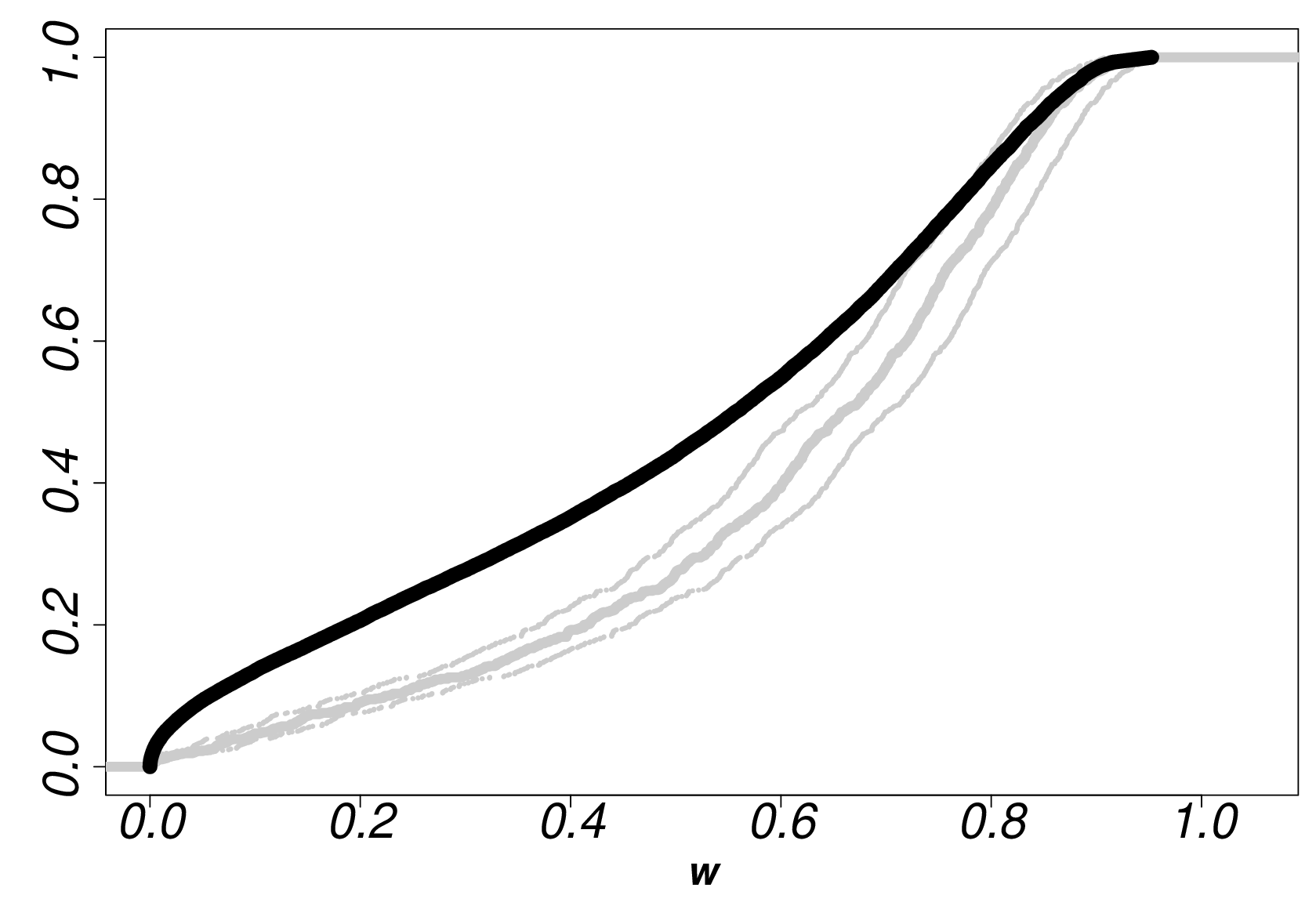}
	\includegraphics[width=0.25\textwidth, angle= 270]{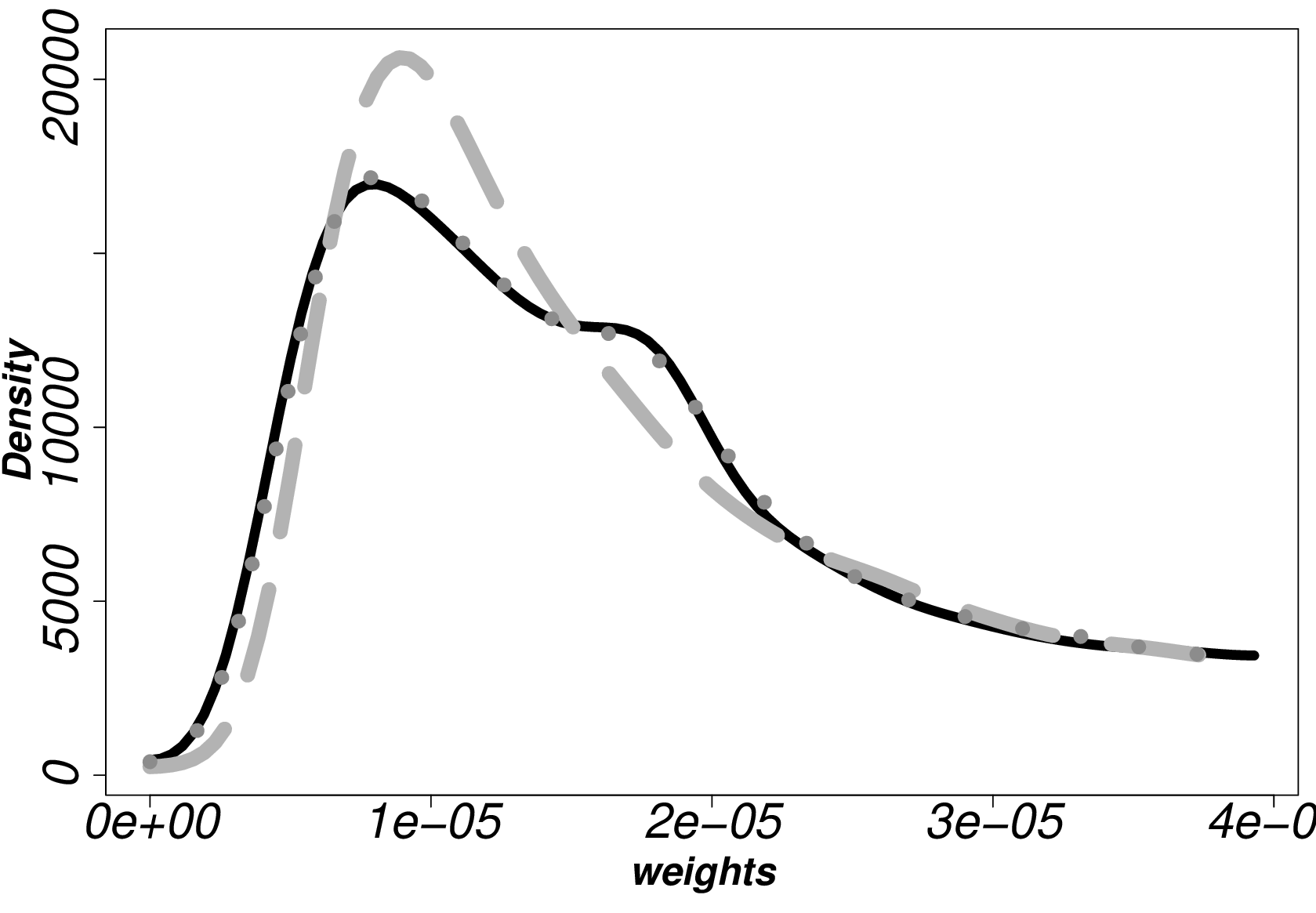}
	\end{center}
	\caption{Illustrations of Algorithm~\ref{alg:Paul} convergence for model C at iterations $s=\{0,1, 2, 100\}$. Left, 
	thick grey solid line is $H^{(0)}_{\widetilde{\mathbf{\Theta}}}(\mathbf{w})$ and the true limit distribution $H_{\hat{\mathbf{\Theta}}}(\mathbf{w})$ is shown by thick black line. For $s=0$ the 95\% confidence intervals are given by light grey lines. Right, kernel density estimate for
the particle mass. Line types are identical in each panel: $s=1$ - dashed grey line, $s=2$ - dotted dark grey line and $s = 100$  black thick solid line.} \label{fig.stability.igarch11}
\end{figure}

Next we assess the convergence of $H^{(s)}_{\widetilde{\mathbf{\Theta}}}(\mathbf{w})$, over $s$, for model A. Here 
$\hat{\mathbf{\Theta}}_0$ is four dimensional and its distribution is not known, so we cannot easily show graphically the convergence of the full joint distribution convergence and even for lower dimensional summaries we can only show the algorithm converges to some limit.  Figure~\ref{fig.stability.garch22} illustrates convergence for each of the marginal distributions of  $H^{(s)}_{\widetilde{\mathbf{\Theta}}}(\mathbf{w})$ over $s$. Other than for the second component the marginals appear to stabilise to their limit after just one iteration, for that margin it occurs in two iterations. The final panel of Figure~\ref{fig.stability.garch22} similarly shows that the distribution of the particle weights of the particles
also converges after two iterations. We also assessed (not shown) the convergence of the dependence structure  of 
$H^{(s)}_{\widetilde{\mathbf{\Theta}}}(\mathbf{w})$ through monitoring how $\mbox{corr}(\widetilde{\vartheta}^{(i)}_s,\widetilde{\vartheta}^{(j)}_s)$ converges to $\mbox{corr}(\hat{\vartheta}^{(i)}_0,\hat{\vartheta}^{(j)}_0)$, where 
$\widetilde{\mathbf{\Theta}}=(\widetilde{\vartheta}^{1}, \ldots ,\widetilde{\vartheta}^{p+q})$. In all cases we found rapid convergence.

We studied a number of other GARCH($p,q$) processes and found excellent convergence of the algorithm, with convergence appearing to occur in $p+q-1$ iterations in all cases with our initialisation method, and convergence to the same value over a range of other initialisation distributions. The theory behind Algorithm~\ref{alg:Paul} indicates that there is a unique solution and that the algorithm will find this, hence our numerical studies support this and show that it works with very high efficiency.

\begin{figure}[H]
\begin{center}	
\includegraphics[width=0.15\textwidth, angle= 270]{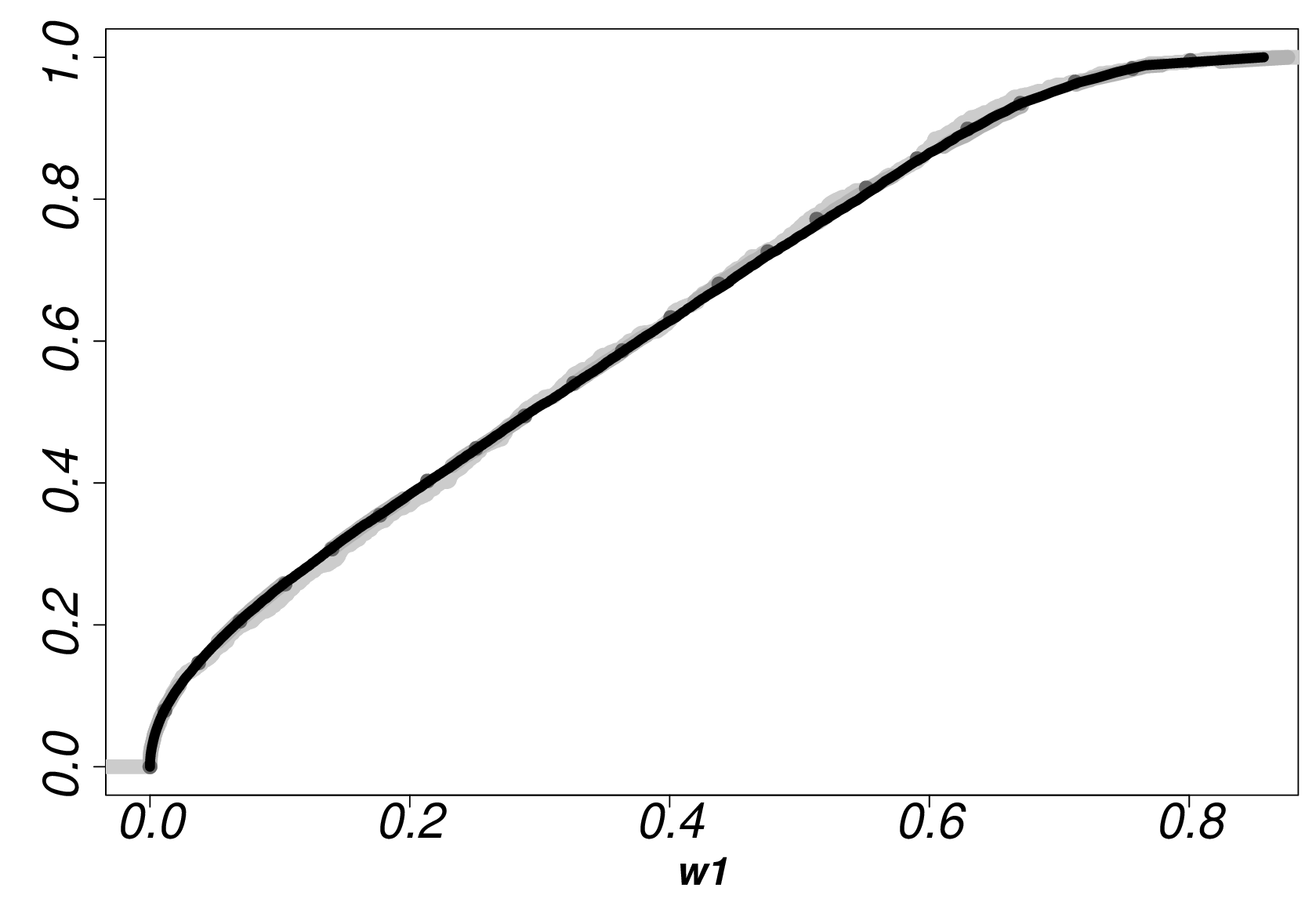}
\includegraphics[width=0.15\textwidth, angle= 270]{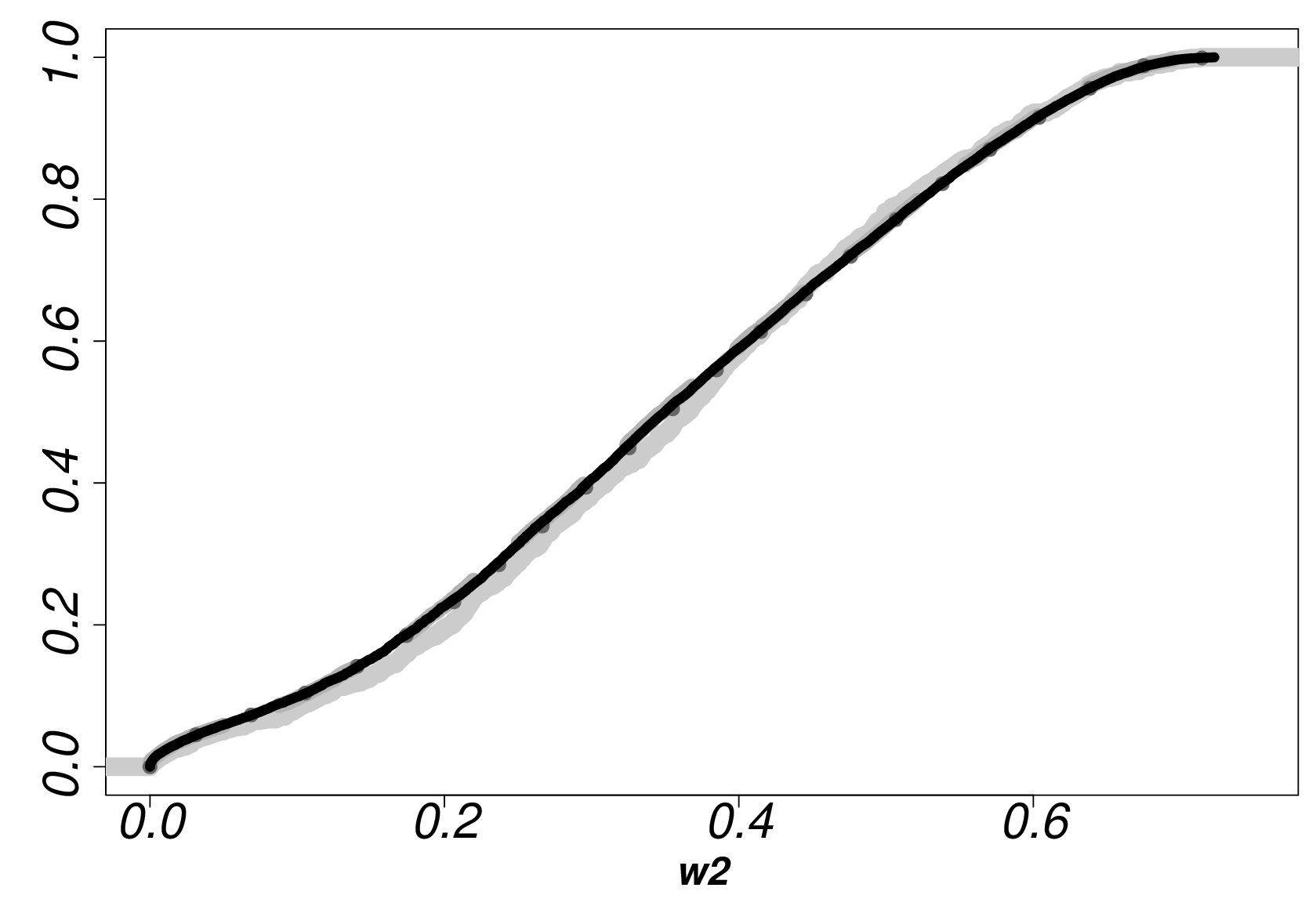}
\includegraphics[width=0.15\textwidth, angle= 270]{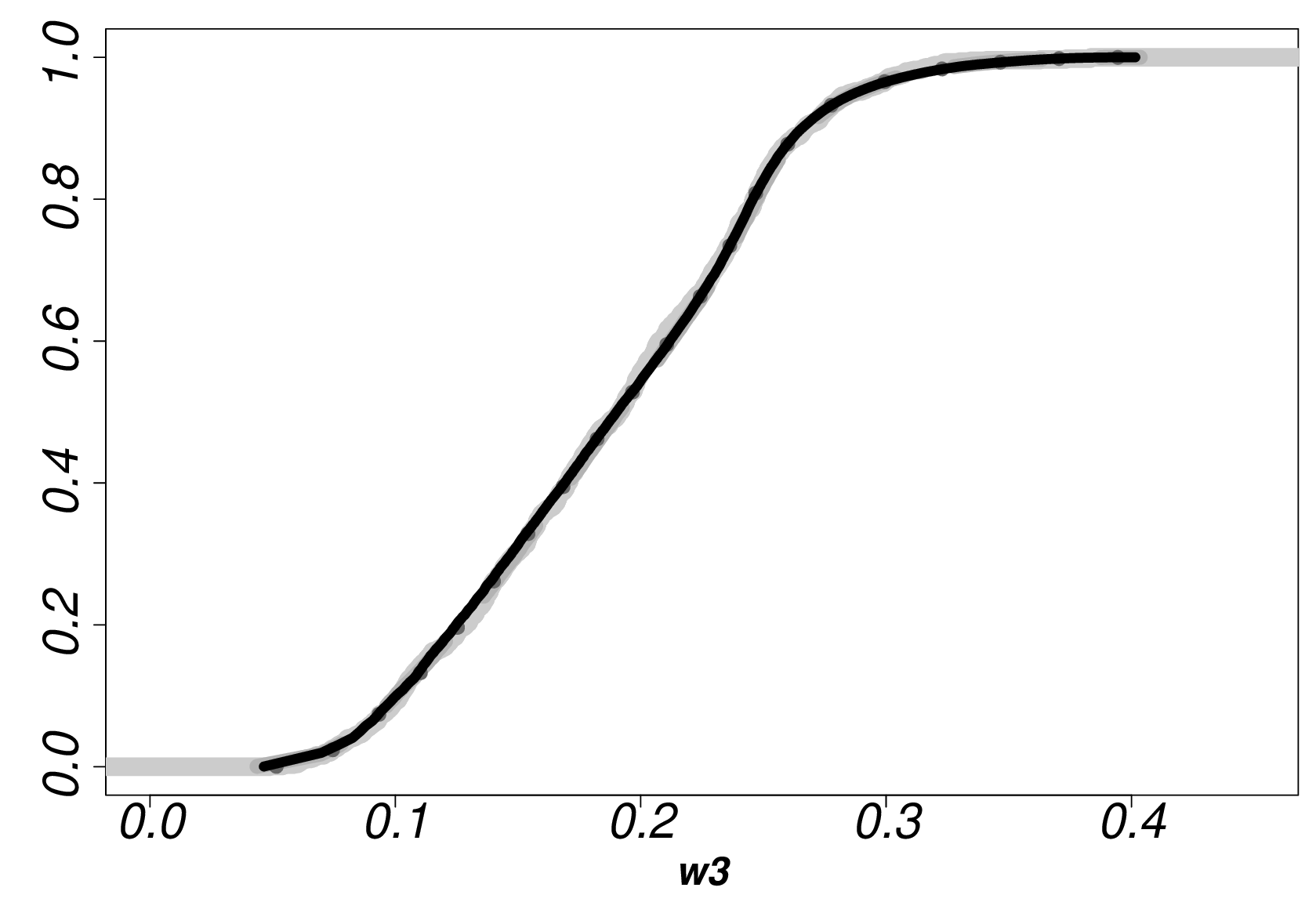}
\includegraphics[width=0.15\textwidth, angle= 270]{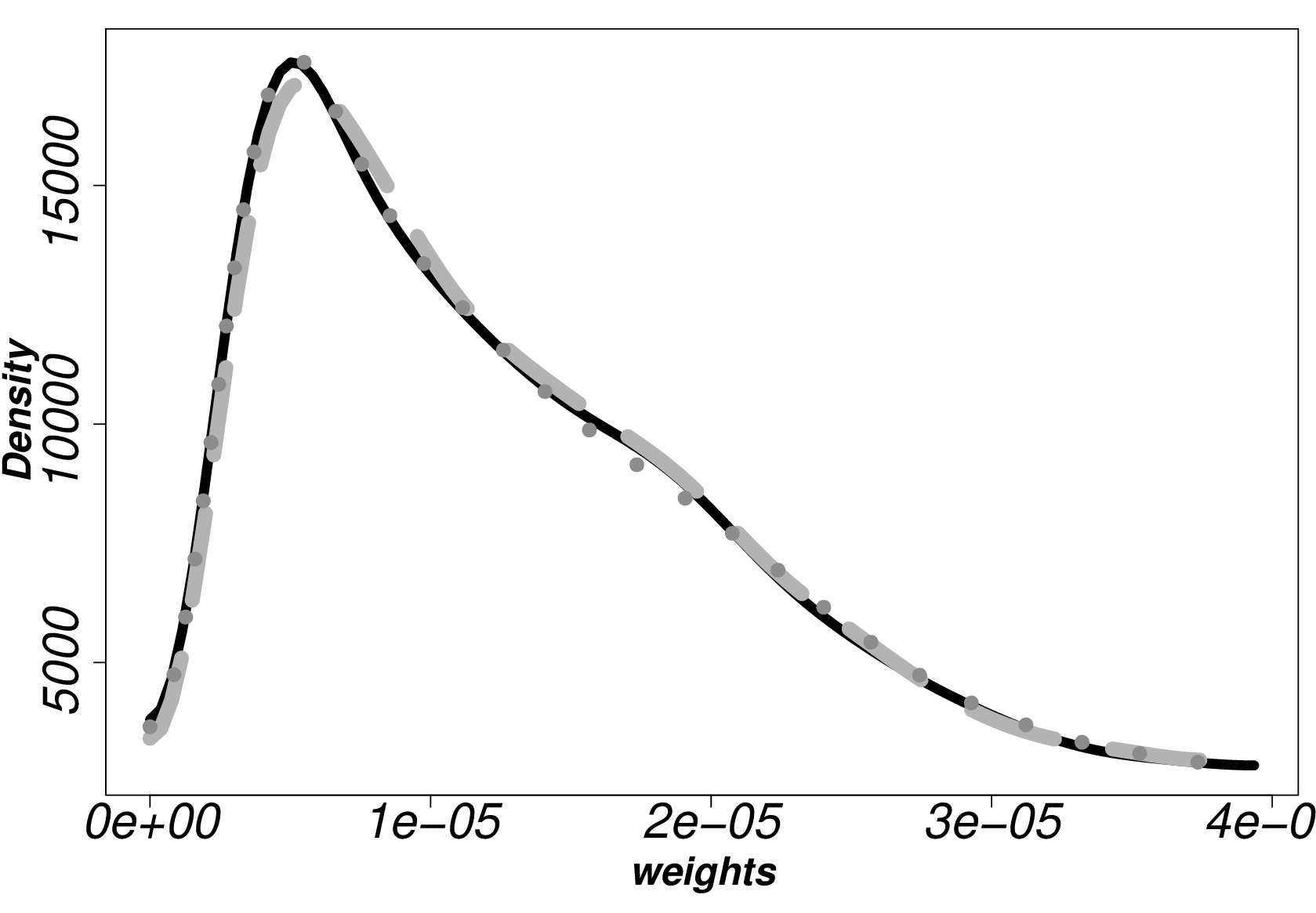}
\end{center}
\caption{Convergence of Algorithm~\ref{alg:Paul} for model A at iterations $s=\{0,1, 2, 100\}$ with marginal distribution convergence for  $H_{\widetilde{\vartheta}^{(i)}_s}(w_i)$, for $i=1,2,3$ and the kernel density for the particle mass. Values for $s>100$ are identical to those for $s=100$. Line types are as for Figure~\ref{fig.stability.igarch11}.}
\label{fig.stability.garch22}
\end{figure}


\subsection{Assessment of the convergence for the Extremal Index}
\label{sec:EIassessment}

We show that our solution for $H_{\hat{\mathbf{\Theta}}_0}(w)$ gives values of functionals, such as the extremal index, which are consistent with estimates obtained from long-run simulations from the GARCH($p,q$) process.
The convergence of Algorithm~\ref{alg:Paul} can be assessed for the cluster functionals by combining outputs from Algorithms~\ref{alg:Paul} and \ref{alg:tail_chain}. We illustrate this by finding $\theta_{X^2}$, the extremal index of the squared process for model A. Using the tail chain $\{\hat{X}_t; t=0,1, \ldots \}$ we have 
\[
\theta_{X^2}= \Pr\bigl(\hat{X}^2_t<1; t=1,2,\ldots \mid \hat{X}^2_0 > 1\bigr). 
\]
We can estimate the extremal index using the runs estimator $\tilde{\theta}_{X^2}(u,m)$, proposed by \citet{sm+we:94}, based on a sample from the GARCH($p,q$) process of length $n$, where
\begin{eqnarray*}
\tilde{\theta}_{X^2}(u,m) & = & \hat{\Pr}\bigl(X^2_t<0; t=1,2,\ldots ,m \mid X^2_0 > 1 \bigr)\\
& = & \frac{\sum_{j=1}^{n-m} \mathbf{1}(\max(X^2_{j+1}, \ldots, X^2_{j+m})<u, X^2_j>u)}{\sum_{j=1}^{n-m} \mathbf{1}(X^2_j>u)}.
\end{eqnarray*}
Figure~\ref{fig.runs.theta} shows $\theta_{X^2}$ together with the runs estimate, based on $n=10^7$, for a range of values of $u$ and $m$. The limit value is within the estimated 95\% confidence intervals for the runs estimate for all $u$ and both values of $m$ and the values of the runs estimate approach the true value $\theta_{X^2}$ as $u$ increases for both $m$ values, although the uncertainty in the estimators increases. This plot also goes to show why it is not possible to derive these extremal features of the GARCH($p,q$) process simply from very long runs as the competing needs of large $u$, for convergence, and large numbers of exceedances of $u$, for numerical stability, makes getting numerically reliable values essentially impossible without taking $n$ to be multiple orders of magnitude larger than here.

\begin{figure}[H]
\begin{center}
	\includegraphics[width=0.35\textwidth, angle= 270]{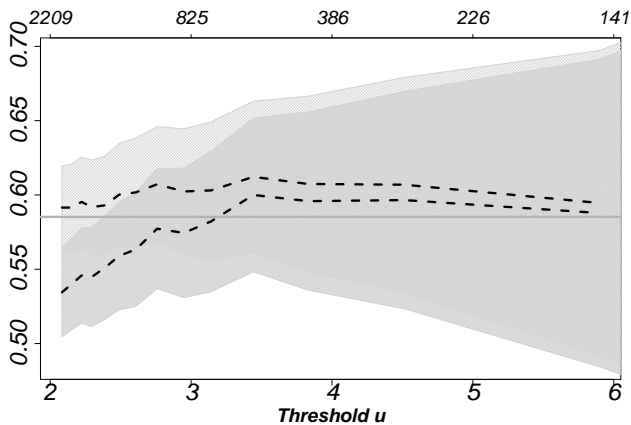}
	\end{center}
	\caption{The squared GARCH process extremal index $\theta_{X^2}$ for the model A shown by the horizontal solid line and runs-method based estimates 
	$\tilde{\theta}(u,m)$ obtained using a simulated process of length $10^7$. The runs estimates for different thresholds $u$ and $m=100$ and $1000$ (top and bottom dashed lines respectively) and the associated pointwise $95\%$ confidence intervals (grey shaded). The number of exceedances for each $u$ is reported along the top axis.} 
	\label{fig.runs.theta}
\end{figure}

\subsection{Evaluation of $\kappa$ } 
\label{sec:kappaEval}

\cite{ba+se:09}, and subsequent authors, imply that the way to evaluate $\kappa$ is to numerical solve the limiting equation~\eqref{eq:root.equation}, although they do not illustrate this. In Figure~\ref{fig.numerical.convergence} (left panel) we showed that there are major numerical instabilities in evaluating $\norm{\mathbf{A}_t \cdots \mathbf{A}_1}$ for large $t$; so in practice it is impossible to solve equation~\eqref{eq:root.equation} directly. In this paper we have discussed three alternative approaches for determining $\kappa$: the algorithm of \cite{ja:2010};  using the formulation for $\kappa$ given by Theorem~\ref{th.kappa};  and exploiting Algorithm~\ref{alg:Paul}. Here we describe, and illustrate, the relative merits of these methods.  

The algorithm of \cite{ja:2010} is only for bounded innovation variables. From a numerical efficiency perspective it suffers from the critical problem that as it is based on rejection sampling, meaning it can get seriously stuck. 
Finally, the routine was written in pure R and has a naive initialisation,  so it is very slow (taking 2/3 days) to evaluate $\kappa$ to an accuracy of three significant figures even when applied for a GARCH(2,1) model. The speed slows at a cubic rate as the number of the GARCH parameters grows. So this algorithm cannot be used for arbitrary GARCH($p,q$) processes, even with bounded innovations.

Both the new approaches that we present for evaluating $\kappa$ are not restricted by the choice of the GARCH dimensions $p$ and $q$, they apply whether the innovations are bounded or unbounded, and they are relatively much faster as they are coded in C wrapped by R and run in parallel. 

The first of our methods is based on the equivalent representation to limiting equation~\eqref{eq:root.equation}, i.e., that $\kappa>0$ satisfies $E[(\lambda\exp(\eta))^{\kappa}]=1$ as given by Theorem~\ref{th.kappa}. As $\lambda$ can be derived analytically (or, less efficiently, numerically) from $\mathbf{A}$ it remains to find $\eta$ and then $\kappa$ can be found when using either numerical integration or Monte Carlo methods to evaluate the required expectation.  We derive an estimate of $\eta$ using the methods presented in Section~\ref{sec:gamma}. Unfortunately, this approach is not ideal as is illustrated in Figure~\ref{fig.numerical.convergence}, which shows that the 10 replicates of $\eta_t$ are not sufficiently stable and self-consistent in their values at large iterations to accurately deduce the precise value of the limit $\eta$.

We find that this approach only works well for calculating $\kappa$ for models where $\abs{\phi -1} > 0.05$  as there is too much sensitivity to the uncertainty of $\eta$ otherwise. Although this is not the ideal way to evaluate $\kappa$, its form gives helpful intuition into  what influences $\kappa$. As seen in Section~\ref{sec:gamma} we can actually evaluate $\eta$ much more accurately, but that needs $\kappa$ to be found, so that would lead to a circular argument.

Our preferred approach to calculating $\kappa$ is to use Algorithm~\ref{alg:Paul}, iterating over $k$ to give $\kappa$,  as this provides no numerical problems whatever the dimension of the GARCH ($p,q$) process. Key to the solution is the evaluation of the Monte Carlo estimate $\tilde{\rho}_k$ of $\rho_k$ in expression~\eqref{eqn:rho_k}.
Figure~\ref{fig.root.function} shows $\tilde{\rho}_k$ against $k$ for each of the models A-E. There is clearly a unique solution for $k>0$ to the equation $\tilde{\rho}_k=1$, with the values of $k=\kappa$ that solve this equation given in Table~\ref{tab:key-stationarity}. However to  achieve this we need to reduce the noise in the Monte Carlo estimates $\tilde{\rho}_k$
of $\rho_k$. For each value of $k$ shown in Figure~\ref{fig.root.function} we used $J = 10^6$ and evaluated the Monte Carlo integral \eqref{eqn:rho_k} with $10^4$ replicates on $Z$ to get $\kappa$ to the required precision. To find $\kappa$, from the curve of $\tilde{\rho}_k$, we used an initial grid search coupled with a bisection method.

\begin{figure}[H]
  \centering
	\includegraphics[width=0.25\textwidth, angle= 270]{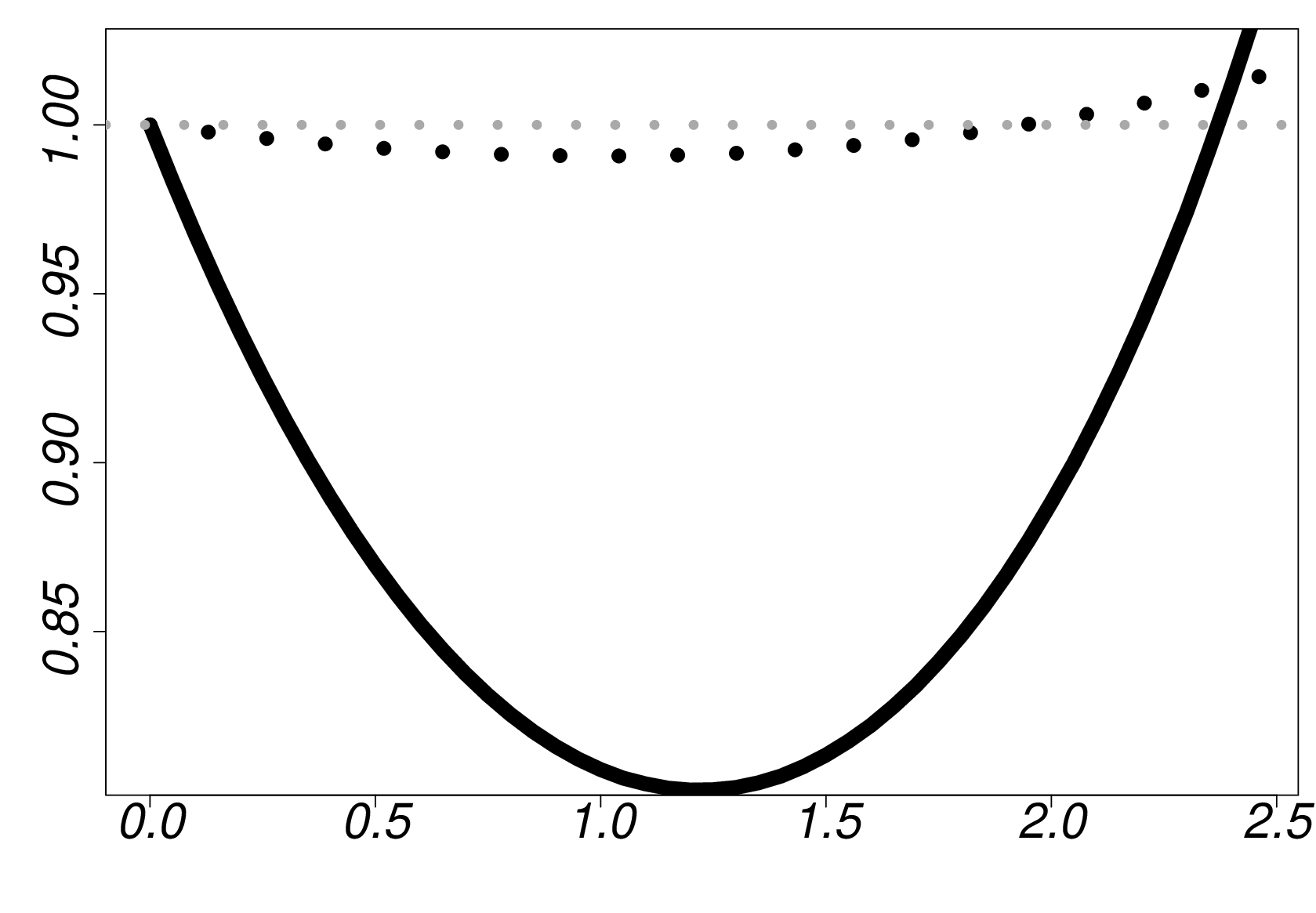}
		\includegraphics[width=0.25\textwidth, angle= 270]{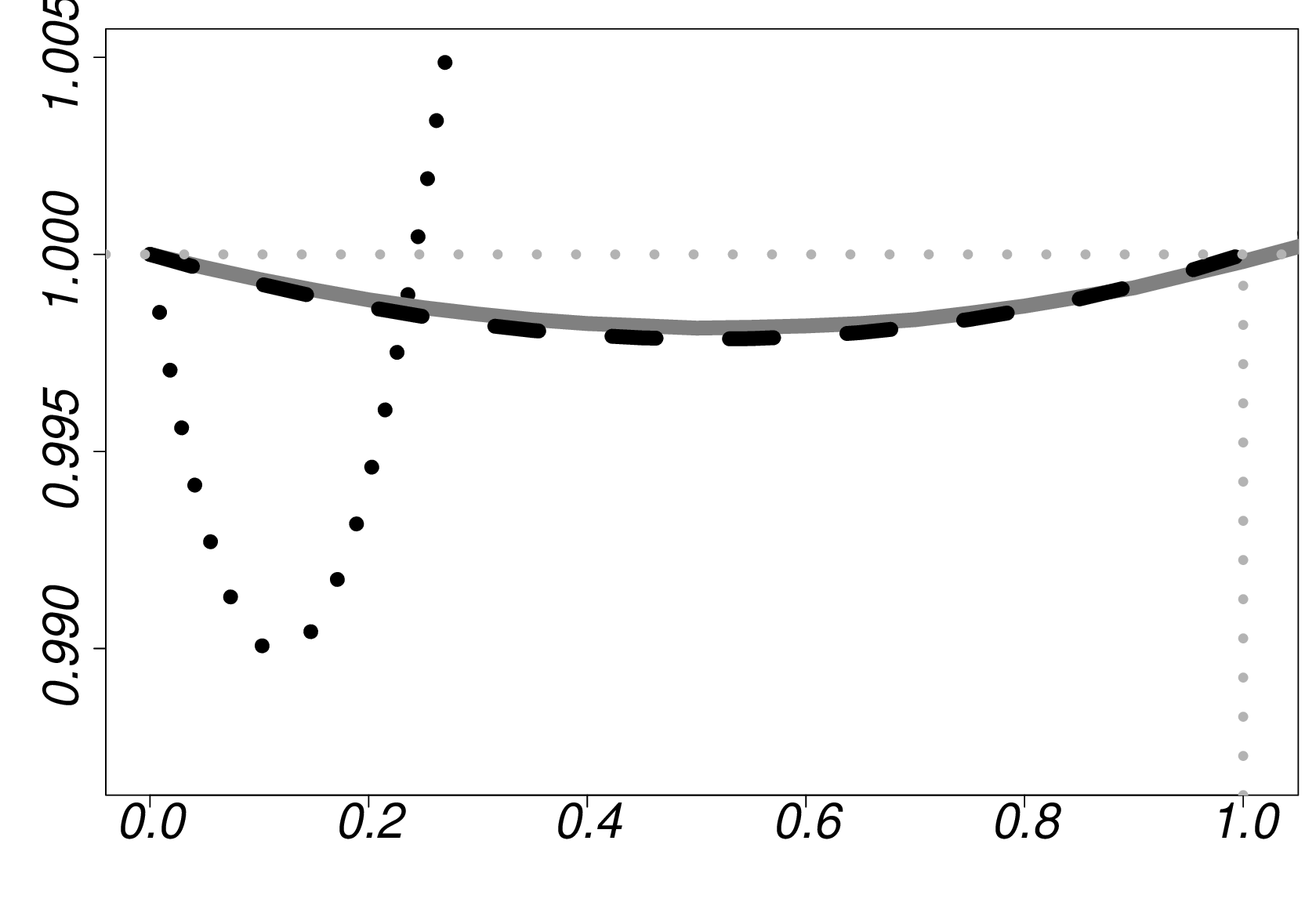}
	\caption{Plots of $(k, \tilde{\rho}_k)$: left, for  models A (---) and B ($\cdots$) which are second order stationary; right for models C (black dashed), D (grey solid) and E (black dotted), which are not second order stationary. In all panels grey dotted lines represent horizontal and vertical lines set at 1.} \label{fig.root.function}
\end{figure}

Figure~\ref{fig.root.function} and Table~\ref{tab:key-stationarity} illustrate that $\phi$, the sum of the meaningful GARCH parameters, has a substantial impact on the value of $\kappa$ with for $\phi>1$, we find $\kappa<1$; when $\phi<1$, then $\kappa>1$; and for  $\phi=1$, $\kappa=1$, with the latter consistent with Theorem~\ref{th.IGARCHkappa}. Unfortunately, when $\phi\not= 1$ no explicit relationship appears to hold between $\phi$ and $\kappa$, as $\kappa$ changes markedly with the innovation distribution. From Table~\ref{tab:key-stationarity} is can be seen that
when $\phi<1$ we have that the shorter the tail of the innovation distribution gives the larger $\kappa$ and hence shorter tails of the GARCH$(p,q)$ marginal distribution; whereas the reverse holds when $\phi>1$; and when $\phi=1$ then $\kappa$ is invariant to the innovation distribution. The case when $\phi>1$ is somewhat surprising as at first thought you would expect that having a heavier tail innovation would result in a heavier tailed GARCH$(p,q)$ process, whereas in fact the opposite occurs.

We finish with empirical diagnostic checks to illustrate that the derived value of $\kappa$ is consistent with the observable tail of the GARCH($p,q$) process. The observable tail can be derived from long run simulations. Specifically we compare the probabilities 
limiting $\Pr(\hat{X}_t^2 > r \mid \hat{X}_t^2 > 1)=r^{-\kappa}$ with the empirical estimate of the probabilities
$\Pr(X_t^2 > rx \mid X_t^2 > x)$ for very large $x$, over a range of $r>1$. 
Figure~\ref{fig.qq.kappa} shows this comparison on a log scale. With the choice of such scaling the true relationship between them has a gradient $\kappa$. The results show that at this far into the distributional tail, and subject to Monte Carlo noise, the empirical distribution is consistent with the limit formulation, and hence also is consistent with the value for $\kappa$ that we have derived.

\begin{figure}[H]
\centering
\includegraphics[width=0.25\textwidth, angle= 270]{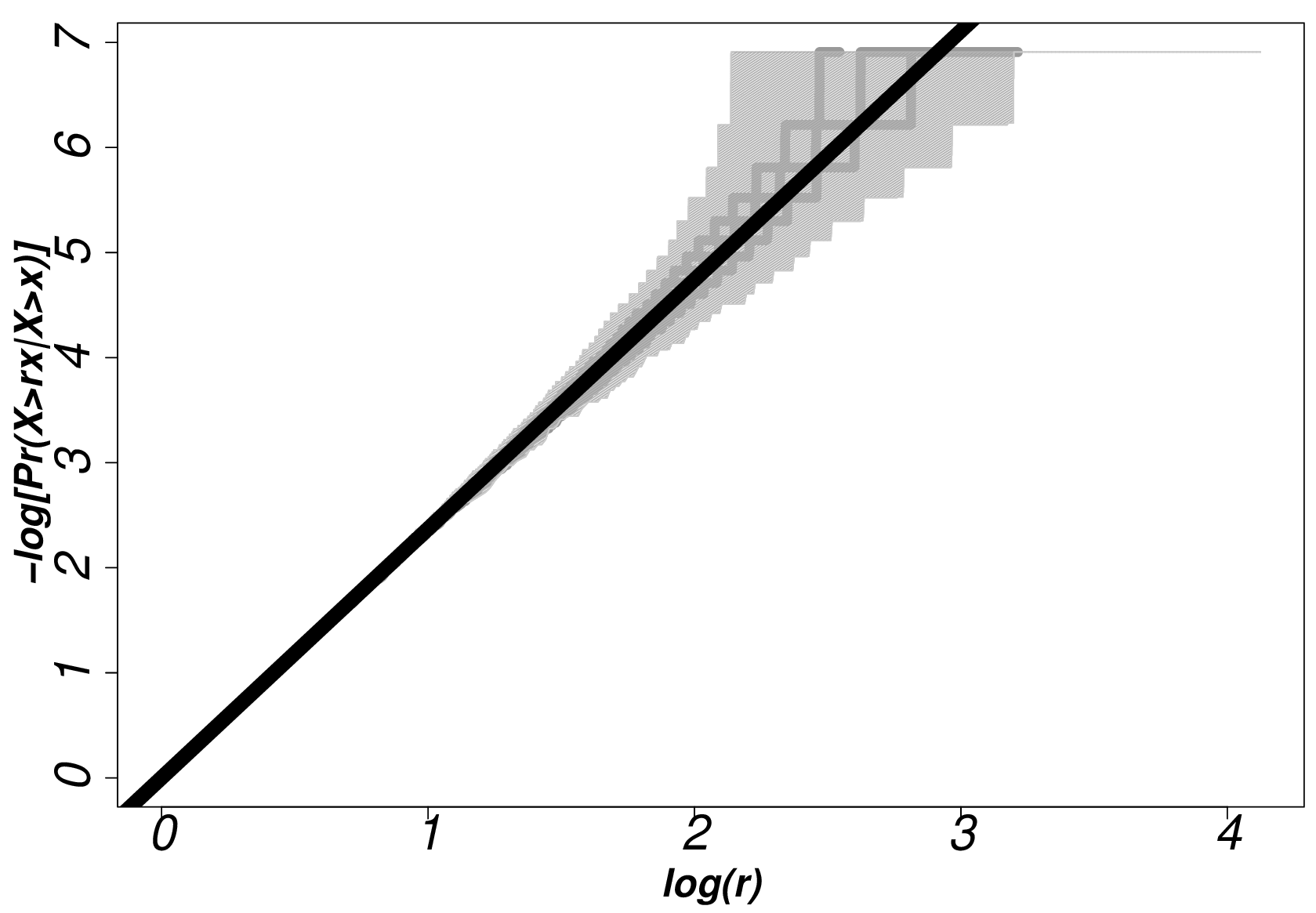}
\includegraphics[width=0.25\textwidth, angle= 270]{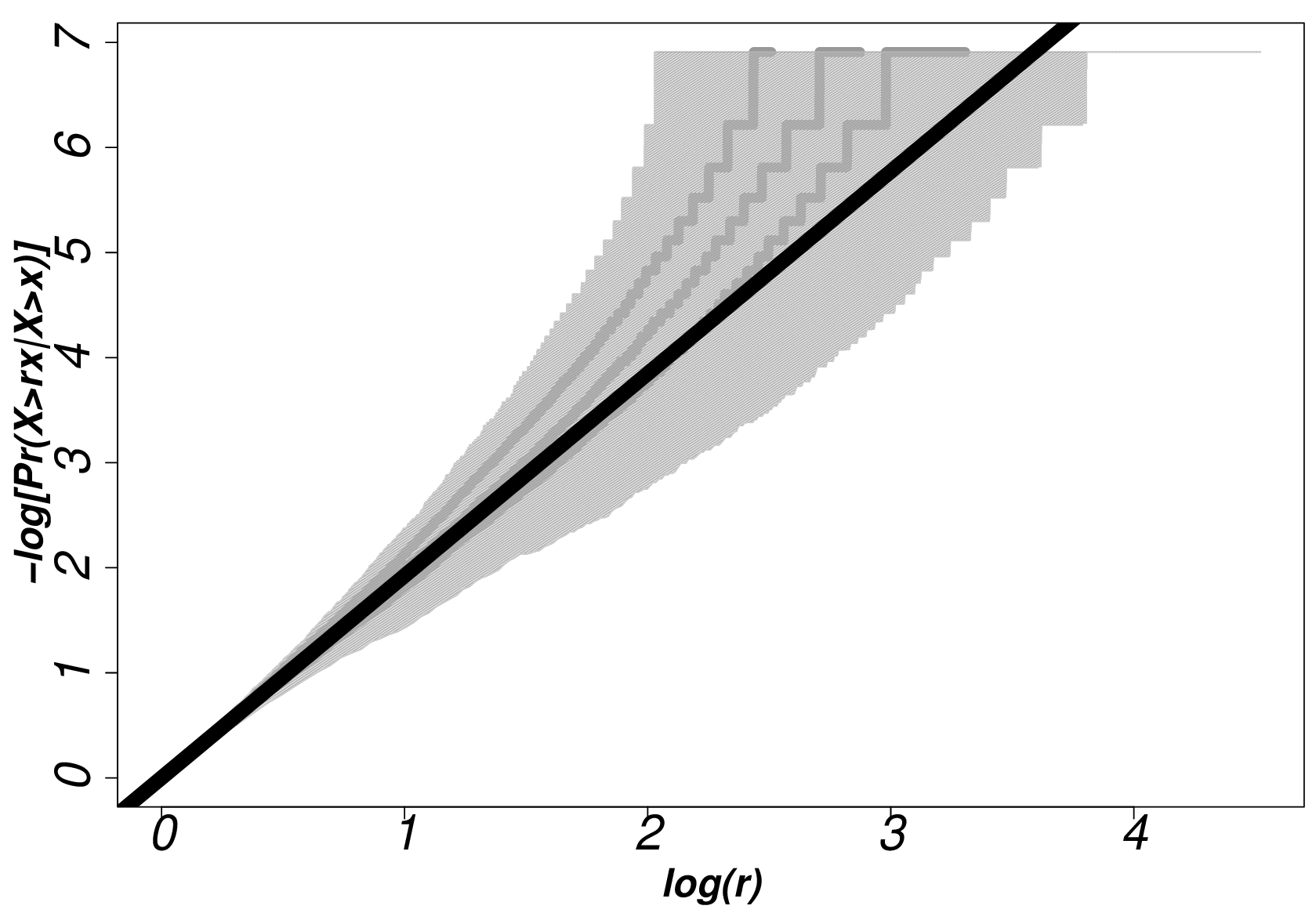}
	\caption{Diagnostic QQ plot for the marginal tail of the squared GARCH models A and B, left and right respectively, comparing empirical and limit distributions. Results are based on 1000 simulations of $5\times 10^7$ GARCH processes with threshold $x$ corresponding to the $0.99998$ marginal quantile. The solid line has a gradient of $\kappa$ and the conditional quantiles of empirical estimators are shown for $2.5\%-97.5\%$ as the shaded region and for $25\%, 50\%$ and $75\%$ quantiles as grey lines.} \label{fig.qq.kappa}
\end{figure}

\subsection{Evaluation of $\delta$} \label{sec:delta}

In Section~\ref{sec:tailchain} we introduced a tail-skewness parameter $\delta$ as a limiting conditional probability~\eqref{eq.tailbalance}. We have not found any previous discussion on the evaluation of $\delta$, which is an important parameter in the calculation of extremal features of GARCH($p,q$) processes when the innovations are asymmetrically distributed, despite \citet{eh+fi+ja+sc:2015} covering this class for GARCH(1,1) processes.
 A natural starting point to evaluate $\delta$ is to take a long-run simulation from the GARCH($p,q$) process and simply estimate the probability~\eqref{eq.tailbalance} empirically for a large enough value of $x$.  However, this is likely to be unreliable in practice. 

Key to the method we propose is the following expression
\begin{eqnarray}
\delta= \lim_{x\rightarrow \infty}\Pr(X_t>x\mid |X_t|>x)  & = & \Pr(\hat{X}_t>1\mid |\hat{X}_t|>1)\nonumber\\
& = & \Pr(\hat{\sigma}_t Z_t>1\mid \hat{\sigma}^2_t Z^2_t>1)\nonumber\\
& = & \Pr(Z_t>0\mid \hat{\sigma}^2_t Z^2_t>1)\nonumber\\
& = & \int_0^{\infty} \Pr(Z_t>0\mid \hat{\sigma}^2_t Z^2_t>1,  \hat{\sigma}^2_t=s)F_{\hat{\sigma}_t^2}(ds)\nonumber\\
& = & \int_0^{\infty} \Pr(Z_t>0\mid Z^2_t>s^{-1})F_{\hat{\sigma}_t^2}(ds)\nonumber\\
& = & \int_0^{\infty} \Pr(Z_t>s^{-1/2}\mid |Z_t|>s^{-1/2})F_{\hat{\sigma}_t^2}(ds),
\label{eqn:waytofinddelta}
\end{eqnarray}
where $F_{\hat{\sigma}_t^2}$ is the distribution of $\hat{\sigma}^2_t$, the $(q+1)$the component of the limit vector 
$\hat{\mathbf{\Theta}}_t$ which is defined by limit~\eqref{eq.mult.var.garchpq}.

For the $\textrm{St}(\mu, \omega, \xi,\nu)$ innovation distribution then expression~\eqref{eqn:waytofinddelta} becomes
\[
\delta=  \int_0^{\infty} 
\frac{1-F_T((s^{-1/2}-\mu)\xi\sqrt{\nu+1}/\omega;\nu+1)}{1-F_T((s^{-1/2}-\mu)\xi\sqrt{\nu+1}/\omega;\nu+1)+F_T(-(s^{-1/2}+\mu)\xi\sqrt{\nu+1}/\omega;\nu+1)}F_{\hat{\sigma}_t^2}(ds)\\
\]
With a sample of $s_1, \ldots , s_m$ from $\hat{\sigma}_t^2$ derived using Algorithm~\ref{alg:Paul} we can get a Monte Carlo approximation, to any desired accuracy though the choice of $m$,  as follows
\[
\delta \approx
\frac{1}{m}\sum_{i=1}^m \frac{1-F_T((s_i^{-1/2}-\mu)\xi\sqrt{\nu+1}/\omega;\nu+1)}{1-F_T((s_i^{-1/2}-\mu)\xi\sqrt{\nu+1}/\omega;\nu+1)+F_T(-(s_i^{-1/2}+\mu)\xi\sqrt{\nu+1}/\omega;\nu+1)}.
\]

Intuitively it seems as though $\delta$ should be equal to $\delta_Z=\lim_{x\rightarrow z_F}\Pr(Z_t>x\mid |Z_t|>x)$, where $z_F$ is the upper end point of $|Z_t|$. This is not the case though, as knowing that $|X_t|$ is large is quite different from knowing that $|Z_t|$ is large, as the former also can be achieved with large volatility and small innovations. For $\textrm{St}(\mu, \omega, \xi,\nu)$ we have that $\delta_Z$ is given by 
\[
\delta_Z=\frac{1-F_T(\xi\sqrt{\nu+1};\nu+1)}{1-F_T(\xi\sqrt{\nu+1};\nu+1)+F_T(-\xi\sqrt{\nu+1};\nu+1)}.
\]  

Figure~\ref{fig.delta} shows how $\delta=\delta(\xi)$ varies with the skew-$t_3$ distribution parameter $\xi$ for $\xi\ge 0$; the values of $\delta$ for $\xi<0$ follow due to $\delta$ being symmetrical about $0.5$, i.e., for $\xi< 0$  then $\delta(\xi)$ is equal to $1-\delta(|\xi|)$.  The figure shows that for a given level of $\xi$, i.e., skewness in the innovation distribution, as $\phi$ increases there is a diminishing
level of skewness in the tails of the GARCH$(p,q)$ process as measured by $\delta$. When $\phi=0$ then 
$\delta=\delta_Z$, and this value is seen to be an upper bound for $\delta$ in Figure~\ref{fig.delta}. These results are consistent with intuition as the larger the value of $\phi$ the more the process is driven by the past values of the process and the less the new innovations (and their skewness) matter. Table~\ref{tab:key-stationarity} illustrates how $\delta$ changes over models when $\xi=1$, with $\delta$ typically decreasing as $\phi$ increases, as the persistence of volatility is  more important than the innovation structure as $\phi$ increases.

\begin{figure}[H]\psfrag{xi}{\LARGE $\xi$} \psfrag{delta}{\LARGE $\delta$}
	\centering
	\includegraphics[width=0.35\textwidth, angle= 270]{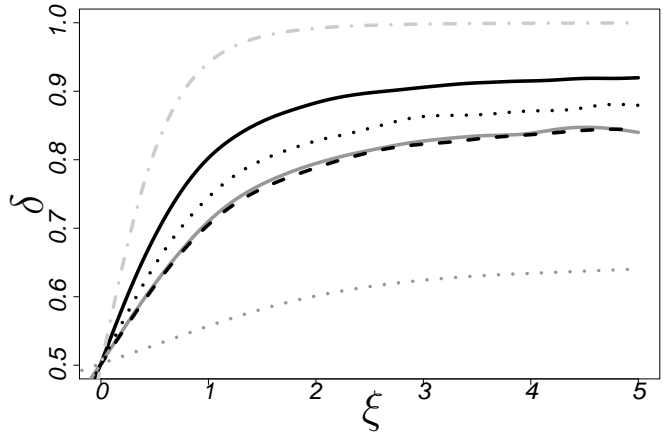}
	\caption{Plot of $\delta=\delta(\xi)$ against $\xi$, for $\xi\ge 0$, for models: A (black continuous), B (black dotted), C (black dashed),  D (grey continuous) and E (grey dotted) for the skew $t$ distribution with $\nu=3$	. Also plotted is  $\delta_Z$ (light grey dashed and dotted). For $\xi< 0$  then $\delta(\xi)$ is equal to $1-\delta(|\xi|)$.	
	Line types are as consistent as possible with Figure~\ref{fig.root.function}.} \label{fig.delta}
\end{figure}

\section{Results for the \garchpq process}
\label{sec:results}

\subsection{Extremogram} \label{sec:extremogram}

Figure~\ref{fig.extremogram.garch22.modAD} gives the extremogram $\chi_{X^2}(\tau)$ for the squared of GARCH process for models A-D with normal and $t$ innovations with $\nu=3$ degrees of freedom. Firstly notice the impact of the innovation distribution on  $\chi_{X^2}(\tau)$. In all cases the heavier tailed innovation distribution leads to weaker extremal dependence at all lags. Models C and D, both IGARCH processes (with $\phi=1$), exhibit much slower decay rates in extremal dependence as lag $\tau$ increases than for models A and B with $\phi<1$, with the level of extremal dependence appearing to be strongly related to  $\phi$. Furthermore, we see for models B and D that $\chi_{X^2}(2)>\chi_{X^2}(1)$, with  $\chi_{X^2}(\tau)$ decaying monotonically for $\tau \ge 2$. The reason for this seems to be that 
$\beta_2>\max(\alpha_1, \alpha_2)$ here. From Section~\ref{sec:tailchain} we have that the evaluation of $\chi_{X^U}(\tau)$ and $\chi_{X^L}(\tau)$ is simple from $\chi_{X^2}(\tau)$ once $\delta$ is known, which we have from Section~\ref{sec:delta} and Table~\ref{tab:key-stationarity}.

An empirical estimate $\tilde{\chi}_{X^2}(\tau, u)$ of the extremogram of $\chi_{X^2}(\tau)$ based on a sample 
of length $n$ from a GARCH($p,q$) process is given by 
\[
\tilde{\chi}_{X^2}(\tau, u)=\frac{\sum_{j=1}^{n-\tau} \mathbf{1}(X^2_j>u, X^2_{j+\tau}>u)}{\sum_{j=1}^{n-\tau} \mathbf{1}(X^2_j>u)},
\]
where $u$ is a threshold. 
Figure~\ref{fig.extremogram.garch22.modAD} shows $\tilde{\chi}_{X^2}(\tau, u)$  for large $n$ and for three threshold choices $u$ corresponding to $0.99, 0.999$ and $0.9999$ quantiles of $X_t^2$. The agreement with 
limit values $\chi_{X^2}(\tau)$ that we have evaluated is very good generally, with the empirical estimates suffering from bias and variance trade-off, as with all threshold based estimates. Model B has the slowest convergence of the empirical estimator, but even here  at the highest threshold there is almost perfect overlap between empirical estimates and the true values for all lags $\tau$. In contrast for model C the highest threshold produces the least good estimate, presumably due to its high variance. This gives strong evidence that our evaluation of $\chi_{X^2}(\tau)$ is accurate, but it also shows how difficult it is to get accurate values from direct simulations due to different convergence rates from apparently rather similar models.

\begin{figure}
	\begin{center}
		\includegraphics[width=0.25\textwidth, angle= 270]{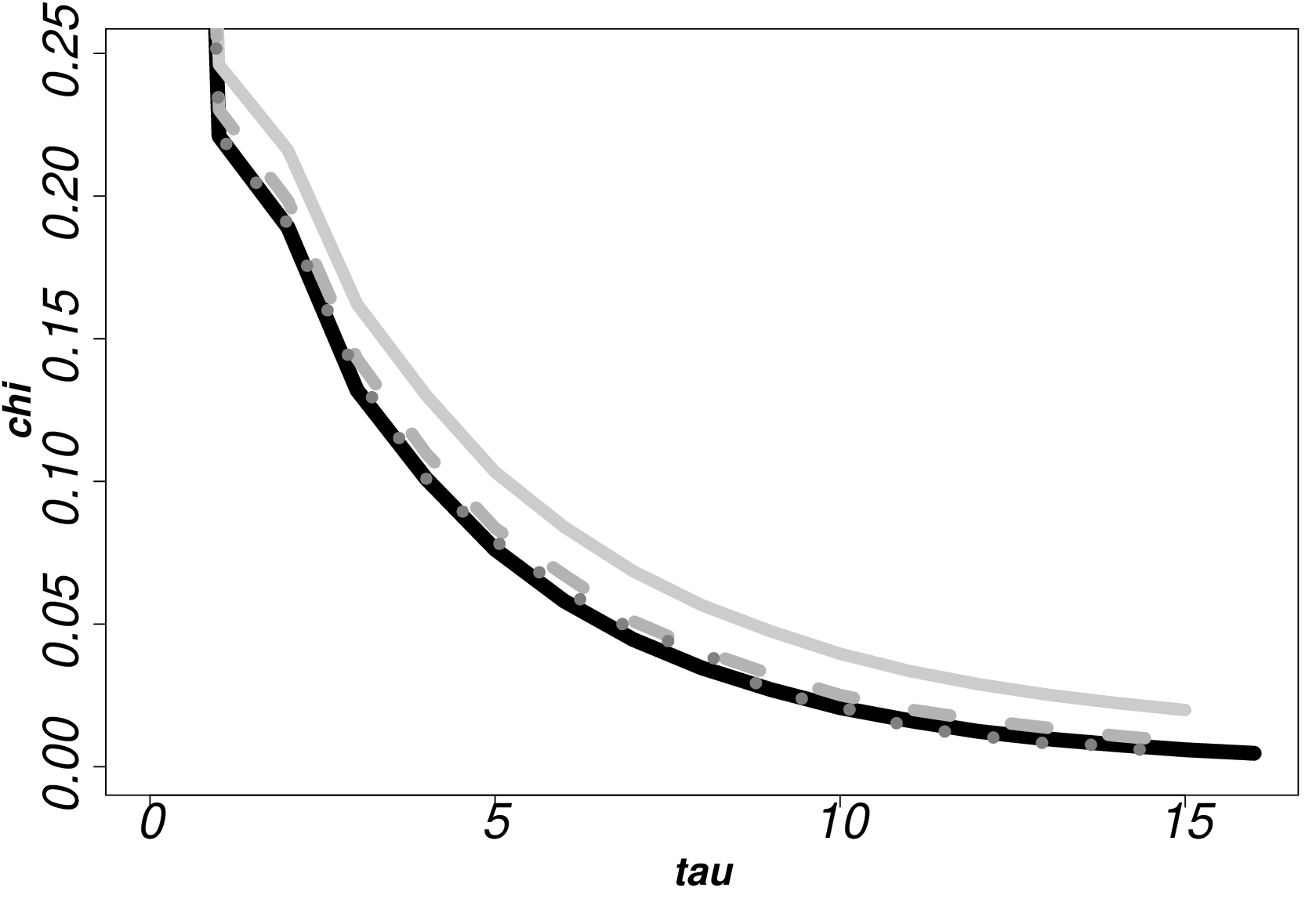}
		\includegraphics[width=0.25\textwidth, angle= 270]{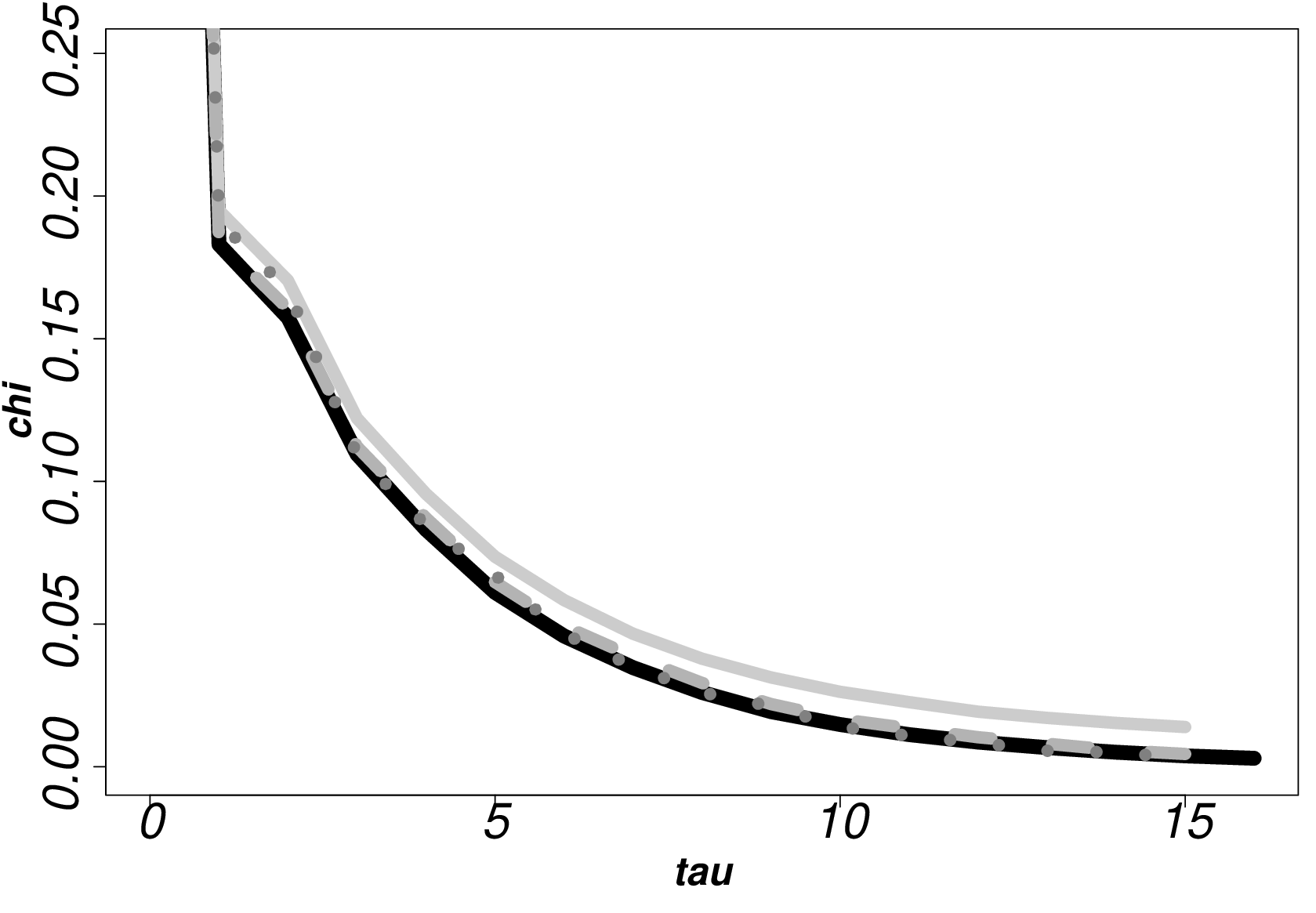}
		\includegraphics[width=0.25\textwidth, angle= 270]{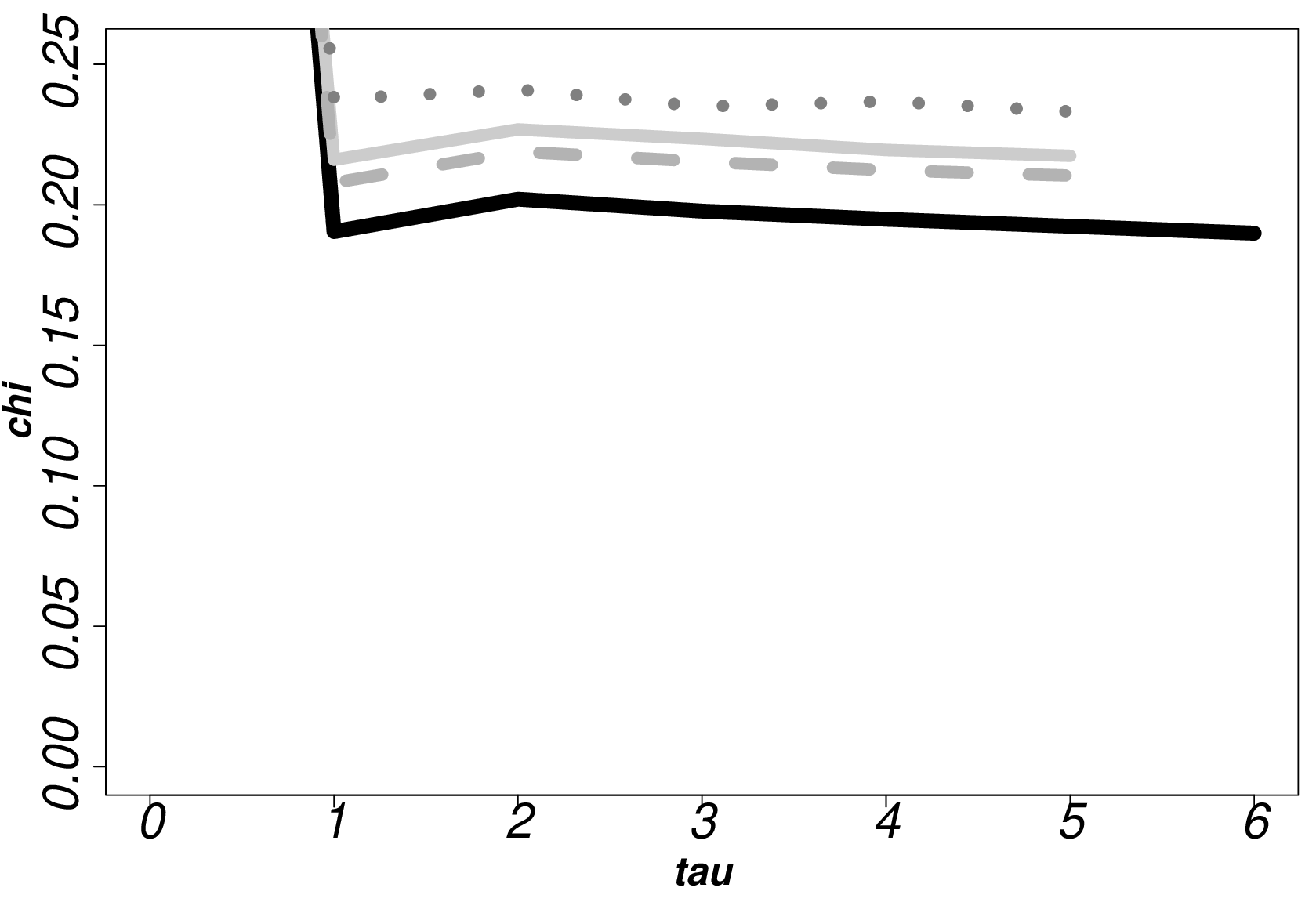}
		\includegraphics[width=0.25\textwidth, angle= 270]{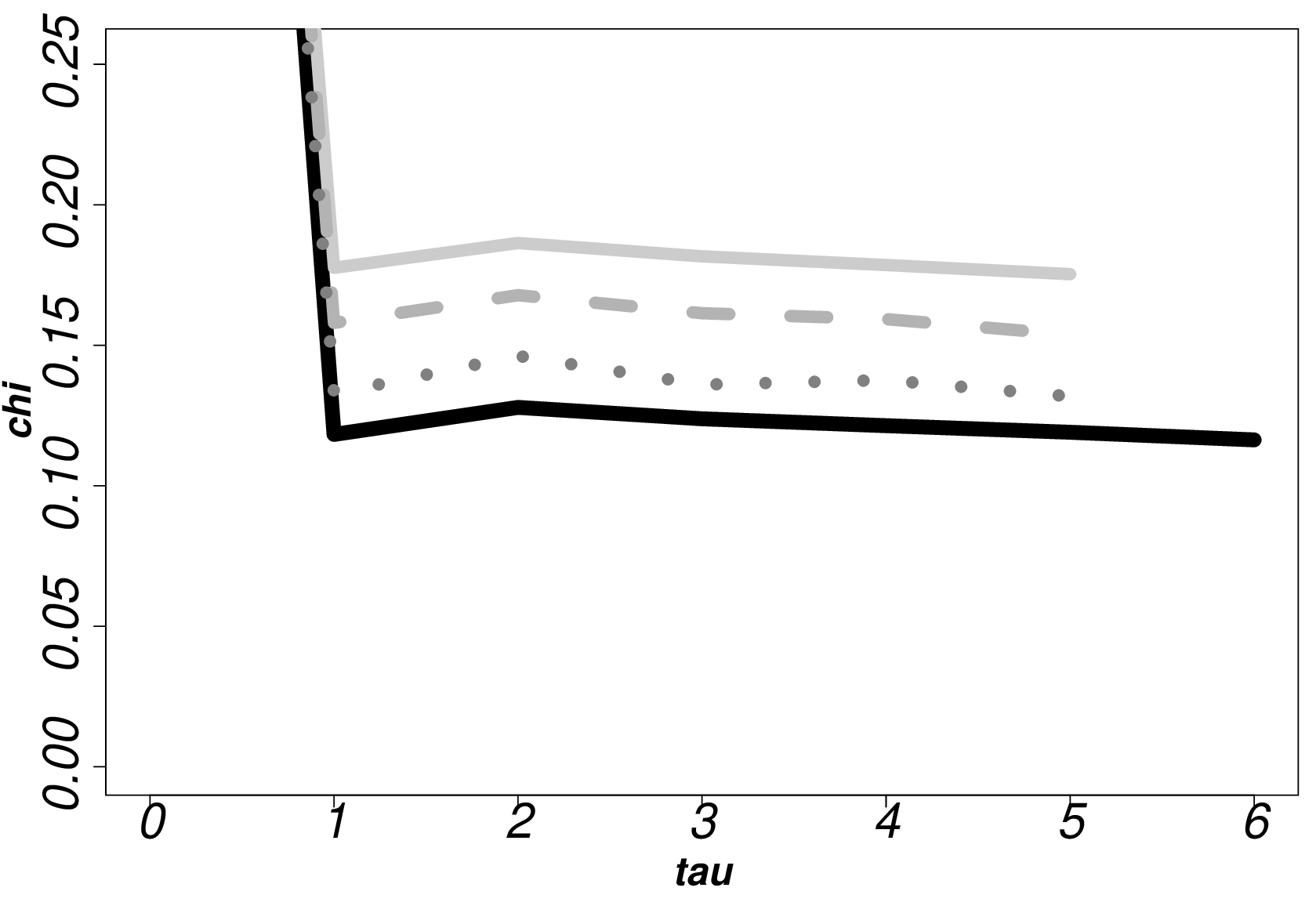}
		\includegraphics[width=0.25\textwidth, angle= 270]{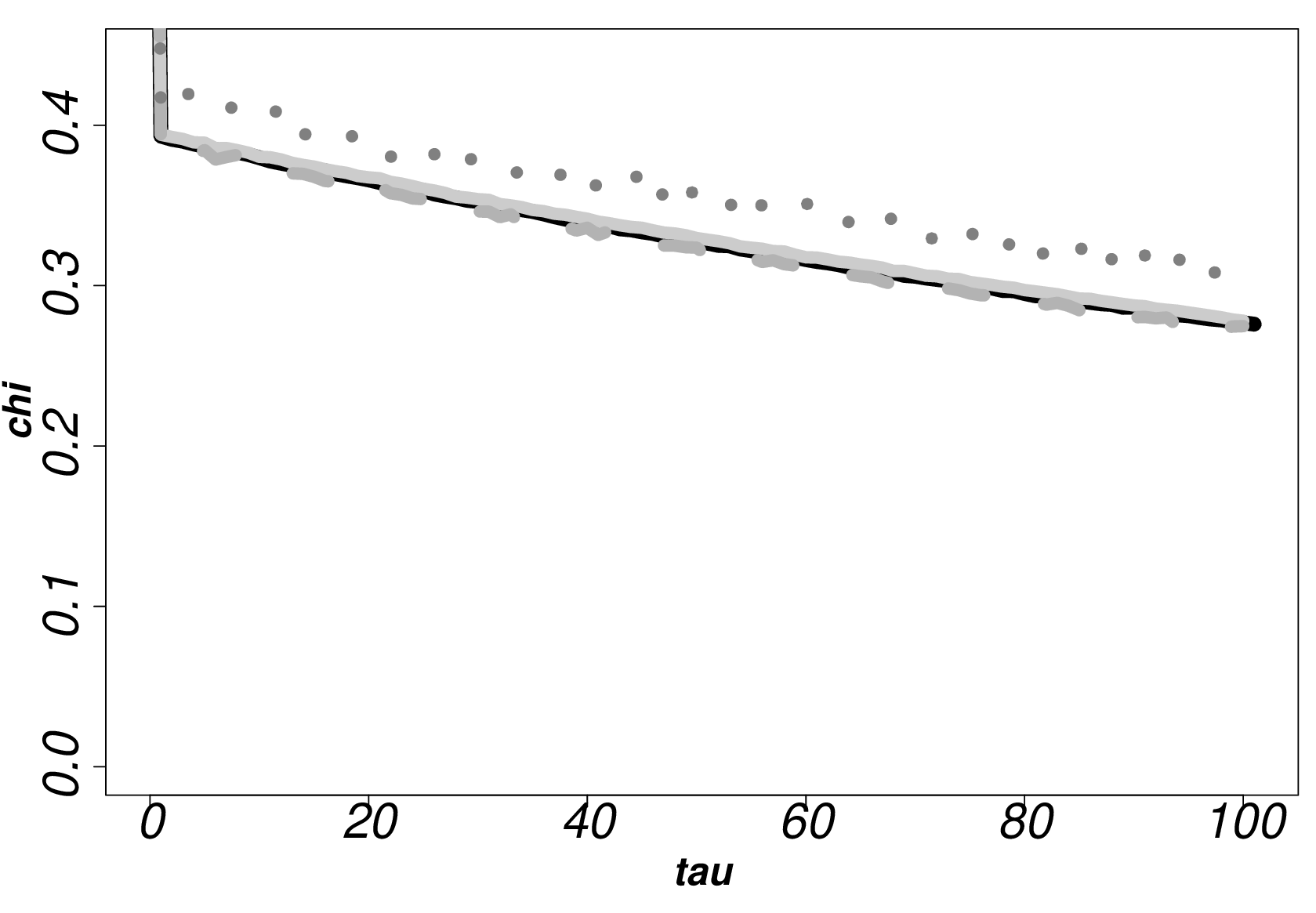}
		\includegraphics[width=0.25\textwidth, angle= 270]{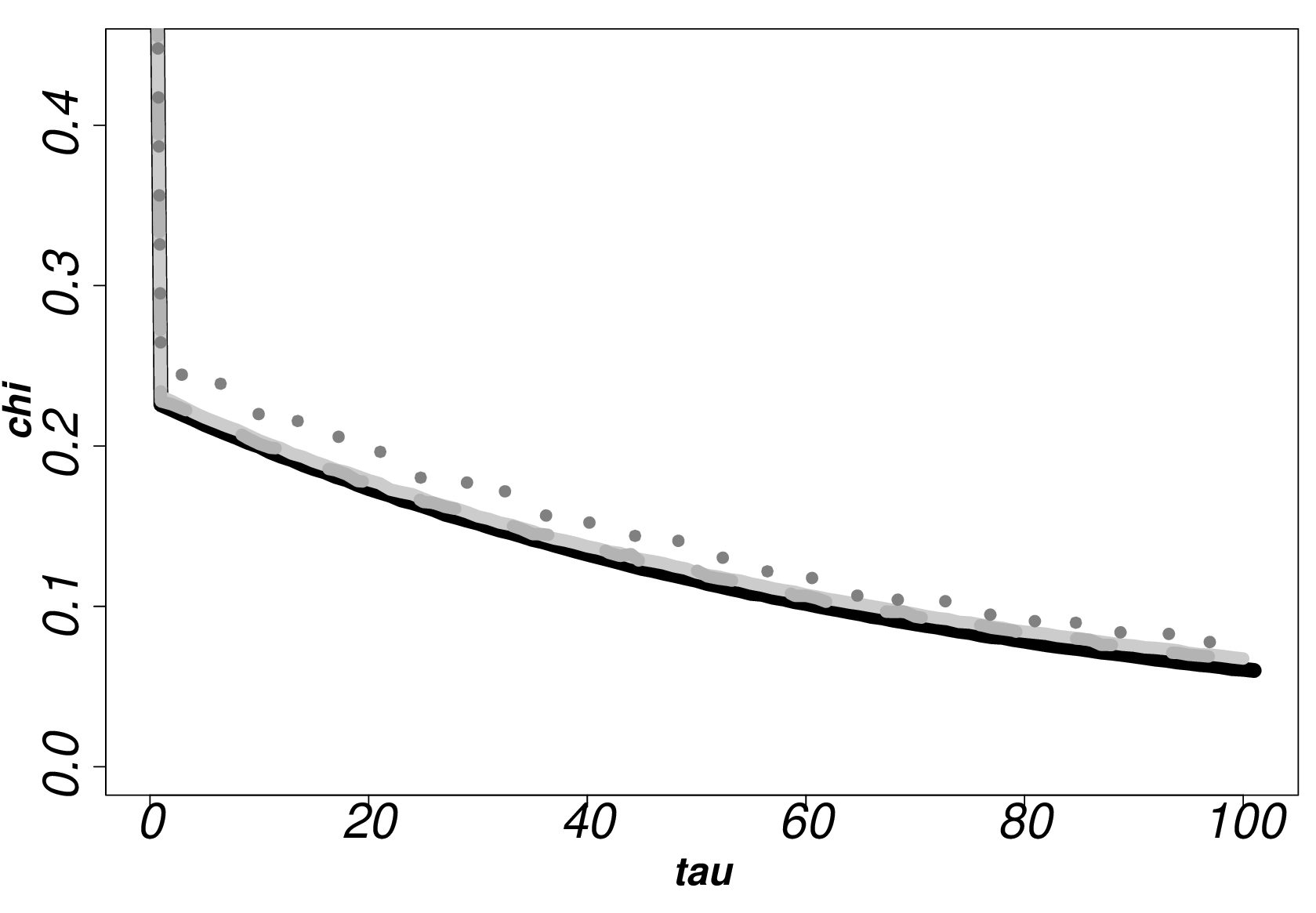}
		\includegraphics[width=0.25\textwidth, angle= 270]{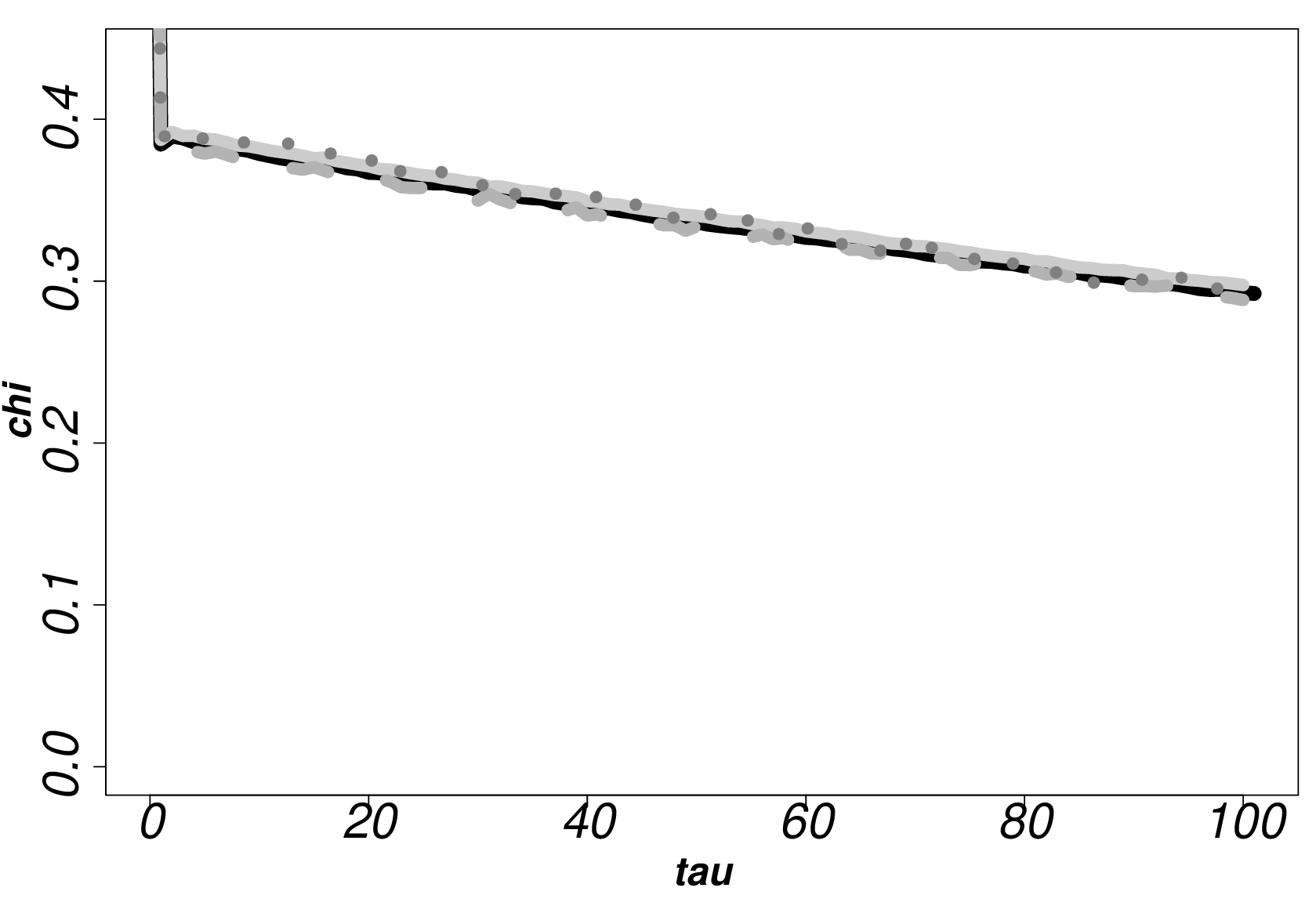}
		\includegraphics[width=0.25\textwidth, angle= 270]{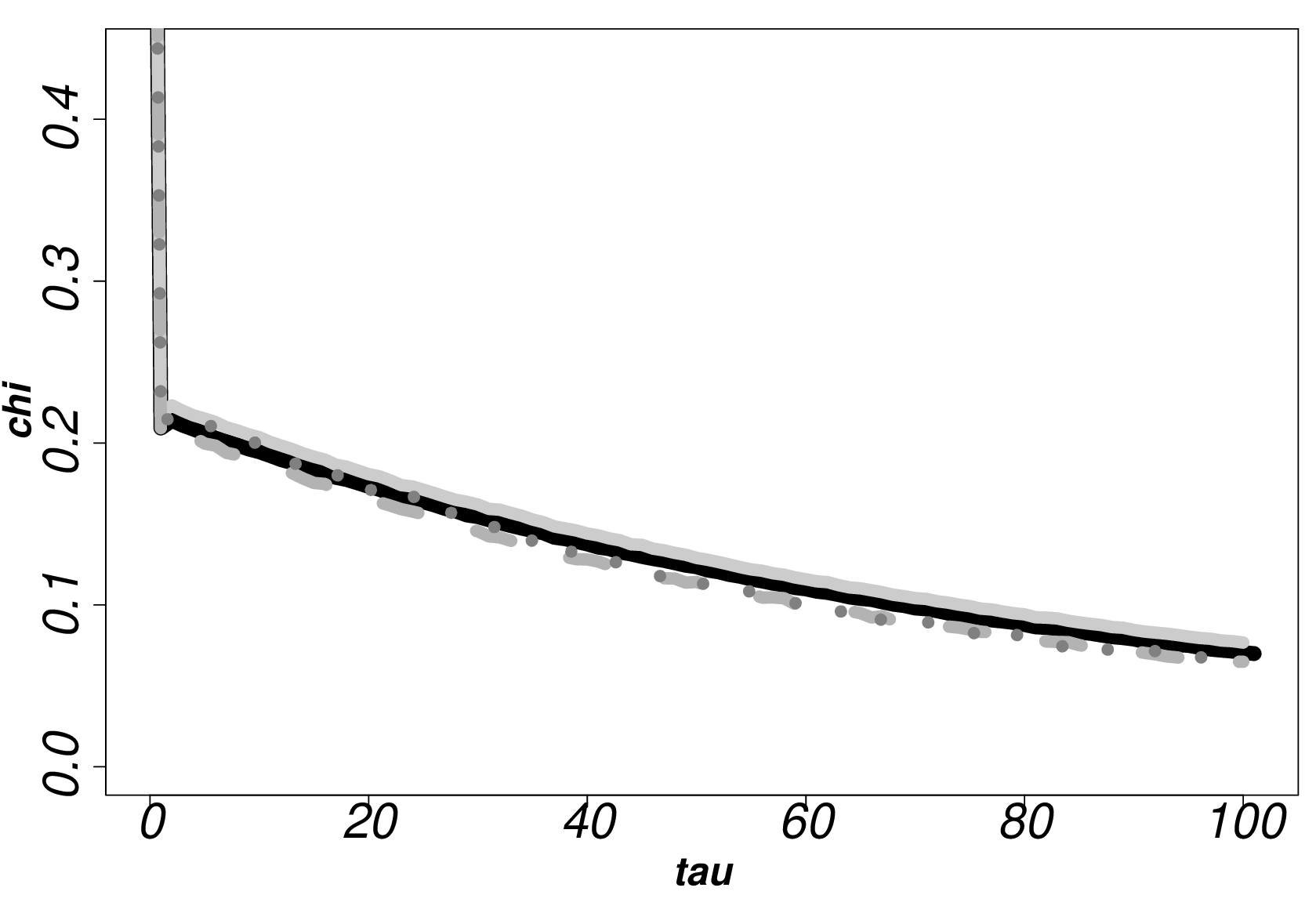}
	\end{center}
	\caption{Extremogram $(\tau, \chi_{\tau})$ for various squared GARCH processes with: $Z_t\sim N(0,1)$ (left panels) and scaled $t(3)$ (right panels) and for models A-D from top to bottom rows respectively.
		Black lines are true limit values and the three grey lines are empirical extremogram estimates $\tilde{\chi}_{\tau}(u)$, based on a sample of size $n=5\times10^7$, at $u$ corresponding to 0.99 (continuous solid light grey), 0.999 (dashed grey) and 0.9999 (dotted dark grey) quantiles of $X_t^2$. } \label{fig.extremogram.garch22.modAD}
\end{figure}

\subsection{Extremal Index}

We finish by looking at how the extremal index $\theta_{X^2}$,  $\theta_{X^U}$ and $\theta_{X^L}$,
of the upper tail of the series $\{X^2_t\}, \{X_t\}$ and $\{-X_t\}$, change over GARCH($p,q$) processes. 
Results for each of these characteristics are given for models A-E and three innovation distributions  are given in Table~\ref{tab:key-stationarity}. Firstly consider the effect of the innovation distribution for a given model on the extremal index of the $\{X^2_t\}$ process, $\theta_{X^2}$ through the associated average size of clusters extremes values, i.e., $1/\theta_{X^2}$. With shorter tailed innovations clusters last longer on average and introducing skewness further reduces the mean cluster size. For increasing $\phi$, for $0<\phi\le 1$, we have increasing average cluster sizes, but that pattern does not follow when $\phi>1$. In all cases, $\min(\theta_{X^U},\theta_{X^U})\ge \theta_{X^2}$,  indicating the extremes of the processes $\{X_t\}$ and $\{-X_t\}$ exhibit less clustering on average than the $\{X^2_t\}$ process. 
We have equality, in this inequality, only  when $\delta=0$ or $1$, and find that as $\delta$ tends to these limits one or other of the processes $\{X_t\}$ and $\{-X_t\}$ has similar cluster of extremes events as the $\{X^2_t\}$ process. For $\xi>0$ more clustering occurs in the upper tail than in the lower tail of the GARCH$(p,q)$ process, with the reverse happening with $\xi<0$.

To give a better idea of how the GARCH($p,q$) parameters affect the extremal index Figure~\ref{fig.ei.garch22} presents a contour plot of $\theta_{X^U}$ for GARCH(2,2) process over different key parameters $(\alpha_1, \alpha_2,\beta_1,\beta_2)$. In particular, in each panel we 
hold fixed two parameters and contour over the other two key parameters of a GARCH(2,2) process. 
Figure~\ref{fig.ei.garch22}, shows that broadly $\theta_{X^U}$  decreases, i.e., average cluster sizes increase,  with increasing $\phi$, up to the boundary case of an IGARCH(2,2) model. Consequently contours are near linear in the parameters. Unless $\max(\beta_1, \beta_2)$ is large then small values of $\alpha_1$ and $\alpha_2$ tend to lead to very limited clustering.  This seems logical given the GARCH formulation, as small $\alpha_1$ and $\alpha_2$ mean that the effect of the large $X_t^2$ value can have limited impact of the subsequent volatilities, so without $\beta_1$ and $\beta_2$ being large, to pick up the momentum of the evolution of the event, the large event is very likely to die out rapidly. 
We also have that we obtain stronger extremal dependence with larger values of the pair $(\alpha_2, \beta_2)$ than for 
$(\alpha_1,\beta_1)$ for equal values of $\phi$, as seen by values of $\theta_{X^U}$ being smaller in  Figure~\ref{fig.ei.garch22} right panel by comparison to the left panel. The reason for this is the stronger effect of $\beta$ coefficients compared to that of the $\alpha$ as the persistency in the volatility is translated into a stronger extremal dependence.

\begin{figure}
\begin{center}
	\includegraphics[width=0.25\textwidth, angle= 270]{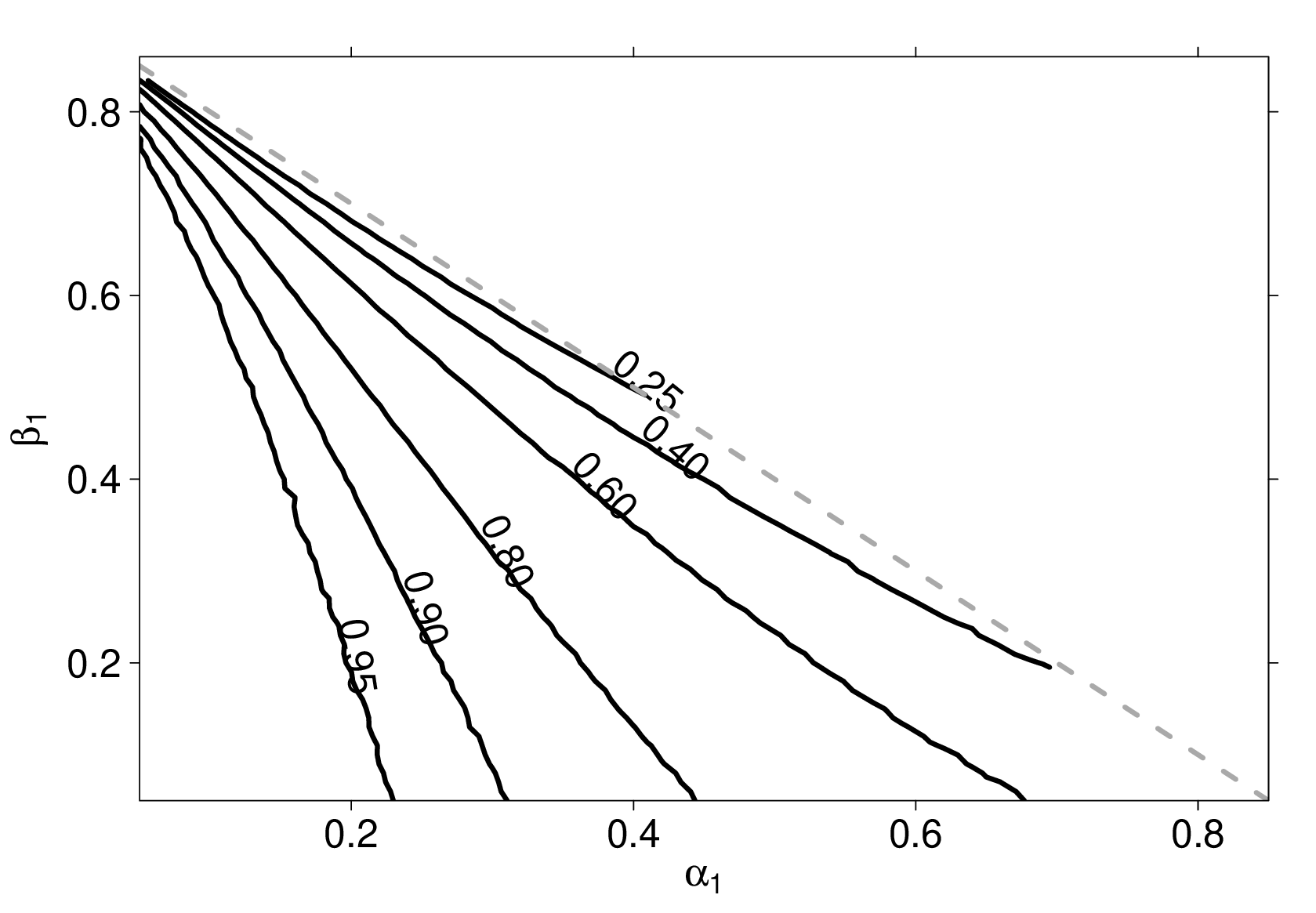}
        \includegraphics[width = 0.25\textwidth, angle= 270]{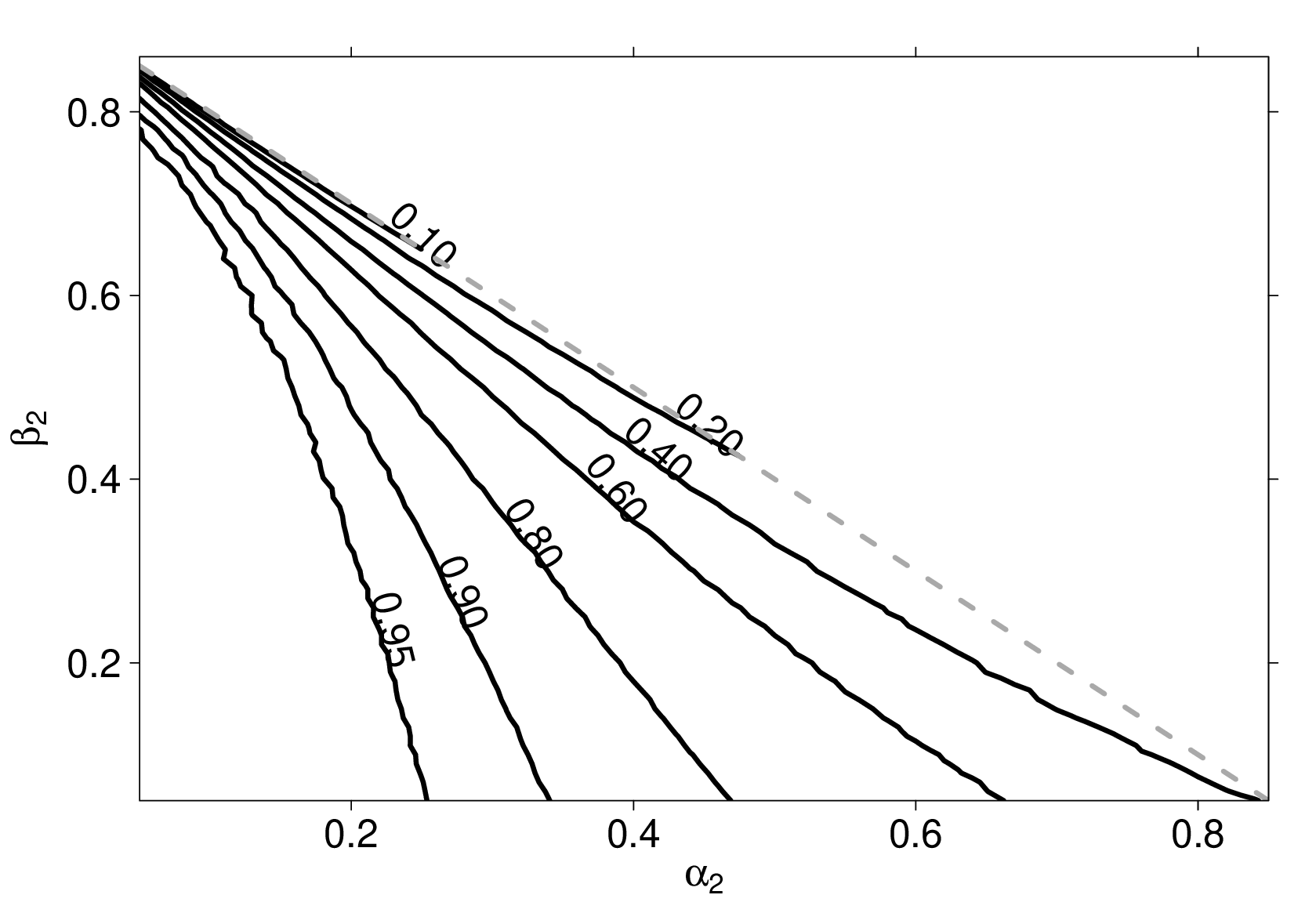}
		\end{center}
	\caption{Contour plots for the extremal index $\theta_{X^U}$ for the GARCH($2,2$) process: 
	left, as function of $(\alpha_1,\beta_1)$ with $\alpha_2 = \beta_2 = 0.05$;
	right, as function of $(\alpha_2,\beta_2)$ with $\alpha_1 = \beta_1 = 0.05$. 
	In both panels the innovation $Z_t$ is standard normal and the grey dashed line is the boundary of the IGARCH($2,2$). } \label{fig.ei.garch22}. 
\end{figure}

\section{Discussion}
\label{sec:discuss}

The new theory and methods we present extend to assessing the strict stationarity and extremal properties for a much broader class of stochastic recurrence equations. Specifically, for any process
\[
\mathbf{Y}_t = \mathbf{A}_t \mathbf{Y}_{t-1} + \mathbf{B}_t, \hspace{5mm} t \in \mathbb{Z},
\]
where $\mathbf{A}_t$ and  $\mathbf{B}_t$  
are stochastic, independent and identically distributed, sequences of matrices and vectors respectively, satisfying the conditions of \cite{ke:73}, then our paper gives methods which: help assess stationarity of the process; determine its tail index of regular variation for the marginals distributions of $\mathbf{Y}_t$; determine ways to simulate from the spectral measure from the limiting joint distribution of $\mathbf{Y}_t$ given some norm, $\norm{\mathbf{Y}_t}$, tends to infinity; and  are able to generate a wide range of properties, such as the extremal index of any marginal of the process $\mathbf{Y}_t$ from its tail chain. Previously none of these properties could be derived due to a combination of  a lack of appropriate numerically stable algorithms. The methods presented here overcome these limitations and provide a broad toolbox of numerically robust approaches to derive the extremal analysis of a wide class of stochastic recurrence equations including all GARCH($p,q$) processes with bounded and unbounded innovation variables.

\begin{appendices}
\section{ Proof of Theorems}\label{sec:app.proofs}

\begin{proof}[Proof of Theorem~\ref{th.gamma}]
	First rewrite the product of independent matrices as
	\begin{equation} \label{eq:prodAscaled}
	\prod_{i=1}^{t}\mathbf{A}_{t+1-i} = \prod_{i=1}^{t}\left(\frac{\mathbf{A}_{t+1-i}}{\lambda_{t+1-i}\exp(\eta)}\right)  \left(\prod_{i=1}^{t}\lambda_i\right) \exp(\eta t)=\mathbf{C}_t \left(\prod_{i=1}^{t}\lambda_i\right)\exp(\eta t).
	\end{equation}
	Thus
	\[
	\norm{\prod_{i=1}^{t}\mathbf{A}_{t+1-i} }  =  \norm{\mathbf{C}_t \left(\prod_{i=1}^{t}\lambda_i \right)\exp(\eta t)}
	= \norm{\mathbf{C}_t}\left(\prod_{i=1}^{t}\lambda_i \right) \exp(\eta t)
	\]
	so as $t \rightarrow \infty$\begin{eqnarray*}
		\frac{1}{t}\ln \norm{\prod_{i=1}^{t}\mathbf{A}_{t+1-i} } & = & \ln(\norm{\mathbf{C}_t})/t+\sum_{i=1}^{t}\ln \lambda_i /t +\eta\\
		& \rightarrow & E(\ln \lambda)+\eta,
	\end{eqnarray*}
	hence $\gamma=E(\ln \lambda)+\eta$.
\end{proof}

\begin{proof} [Proof of Theorem~\ref{th.kappa}]

With the same notation as for Theorem~\ref{th.gamma} we have that
\begin{eqnarray}
\norm{\mathbf{A}_t \mathbf{A}_{t-1} \cdots \mathbf{A}_1}^{\kappa}
&=&  \norm{\mathbf{C}_t\left(\prod_{i=1}^{t}\lambda_i\right)\exp(\eta t) }^{\kappa}\nonumber\\
&=&  \norm{\mathbf{C}_t}^{\kappa}\left(\prod_{i=1}^{t}\lambda_i^{\kappa}\right)\exp( \eta \kappa t).
\label{eqn:prod}
\end{eqnarray}
Under condition~\eqref{eqn:Ccondition} we have that there exists a sequence $s_t>0$ such that $s_t\rightarrow 0$ and $s_t t\rightarrow \infty$ as $t\rightarrow \infty$
\[
\exp(-s_t t)<\norm{\mathbf{C}_t}<\exp(s_t t)
\]
for all $t$. Combining these inequalities with expression~\eqref{eqn:prod} we obtain that 
\[
E \left(\exp(-s_t t)\left(\prod_{i=1}^t\lambda_i ^{\kappa}\right)\right) \exp( \eta\kappa t))  < 
E\left(\norm{\mathbf{A}_t \mathbf{A}_{t-1} \cdots \mathbf{A}_1}^{\kappa}\right)  < 
E \left(\exp(s_t t)\left(\prod_{i=1}^t\lambda_i ^{\kappa}\right)\right) \exp( \eta\kappa t)).
\]
Focusing first on the upper bound we have that 
\begin{eqnarray*}
E \left(\exp(s_t t)\prod_{i=1}^t\lambda_i ^{\kappa}\right) \exp( \eta\kappa t) 
&=&   \exp(s_t t) E \left(\prod_{i=1}^t\lambda_i ^{\kappa}\right) \exp( \eta\kappa t)\\
&=&   \exp(s_t t)\{E\left( \lambda^{\kappa}\right)\}^t\exp( \eta\kappa t)\\
&=&  \exp(s_t t) \{E[(\lambda \exp(\eta))^{\kappa}]\}^t,
\end{eqnarray*}
where the second step comes from $\lambda_i$, $i=1, \ldots, t$ being independent and identically distributed. 
Hence, as $t\rightarrow \infty$
\begin{eqnarray*}
\frac{1}{t} \ln \left\{E \left(\exp(s_t t)\left(\prod_{i=1}^t\lambda_i ^{\kappa}\right)\right) \exp( \eta\kappa t)\right\} 
& = & s_t  + \ln \{E[(\lambda \exp(\eta))^{\kappa}]\}\\
& \rightarrow & \ln \{E[(\lambda \exp(\eta))^{\kappa}]\}.
\end{eqnarray*}
By an identical argument, the lower bound is found to be asymptotically equal to the upper bound. Hence 
\[
\lim_{t\rightarrow \infty}\frac{1}{t} \ln E \left(\norm{\mathbf{A}_t \mathbf{A}_{t-1} \cdots \mathbf{A}_1}^{\kappa}\right)=  \ln\{E\left( (\lambda\exp(\eta))^{\kappa}\right)\},
\]
so this limit is equal to 0, as required by condition~\eqref{eq:root.equation}, only when $\kappa$ satisfies $E[(\lambda \exp(\eta))^{\kappa}]=1$.

\end{proof}

\begin{proof} [Proof of Theorem~\ref{th.IGARCHkappa}]
First suppose that $\kappa=1$, then from property~\eqref{eqn:moment} we have that 
$E(\norm{\mathbf{A}\hat{\mathbf{\Theta}}_t}) = 1$, where $\mathbf{A}$ is independent of $\hat{\mathbf{\Theta}}_t$. Then, by considering the vector $\mathbf{Y}_t$, we see that marginally the first $q$ and the last $p$ components of $\hat{\mathbf{\Theta}}_t$ each have identical marginal distributions. Furthermore, as $E(Z_t^2)=1$, it follows that all $E(\hat{\vartheta}_t^{(i)})=1/(p+q)$ for all $i=1, \ldots , p+q$, where recall $\hat{\mathbf{\Theta}}_t=(\hat{\vartheta}_t^{(1)}, \ldots ,\hat{\vartheta}_t^{(p+q)})$. We also have that 
\[
E(\mathbf{A}) = \begin{pmatrix} \alpha^{(q-1)}&\alpha_q &\beta^{(p-1)} &
\beta_p  \\ \boldsymbol{I}_{q-1} &0_{q-1}& \boldsymbol{0}_{(q-1)\times (p-1)} & 0_{q-1} \\ \alpha^{(q-1)} &\alpha_q &\beta^{(p-1)} &
\beta_p \\
\boldsymbol{0}_{(p-1)\times(q-1)} &0_{p-1} & \boldsymbol{I}_{p-1} & 0_{p-1}\end{pmatrix}.
\]
Let $a_{i,j}$ denote the $(i,j)$th element of $\mathbf{A}$, using the independence of $a_{i,j}$ from $\hat{\vartheta}_t^{(j)}$  then 
\begin{eqnarray*}
E(\norm{\mathbf{A}\hat{\mathbf{\Theta}}_t})  & = & E\left(\sum_{i=1}^{p+q}\sum_{j=1}^{p+q} a_{i,j} \hat{\vartheta}_t^{(j)}\right)\\
& = & \sum_{i=1}^{p+q}\sum_{j=1}^{p+q} E(a_{i,j}) E(\hat{\vartheta}_t^{(j)})\\
& = & \frac{1}{p+q}\sum_{i=1}^{p+q}\sum_{j=1}^{p+q} E(a_{i,j}) \\
& = & \frac{1}{p+q}\left(2\sum_{i=1}^{q} \alpha_i+2\sum_{i=1}^{p} \beta_i+q-1+p-1\right) \\
& = & 1+\frac{2}{p+q}\left(\sum_{i=1}^{q} \alpha_i+\sum_{i=1}^{p}\beta_i-1\right).
\end{eqnarray*}
Thus $E(\norm{\mathbf{A}\hat{\mathbf{\Theta}}_t})=1$ only when $\sum_{i=1}^{q} \alpha_i+\sum_{i=1}^{p}\beta_i=1$, i.e., when the process is IGARCH($p,q$). The argument is simply reversed giving that $\kappa=1$ for any IGARCH($p,q$) process.
\end{proof}
\end{appendices}

\section*{Acknowledgements}
We would like to thank Feridun Turkman for encouraging and helpful discussions.

\bibliographystyle{jrss}
\bibliography{reduced}
\end{document}